\documentclass[conference]{IEEEtran}
\IEEEoverridecommandlockouts
\usepackage{cite}
\usepackage{amsmath,amssymb,amsfonts}
\usepackage{graphicx}
\usepackage{textcomp}
\usepackage{xcolor}

\usepackage{amsthm}
\usepackage{mathtools}
\mathtoolsset{showonlyrefs=true}
\usepackage{stmaryrd}
\usepackage{mathpartir}
\usepackage{tikz}
\usetikzlibrary{tikzmark}
\usepackage{subcaption}
\usepackage{algpseudocodex}
\algrenewcommand\algorithmicfunction{}

\usepackage{ifthen}
\newif\iflong
\longtrue

\iflong
\usepackage[bibliography=common]{apxproof}
\else
\usepackage[bibliography=common,appendix=strip]{apxproof}
\fi

\usepackage[verbose]{hyperref}

\def\BibTeX{{\rm B\kern-.05em{\sc i\kern-.025em b}\kern-.08em
    T\kern-.1667em\lower.7ex\hbox{E}\kern-.125emX}}

\theoremstyle{definition}
\newtheorem{definition}{Definition}[section]
\newtheorem{example}[definition]{Example}
\newtheorem{remark}[definition]{Remark}

\theoremstyle{plain}
\newtheorem{lemma}[definition]{Lemma}
\newtheorem{theorem}[definition]{Theorem}

\newtheorem{proposition}[definition]{Proposition}
\newtheorem{corollary}[definition]{Corollary}

\usepackage[warnconjecture]{macros.aux}

\newcommand{\llbrace}{\mathopen{\{\mkern-6mu\{}}
\newcommand{\rrbrace}{\mathclose{\}\mkern-6mu\}}}

\newcommand{\Conditioning}{\mathop{\mathbf{C}}}
\newcommand{\Own}{\mathop{\mathbf{own}}}

\newcommand{\Persistent}{\mathop{\square}}
\newcommand{\Uniform}[1]{\mathop{\mathbf{unif}} #1}
\newcommand{\AlmostSurelyEqual}{\stackrel{\mathrm{as}}{=}}
\newcommand{\sepimp}{\mathrel{-\mkern-6mu*}}
\newcommand{\bigast}{\mathop{\scalebox{2}{\raisebox{-0.2ex}{$\ast$}}}}
\newcommand{\LilacEntails}{\vdash}
\newcommand{\LilacIff}{\dashv\vdash}
\newcommand{\HoareTriple}[3]{\{ #1 \}\ #2\ \{ #3 \}}
\newcommand{\HoareTripleOwn}[3]{\llbrace #1 \rrbrace\ #2\ \llbrace #3 \rrbrace}
\newcommand{\ifexpr}[3]{\mathtt{if}\ #1\ \mathtt{then}\ #2\ \mathtt{else}\ #3}
\newcommand{\forloop}[6]{\mathtt{for}\ (#4 \gets #5;\ #1 \gets #2 \dots #3)\ \{ #6 \}}
\newcommand{\ret}[1]{\mathtt{ret}\ #1}
\newcommand{\interpret}[1]{\llbracket #1 \rrbracket}
\newcommand{\RV}[1]{\mathop{\mathbf{RV}} #1}
\newcommand{\RangeOneTo}[1]{[#1]}
\newcommand{\Expectation}[3]{\mathbb{E}_{#1 \sim #2}\big[#3\big]}
\newcommand{\IversonBracket}[1]{[#1]}
\newcommand{\Giry}{\mathcal{G}}
\newcommand{\KRM}{\mathcal{M}_{\mathrm{disint}}}

\newcommand{\ExposingTranslation}[1]{#1^{\sharp}}
\renewcommand{\vec}[1]{\mathbf{#1}}

\newcommand{\emptyctx}{{\cdot}}

\newcommand{\Prob}[1]{\mathrm{Pr}[#1]}

\renewcommand{\appendix}{\appendices}

\begin{document}

\title{Formal Verification of Probing Security via Conditional Independence\\
\thanks{This work is supported by
JSPS Grant-in-Aid for Scientific Research(C) JP22K11912,
JST K Program Grant Number JPMJKP24U2,
MEXT Quantum Leap Flagship Program (MEXT Q-LEAP) JPMXS0120319794,
and JSPS KAKENHI Grant Number JP25H00446, JP25K21183.}
}

\author{\IEEEauthorblockN{Satoshi Kura}
\IEEEauthorblockA{\textit{Faculty of Education and Integrated Arts and Sciences} \\
\textit{Waseda University}\\
Tokyo, Japan \\
ORCID: 0000-0002-3954-8255}
\and
\IEEEauthorblockN{Katsuyuki Takashima}
\IEEEauthorblockA{\textit{Faculty of Education and Integrated Arts and Sciences} \\
\textit{Waseda University}\\
Tokyo, Japan \\
ORCID: 0000-0001-5216-2229}
}

\maketitle

\begin{abstract}
    Side-channel attacks are a major threat to the security of cryptosystems.
    Masking is a widely used countermeasure against such attacks, but proving the security of masked algorithms is error-prone without formal verification.
    In this work, we propose a novel approach to formal verification of noninterference properties of masked algorithms based on probabilistic separation logic.
    By establishing a connection between noninterference and conditional independence, we show how noninterference can be verified using Lilac, a separation logic for conditional independence.
    We also provide several proof rules that facilitate the verification of probing security and demonstrate their application to example algorithms.
\end{abstract}


\section{Introduction}

\subsection{Formal Verification of Probing Security}

Information leakage through side-channel attacks is one of the major threats to the security of cryptographic schemes and digital signatures.
Among the widely used countermeasures against side-channel attacks is \emph{masking} \cite{IshaiCRYPTO2003}.
Masking protects secrets by splitting a secret value into several shares via secret sharing and performing computations on those shares, so that the leakage of any subset of intermediate values does not reveal the underlying secret.
Applications of masking include the digital signature scheme Raccoon \cite{DelPino2023}, which was submitted to the NIST Post-Quantum Cryptography Standardization process, as well as Masked ML-DSA \cite{CoronTCHES2023,CoronTCHES2024,HermelinkCRYPTO2025}, a masked variant of ML-DSA \cite{NISTFIPS204}, among many others.

The notion of \emph{noninterference} is commonly used to formalize the security of masked algorithms in the ISW probing model~\cite{IshaiCRYPTO2003}.
However, proving noninterference by hand is error-prone, and several flaws have been identified in manual security proofs for masked algorithms~\cite{PrestPKC2023,CoronFSE2014}.
To avoid such errors, \emph{formal verification} of probing security has been studied in recent years~\cite{BartheCCS2016,BartheEUROCRYPT2015,BartheESORICS2019,ZhangCAV2018,GaoTOSEM2021,GaoTACAS2019,GaoTOSEM2024,EldibTOSEM2014,CoronACNS2018,BloemEUROCRYPT2018,KnichelASIACRYPT2020,MeunierIIEEETransSoftwareEng2023}.

\subsection{Formal Verification via Separation Logic}
\emph{Probabilistic separation logic}~\cite{BarthePOPL2020} is a recently developed formal verification method for reasoning about probabilistic programs, including cryptographic algorithms.
This is a variant of \emph{separation logic}~\cite{ReynoldsLICS2002}, which is a Hoare-style logic that can reason about separation of resources using the \emph{separating conjunction} $\phi * \psi$.
Together with the \emph{frame rule}, the separating conjunction enables \emph{local reasoning}, which is particularly useful for modular verification of large programs.

Separation logic was originally introduced to reason about mutable data structures in imperative programs \cite{ReynoldsLICS2002}, but it has since been applied in a variety of contexts beyond memory separation.
Probabilistic separation logic interprets the separating conjunction as \emph{probabilistic independence} between random variables and has been applied to verify input independence properties of cryptographic algorithms, such as the one-time pad and private information retrieval.
Several extensions of probabilistic separation logic have been proposed~\cite{BaoLICS2021,HoPOPL2026,BaoPOPL2025,ZilbersteinPOPL2026,DalLagoCSF2024,BaoPOPL2022}, including Lilac~\cite{LiPLDI2023}, which supports reasoning about \emph{conditional independence}.

Masking can be viewed as a generalization of the one-time pad, which strongly suggests that probabilistic separation logic is a natural framework for verifying probing security.
However, to the best of our knowledge, none of these logics has been used to verify probing security of masked algorithms.

\begin{figure*}
	\centering
	\begin{minipage}{0.8\textwidth}
		\centering
		\begin{tikzpicture}
			\node[align=center, draw] at (0, 0) (cryptoalg) {cryptographic \\ algorithm};
			\node at (-4, 0.3) (secret) {secret input $A$};
			\node at (-4, -0.3) (nonsecret) {public input $B$};
			\node at (4, 0) (output) {output $C$ observed by adversary};
			\draw[->] (secret) -- (cryptoalg);
			\draw[->] (nonsecret) -- (cryptoalg);
			\draw[->] (cryptoalg) -- (output);
		\end{tikzpicture}
		\subcaption{A cryptographic algorithm with input and output}
		\label{subfig:crypto-algorithm}
	\end{minipage}

	\vspace{1em}
	\begin{minipage}[b]{0.45\textwidth}
		\centering
		\begin{tikzpicture}
			\node[align=center, draw] at (0, 0) (sim) {simulator};
			\node at (-3, 0) (nonsecret) {public input $B$};
			\node at (3, 0) (output) {observed output $C$};
			\draw[->] (nonsecret) -- (cryptoalg);
			\draw[->] (cryptoalg) -- (output);
		\end{tikzpicture}
		\subcaption{Simulator-based formulation}
		\label{subfig:simulator-based-view}
	\end{minipage}
	\begin{minipage}[b]{0.45\textwidth}
		\[ A \perp C \mid B \]
		\subcaption{Conditional-independence-based formulation}
		\label{subfig:conditional-independence-view}
	\end{minipage}
	\caption{Two formulations of security: simulator-based formulation and conditional-independence-based formulation.}
	\label{fig:two-views-of-security}
\end{figure*}

\subsection{Connection Between Simulator-Based Security Notions and Conditional Independence}
In this work, we establish a connection between simulator-based definitions of probing security and conditional independence, enabling the use of probabilistic separation logic for formal verification of probing security.

In cryptography, many security notions, including non-interference, are defined via simulators.
To understand the connection between such simulator-based definitions and conditional independence, consider a cryptographic algorithm (Fig.~\ref{subfig:crypto-algorithm}) and an adversary who attempts to learn some secret information from the input by observing information leaked during the execution of the algorithm.
In this situation, three kinds of information are involved: (i) secret information $A$ in the input, (ii) public (non-secret) information $B$ in the input, which is not directly observed by the adversary, and (iii) information $C$ observed by the adversary.
To formalize the requirement that the adversary cannot learn $A$ from $C$, simulator-based security definitions postulate the existence of a simulator that can compute $C$ solely from $B$, without access to $A$ (Fig.~\ref{subfig:simulator-based-view}).
If such a simulator exists, then the adversary cannot distinguish the real execution from the simulated one.
Consequently, although the adversary may obtain some information about $B$, it is impossible to learn $A$.

We rephrase this simulator-based definition in terms of \emph{conditional independence} (Fig.~\ref{subfig:conditional-independence-view}).
Suppose that $A$, $B$, and $C$ are random variables.
Given $A$ and $B$, the probability distribution of $C$ generated by the cryptographic algorithm (Fig.~\ref{subfig:crypto-algorithm}) corresponds to the conditional probability distribution $\Prob{C | A,B}$.
On the other hand, given only $B$, the probability distribution of $C$ generated by the simulator (Fig.~\ref{subfig:simulator-based-view}) corresponds to $\Prob{C | B}$.
The simulator-based definition requires that these two probability distributions are equal: $\Prob{C | A,B} = \Prob{C | B}$.
This is exactly the definition of conditional independence of $A$ and $C$ given $B$, denoted by $A \perp C \mid B$.
This connection leads to the idea of verifying simulator-based security notions by proving conditional independence properties using probabilistic separation logic that can reason about conditional independence.

\subsection{Contributions}

In this paper, we formalize the connection between simulator-based definitions of probing security and conditional independence within the semantic framework of Lilac~\cite{LiPLDI2023}.
Although the underlying idea of this connection is conceptually simple, its formalization in Lilac is technically non-trivial. This difficulty stems from the fact that Lilac models probabilistic independence using \emph{independent combination}, rather than the standard construction based on product probability spaces.
The relationship between these two notions of independence has been studied in~\cite{LiLICS2024}; however, that work does not address the semantics of the conditioning modality defined via disintegration.
Consequently, establishing the desired connection in Lilac requires a careful measure-theoretic argument.

Although technically demanding, this formalization pays off by making it possible to reason about probing security directly within Lilac.
Our main theorem (Theorem~\ref{thm:noninterference-conditional-independence}) states that the simulator-based definition of noninterference is equivalent to the conditional independence property expressed as a Hoare triple in Lilac.
By the main theorem, to verify probing security of a masked algorithm, it suffices to prove the corresponding conditional independence property using proof rules of Lilac.

Moreover, we provide additional proof rules for Lilac that are necessary to verify probing security of masked algorithms but are not provided in the original Lilac paper~\cite{LiPLDI2023}.
In particular, we show that semi-graphoid axioms for conditional independence hold in Lilac.
These additional rules enhance the verification power of Lilac for reasoning about probing security.
We also consider proof rules related to composability of noninterference properties and capture them in Lilac.
Composability is an important aspect and has been studied in prior works~\cite{BartheCCS2016,CassiersTIFS2020}.
In contrast to prior works, our composability rules focus on $(\mathcal{I}, \mathcal{O})$-noninterference, which is a basic building block for the definitions of $t$-noninterference and its variants.
We demonstrate the effectiveness of our approach by verifying non-interference of the \textsc{Refresh} algorithm~\cite{CoronFSE2014,BartheCCS2016}, multiplication~\cite{IshaiCRYPTO2003}, and \textsc{AddRepNoise}~\cite{DelPinoCRYPTO2024} using Lilac: the first two are fundamental building blocks for masked algorithms, and the last one is a key component of the security proof of Raccoon~\cite{DelPino2023}.

The contributions of this work are summarized as follows:
\begin{itemize}
	\item We establish a connection between existence of simulators and conditional independence, which allows us to use probabilistic separation logics for conditional independence, such as Lilac, to verify probing security.
	\item We provide additional proof rules for Lilac that are necessary to verify probing security of masked algorithms.
	In particular, we show that semi-graphoid axioms for conditional independence hold in Lilac.
	We also provide proof rules related to composability of noninterference.
	\item We demonstrate the effectiveness of our approach by verifying non-interference of the \textsc{Refresh} algorithm~\cite{CoronFSE2014,BartheCCS2016}, multiplication~\cite{IshaiCRYPTO2003}, and \textsc{AddRepNoise}~\cite{DelPinoCRYPTO2024} using Lilac.
\end{itemize}

\section{Overview}

\subsection{Simulator-Based Security Notions and Probabilistic Independence}

We illustrate the connection between simulator-based security notions and probabilistic independence, using the one-time pad encryption as an example. 

\begin{example}[One-time pad]
	Suppose that a plaintext message is given as an element $M$ in a finite abelian group $\mathbb{K}$.
	The following program implements the one-time pad encryption:
	\begin{algorithmic}
		\Function{OTP}{$M$}
			\State $K \gets \mathtt{unif}\ \mathbb{K};$
			\State $C \gets \ret{(M - K)};$
			\State $\ret{C}$
		\EndFunction
	\end{algorithmic}
	Here, $\mathtt{unif}\ \mathbb{K}$ is the uniform distribution over $\mathbb{K}$, and $C \gets \ret{(M - K)}$ means deterministic assignment of the value $M - K$ to the variable $C$.
	The programming language used in this paper will be formally introduced later in Section~\ref{sec:appl}.
\end{example}

\paragraph{Simulator-based formulation of security}
We can formulate the security of the one-time pad using simulators.
In this setting, a simulator is defined as a program $\mathbf{Sim}(\textsc{OTP})$ that simulates the distribution of $C$ computed by $\textsc{OTP}(M)$ without knowing $M$.
If such a simulator exists, then the adversary cannot learn any information about $M$ because the adversary cannot distinguish the real execution from the simulated one, which does not depend on $M$.
Concretely, a valid simulator for the one-time pad can be given as a program that uniformly samples a random element from $\mathbb{K}$ as the ciphertext: $\mathbf{Sim}(\textsc{OTP}) \coloneqq \mathtt{unif}\ \mathbb{K}$.

\paragraph{Independence-based formulation}
Shannon's \emph{perfect secrecy} is a well-known security notion for encryption schemes \cite{Shannon1949} and defined as the independence between the plaintext $M$ and the ciphertext $C$.
This captures the intuition that observing the ciphertext $C$ does not provide any information about the plaintext $M$.
In fact, the independence of $M$ and $C$ is equivalent to the existence of a simulator described above.

To express and verify such independence properties of randomized programs, we can use probabilistic separation logic \cite{BarthePOPL2020,LiPLDI2023}.
In probabilistic separation logic, sources of randomness are treated as resources.
We have a basic assertion $\Own E$, which asserts ownership of a random variable $E$, and the separating conjunction $\phi * \psi$ expresses that the random variables mentioned in $\phi$ and those in $\psi$ are independent.
Therefore, we can express the independence between two random variables $X$ and $Y$ as $\Own X * \Own Y$.

Using probabilistic separation logic, the probabilistic independence between the plaintext $M$ and the ciphertext $C$ can be written as a Hoare triple
\[ \HoareTriple{\Own M}{\textsc{OTP}(M)}{C. \Own M * \Own C} \]
This states that, starting from an initial random variable satisfying the precondition $\Own M$, after executing the program $\textsc{OTP}(M)$, the resulting random variable satisfies the postcondition $\Own M * \Own C$ where $C$ is a random variable bound to the output of the program.

To prove that the above Hoare triple holds, we can use the proof rules of probabilistic separation logic, which have been already demonstrated in~\cite[Section~5.1.2]{BarthePOPL2020} for the one-time pad.
In this case, we do not need \emph{conditional} independence since the simulator does not take any input.
However, as we will see below, conditional independence becomes necessary when we consider simulators that take some input, such as in the case of non-interference verification.

\subsection{Noninterference of Masked Programs and Conditional Independence}

\emph{Masking} is a widely used countermeasure against side-channel attacks \cite{IshaiCRYPTO2003}, where a secret value is split into several shares via secret sharing and computations are performed on those shares.
A bit more formally, let $X$ be a secret value represented as an element in a finite abelian group $\mathbb{K}$.
The secret value $X$ is split into $t + 1$ \emph{shares}
\[ \mathbb{K} \ni X \quad\mapsto\quad \mathbf{X} = (X^{(0)}, \dots, X^{(t)}) \in \mathbb{K}^{t + 1} \]
such that the secret value $X$ can be recovered as the sum of all shares $X = X^{(0)} + \dots + X^{(t)}$, but any proper subset of the shares does not reveal any information about $X$.
The value $t$ is called the \emph{masking order}.
An algorithm is said to be \emph{masked} if secret input values are represented using such shares, and the algorithm operates on such shares.

\begin{example}\label{ex:masked-add}
	Consider adding two values $X$ and $Y$ in a finite abelian group $\mathbb{K}$.
	A masked implementation of the addition with two shares ($t = 1$) is given as follows.
	\begin{algorithmic}
		\Function{MaskedAdd}{$X^{(0)}, X^{(1)}, Y^{(0)}, Y^{(1)}$}
			\State $Z^{(0)} \gets \ret{(X^{(0)} + Y^{(0)})};$
			\State $Z^{(1)} \gets \ret{(X^{(1)} + Y^{(1)})};$
			\State $\ret{(Z^{(0)}, Z^{(1)})}$
		\EndFunction
	\end{algorithmic}
	It is straightforward to see that the output shares $Z^{(0)}$ and $Z^{(1)}$ represent the sum $Z = X + Y$ of two input values.
\end{example}

Masked algorithms are expected to be secure against side-channel attacks, since the leakage of at most $t$ shares does not reveal any information about the underlying secret value.
This security notion is formalized by \emph{$t$-noninterference}~\cite{BartheEUROCRYPT2015}, which we informally describe below.

\paragraph{Simulator-based formulation of $t$-noninterference}
In Example~\ref{ex:masked-add}, the secret values $X$ and $Y$ remain secure even if an adversary can only guess at most $t$ input shares for each of $X$ and $Y$.
This intuition is formalized as a security definition for masked algorithms, by considering a simulator that takes $t$ input shares as its input and outputs the values that are probed by an adversary.
In the $t$-probing model, an adversary is allowed to choose arbitrary probe locations and observe up to $t$ intermediate values during the execution of the program.
A masked algorithm is said to be \emph{$t$-noninterfering} if, for any choice of at most $t$ probes, there exists a simulator that, given at most $t$ input shares, can simulate the values observed at those probes.
The existence of such a simulator ensures that the adversary cannot learn any information about the secret values from the probed values, since the adversary cannot guess input shares that are not given to the simulator.

\paragraph{Conditional-independence-based formulation}
Similarly to the case of the one-time pad, we can rephrase the definition of $t$-noninterference in terms of \emph{conditional} independence.
Specifically, we consider the conditional independence between the input shares not in the simulator's input and the probed values, given the input shares provided to the simulator.
It turns out that this conditional independence is equivalent to the simulator-based definition of $t$-noninterference, which we will formally state in Theorem~\ref{thm:noninterference-conditional-independence}.

To reason about such conditional independence properties of masked programs, we can use an extension of probabilistic separation logic for conditional independence, namely, Lilac \cite{LiPLDI2023}.
In Lilac, the conditional independence between random variables $X$ and $Y$ given $Z$ can be expressed as a formula $\Conditioning_{z \gets Z} (\Own X * \Own Y)$ where the \emph{conditioning modality} $\Conditioning_{z \gets Z}$ represents conditioning on the random variable $Z$.
Using Lilac, the conditional independence required for $t$-noninterference of a masked program $\textsc{MaskedAdd}$ can be expressed as follows: for each choice of probe locations\footnote{Without loss of generality, we may assume that the set of probes $\mathcal{P}$ does not contain input shares: $\mathcal{P} \subseteq \{ Z^{(0)}, Z^{(1)} \}$.} $\mathcal{P} \subseteq \{ X^{(0)}, X^{(1)}, Y^{(0)}, Y^{(1)}, Z^{(0)}, Z^{(1)} \}$ such that $|\mathcal{P}| \le t = 1$, there exists a choice of at most one input share $\mathcal{X} \subseteq \mathbf{X}$ and $\mathcal{Y} \subseteq \mathbf{Y}$ from each of $\mathbf{X} = (X^{(0)}, X^{(1)})$ and $\mathbf{Y} = (Y^{(0)}, Y^{(1)})$ such that $|\mathcal{X}|, |\mathcal{Y}| \le |\mathcal{P}|$ and the following Hoare triple holds.
\begin{align}
	&\{\ \Own (\mathbf{X}, \mathbf{Y})\ \} \\
	&\textsc{MaskedAdd}(\mathbf{X}, \mathbf{Y}) \\
	&\{\ \mathbf{Z}. \Conditioning_{(x, y) \gets (\mathcal{X}, \mathcal{Y})} (\Own (\mathbf{X} \setminus \mathcal{X}, \mathbf{Y} \setminus \mathcal{Y}) * \Own \mathcal{P})\ \}
\end{align}
Here, $\mathbf{X} \setminus \mathcal{X}$ denotes the tuple of variables that are in $\mathbf{X}$ but not in $\mathcal{X}$.
This Hoare triple states that the information observed by the adversary (or the output of the simulator) $\mathcal{P}$ is independent of the secret input $\mathbf{X} \setminus \mathcal{X}$ and $\mathbf{Y} \setminus \mathcal{Y}$, given the public input $\mathcal{X}$ and $\mathcal{Y}$ provided to the simulator.

Once we have established the connection between simulator-based definitions of $t$-noninterference and conditional independence, it remains to show that the above Hoare triple holds for each choice of probe locations $\mathcal{P}$.
This can be done by applying the proof rules of Lilac.
For example, consider the case when the adversary probes the output share $\mathcal{P} = \{ Z^{(0)} \}$.
Then, we can choose $\mathcal{X} = \{ X^{(0)} \}$ and $\mathcal{Y} = \{ Y^{(0)} \}$ as the input shares to be given to the simulator.
Using Lilac's proof rules, we can derive the following Hoare triple:
\begin{align}
	&\{\ \Own (\mathbf{X}, \mathbf{Y})\ \} \\
	&\textsc{MaskedAdd}(\mathbf{X}, \mathbf{Y}) \\
	&\{\ \mathbf{Z}. \Conditioning_{(x, y) \gets (X^{(0)}, Y^{(0)})} (\Own (X^{(1)}, Y^{(1)}) * \Own Z^{(0)})\ \} \label{eq:masked-add-conditional-independence}
\end{align}
A bit more concretely, this is proved as follows.
We first use the rule for Hoare triples in Fig.~\ref{fig:lilac-hoare-rules} to derive the following Hoare triple.
\begin{align}
	&\{\ \Own (\mathbf{X}, \mathbf{Y})\ \} \\
	&\textsc{MaskedAdd}(\mathbf{X}, \mathbf{Y}) \\
	&\{\ \mathbf{Z}. (\Own (\mathbf{X}, \mathbf{Y})) * Z^{(0)} \AlmostSurelyEqual X^{(0)} + Y^{(0)} \ \}
	\label{eq:masked-add-intermediate}
\end{align}
Here, $\AlmostSurelyEqual$ stands for the almost-sure equality between two random variables.
By the consequence rule (\hyperlink{rule:H-Conseq}{\textsc{H-Conseq}}), it remains to show that the postcondition of~\eqref{eq:masked-add-intermediate} implies that of~\eqref{eq:masked-add-conditional-independence}.
This is done by using the proof rules for Lilac formulas.
Several proof rules are provided in~\cite{LiPLDI2023}; however, they are not sufficient to derive the desired entailment.
Therefore, we introduce several new rules for Lilac formulas in Section~\ref{sec:additional-rules}, including the following rule.
\[ \phi * (\Own (X, Y)) \LilacEntails \mathop{\mathbf{C}}_{x \gets X}(\phi * \Own Y) \]
This is a generalization of the \emph{weak union} property of conditional independence, which states $W \perp (X, Y) \mid Z$ implies $W \perp Y \mid (X, Z)$ for random variables $W$, $X$, $Y$, and $Z$.
We will prove the soundness of this rule in Proposition~\ref{prop:semi-graphoid-axioms}.
Using this rule, we can have the following entailment:
\begin{align}
	&\Own (\mathbf{X}, \mathbf{Y}) * Z^{(0)} \AlmostSurelyEqual X^{(0)} + Y^{(0)} \\
	&\LilacEntails \!\!\Conditioning_{(x, y) \gets (X^{(0)}, Y^{(0)})} \!\!(\Own (X^{(1)}, Y^{(1)}) * Z^{(0)} \AlmostSurelyEqual X^{(0)} + Y^{(0)})
\end{align}
Under the conditioning modality, $(X^{(0)}, Y^{(0)})$ are equal to deterministic values $(x, y)$.
Thus, we can conclude that $Z^{(0)} \AlmostSurelyEqual x + y$ is conditionally independent of $(X^{(1)}, Y^{(1)})$, as desired.

As we have seen in the above example, we can use Lilac to prove $t$-noninterference of masked programs, given (1) the connection between simulator-based definitions and conditional independence and (2) additional proof rules for Lilac formulas, such as the weak union rule.
In the subsequent sections, we will formally present such results.
We will also show more examples of masked programs and their security proofs using Lilac in Section~\ref{sec:examples}.

\section{Preliminaries}

\subsection{A Probabilistic Programming Language: APPL}\label{sec:appl}

We define a probabilistic programming language called APPL following \cite{LiPLDI2023}.

\subsubsection{Syntax}

\emph{Expressions} are defined as follows.
\begin{align}
	E \quad&\coloneqq\quad x \mid X \mid \mathtt{T} \mid \mathtt{F} \mid k \mid i \mid \ifexpr{E_0}{E_1}{E_2} \\
	&\qquad \mid (E_1, \dots, E_n) \mid \mathtt{proj}_{i}\ E \mid E_1 \mathbin{\mathtt{op}} E_2
\end{align}
Here, $x$ is a \emph{deterministic variable}, $X$ is a \emph{random variable}, $k \in \mathbb{K}$ is a constant in a finite abelian group $\mathbb{K}$, and $i$ is a natural number.
The constant $\mathtt{T}$ (resp.\ $\mathtt{F}$) represents the boolean value true (resp.\ false).
The symbol $\mathtt{op}$ ranges over binary operations, including group operations such as $+$ and $-$; boolean operations such as $\land$ and $\lor$; and comparison operations such as $=$, and $\neq$.

Expressions are typed in a standard way.
\emph{Types} are defined as follows.
\[ A, B \quad\coloneqq\quad \mathtt{bool} \mid \mathbb{K} \mid \mathtt{index} \mid A_1 \times \dots \times A_n \]
We have base types $\mathtt{bool}$, $\mathbb{K}$, and $\mathtt{index}$ for boolean values, elements of the group $\mathbb{K}$, and indices for for loops, respectively.
Types also include product types $A_1 \times \dots \times A_n$.
The empty product type ($n = 0$) is written as $\mathtt{unit}$. 

\emph{Typing judgements} for expressions are given in the form of $\Gamma; \Delta \vdash E : A$ where $\Gamma = x : A_1, \dots, x_m : A_m$ is a typing context for deterministic variables and $\Delta = X : A_1, \dots, X_n : A_n$ is a typing context for random variables.
Typing rules are standard~\cite{LiPLDI2023} and omitted here.

\emph{APPL programs} are defined as follows.
\begin{align}
	L, M, N &\coloneqq\quad \ret{E} \mid X \gets M; N \mid \mathtt{unif}\ \mathbb{K} \\
	&\quad \mid \forloop{i}{n_{\mathrm{start}}}{n_{\mathrm{end}}}{X}{E}{M}
\end{align}
In the for-loop $\forloop{i}{n_{\mathrm{start}}}{n_{\mathrm{end}}}{X}{E}{M}$, $i$ is the loop index variable ranging from $n_{\mathrm{start}}$ to $n_{\mathrm{end}}$.
The variable $X$ is initialized with $E$ at the beginning of the loop, and in each iteration, $X$ is updated with the result of $M$.
The final value of $X$ after the last iteration is returned as the result of the for-loop.
One might expect the for-loop to be written as $X \gets E; \mathtt{for}\ i \gets n_{\mathrm{start}} \cdots n_{\mathrm{end}}\ \{ M \}$.
However, we adopt the above syntax because $X$ is bound by
the for-loop construct.

In sequential composition $X \gets M; N$ and for-loops, we often bind multiple variables at once.
Formally, when the result of $M$ is a tuple, we define $(X, Y) \gets M; N$ as syntactic sugar for $Z \gets M; N[\mathtt{proj}_{1}\ Z / X, \mathtt{proj}_{2}\ Z / Y]$ where $Z$ is a fresh variable, and similarly for for-loops.
In many examples in this paper, we also use indexed variables such as $X_i$ and $X_{i, j}$ and consider programs such as $X_E \gets M; N$ where $E$ is an expression of type $\mathtt{index}$.

We do not include conditional branching as a primitive construct in APPL.
However, they can be defined as syntactic sugar using conditional branching for expressions.
\begin{align}
	&\ifexpr{E}{M}{N} \\
	&\coloneqq\quad X \gets M;\ Y \gets N;\ \ret{(\ifexpr{E}{X}{Y})}
\end{align}
Since APPL does not cause computational effects other than randomness, this definition is semantically equivalent to the standard conditional branching construct.

\begin{figure*}
	\begin{mathpar}
		\inferrule{
			\emptyset; \Delta \vdash E : A
		}{
			\Delta \vdash \ret{E} : \mathtt{G} A
		}
		\and
		\inferrule{
			\Delta \vdash M : \mathtt{G} A \\
			\Delta, X : A \vdash N : \mathtt{G} B
		}{
			\Delta \vdash X \gets M; N : \mathtt{G} B
		}
		\and
		\inferrule{ }{
			\Delta \vdash \mathtt{unif}\ \mathbb{K} : \mathtt{G} \mathbb{K}
		}
		\and
		\inferrule{
			\emptyset; \Delta \vdash E : A \\
			\Delta, i : \mathtt{index}, X : A \vdash M : \mathtt{G} A
		}{
			\Delta \vdash \forloop{i}{n_s}{n_e}{X}{E}{M} : \mathtt{G} A
		}
	\end{mathpar}
	\caption{Typing rule for APPL programs.}
	\label{fig:appl-typing-rule}
\end{figure*}

Typing judgments for programs are given in the form of $\Delta \vdash M : \mathtt{G} A$ where $\Delta$ is a typing context for random variables.
Here, $\mathtt{G} A$ represents a type of distributions over values of type $A$.
Typing rules are shown in Fig.~\ref{fig:appl-typing-rule}.

\subsubsection{Semantics}

Each type $A$ is interpreted as a measurable space $\interpret{A}$ in the standard way.
Note that all base types are interpreted as countably generated measurable spaces, which implies that any type is also interpreted as a countably generated measurable space.
The interpretation of a context $\Gamma = x : A_1, \dots, x_m : A_m$ is defined as the product measurable space $\interpret{\Gamma} = \interpret{A_1} \times \dots \times \interpret{A_m}$.

Expressions $\Gamma; \Delta \vdash E : A$ are interpreted as functions
\[ \interpret{E} : \interpret{\Gamma} \to \mathbf{Meas}(\interpret{\Delta}, \interpret{A}) \]
where $\mathbf{Meas}(X, Y)$ is the set of measurable functions from $X$ to $Y$.
Programs $\Delta \vdash M : \mathtt{G} A$ are interpreted as Markov kernels, i.e., measurable functions that maps each element in $\interpret{\Delta}$ to a probability measure over $\interpret{A}$:
\[ \interpret{M} \in \mathbf{Meas}(\interpret{\Delta}, \Giry \interpret{A}) \]
Here, $\Giry$ is the \emph{Giry monad}, which maps a measurable space to the measurable space of probability measures over it.
Concrete definitions of the semantics are standard \cite{LiPLDI2023} and omitted.

\subsection{$t$-Probing Model and Noninterference}
\label{sec:probing-model}

In the $t$-probing model, an adversary is allowed to probe up to $t$ internal or output values during the execution of a masked program where $t$ is the order of masking.
The adversary aims to guess secret information from the observed values, which is achieved if the adversary can guess more than $t$ input shares.
The notion of \emph{$t$-noninterference} formalizes the security requirement in this setting.
Suppose that $\mathbf{X}_1, \dots, \mathbf{X}_n$ are shared inputs where each $\mathbf{X}_i = (X_i^{(0)}, \dots, X_i^{(t)})$ is a tuple of $t + 1$ shares.
A masked program $M$ is said to be \emph{$t$-noninterfering} ($t$-NI) if for any set of probes $\mathcal{O}$ with $|\mathcal{O}| \le t$, there exists a subset $\mathcal{I}_i \subseteq \mathbf{X}_i$ for each $i$ such that $|\mathcal{I}_i| \le |\mathcal{O}|$ and the values of probes in $\mathcal{O}$ can be simulated using only the input shares in $\mathcal{I}_1, \dots, \mathcal{I}_n$.

Below, we rephrase the definitions of $t$-noninterference and related notions in terms of APPL programs.
Suppose that we have an APPL program $\Delta \vdash M : \mathtt{G} B$ where $\Delta = X_1 : A_1, \dots, X_m : A_m$ and $B = B_1 \times \dots \times B_n$.
We assume that $M$ is in the static single assignment (SSA) form so that each internal value computed in $M$ is assigned to a unique variable name.
Then, probes can be specified by a set of variables.
We call variables $X_1, \dots, X_m$ in the context $\Delta$ the \emph{input variables} of $M$.
The \emph{output variables} of $M$ are variables bound to the output value of $M$ and have type $B$.
The output variables do not a priori occur in the program, but we often introduce them to refer to the output of $M$.
The \emph{internal variables} of $M$ are variables occurring in $M$ other than input variables.

Later in Section~\ref{sec:lilac}, we will introduce Hoare triples to reason about APPL programs.
A Hoare triple is given in the form $\HoareTriple{\phi}{M}{\mathbf{Y}. \psi}$ where $\phi$ is an assertion over input variables of $M$, $\mathbf{Y}$ is a tuple of variables bound to the output of $M$, and $\psi$ is an assertion over both input variables and output variables $\mathbf{Y}$.
Hoare triples cannot refer to internal variables directly, but we can easily transform a given APPL program to expose internal variables as output variables.

For example, the APPL program $X : \mathbb{K} \vdash Y \gets \ret{X + X}; \ret{Y + X} : \mathbb{K}$ has one input variable $X$ and one internal variable $Y$, and outputs $Y + X$.
Let $Z$ be an output variable bound to the output.
Hoare triples $\HoareTriple{\phi}{Y \gets \ret{X + X}; \ret{Y + X}}{Z. \psi}$ for this program can refer to $X$ and $Z$, but not to $Y$ directly.
In this case, we can transform the program to $X : \mathbb{K} \vdash Y \gets \ret{X + X}; Z \gets \ret{Y + X}; \ret{(Y, Z)} : \mathbb{K} \times \mathbb{K}$ so that both $Y$ and $Z$ are accessible as output variables.

In general, given a program $M$, we can construct a program $\ExposingTranslation{M}$ that exposes all internal variables as output.
\[ \Delta \vdash M : \mathtt{G} B \quad\mapsto\quad \Delta \vdash \ExposingTranslation{M} : \mathtt{G} (B_{\mathrm{int}} \times B) \]
Here, $B_{\mathrm{int}}$ is the type of the tuple of all internal variables of $M$.
Of course, the transformed program $\ExposingTranslation{M}$ is defined so that $M$ is semantically equal to $(\mathbf{Y}, \mathbf{Z}) \gets \ExposingTranslation{M}; \ret{\mathbf{Z}}$ where $\mathbf{Y}$ and $\mathbf{Z}$ are tuples of internal and output variables, respectively.

In what follows, we assume that APPL programs are appropriately transformed so that internal variables are accessible as parts of output variables.
Then, probes can be specified by output variables only.
Now, we define the notion of noninterference for APPL programs.
Throughout the paper, we write a tuple of variables $\mathbf{X} = (X_1, \dots, X_m)$ in bold font and often identify the tuple with the set $\{X_1, \dots, X_m\}$ for notational convenience.

\begin{definition}
	Let $\Delta \vdash M : \mathtt{G} B$ be an APPL program where $\Delta = X_1 : A_1, \dots, X_m : A_m$ and $B = B_1 \times \dots \times B_n$.
	We write $\vec{I} = \{X_1, \dots, X_m\}$ for the set of input variables and $\vec{O} = \{Y_1, \dots, Y_n\}$ for the set of output variables of $M$.
	For $\mathcal{I} \subseteq \vec{I}$ and $\mathcal{O} \subseteq \vec{O}$, we say that $M$ is \emph{$(\mathcal{I}, \mathcal{O})$-noninterfering} (\emph{$(\mathcal{I}, \mathcal{O})$-NI}) if there exists an APPL program $\Delta_{\mathcal{I}} \vdash \mathbf{Sim}(M) : \mathtt{G} B_{\mathcal{O}}$ such that the following equation holds:
	\begin{equation}
		\mathbf{Sim}(M) \quad=\quad \vec{O} \gets M;\ \ret{\mathcal{O}}
	\end{equation}
	Here, $\Delta_{\mathcal{I}}$ is the restriction of $\Delta$ to variables in $\mathcal{I}$ and $B_{\mathcal{O}} = \prod_{Y_i \in \mathcal{O}} B_i$ is the restriction of $B$ to components in $\mathcal{O}$.
	We call such a program $\mathbf{Sim}(M)$ a \emph{simulator} for $M$ with input $\mathcal{I}$ and output $\mathcal{O}$.
\end{definition}

The notion of $t$-noninterference \cite{BartheEUROCRYPT2015} for masked algorithms is defined using $(\mathcal{I}, \mathcal{O})$-NI.

\begin{definition}[$t$-NI]\label{def:t-noninterference}
	Suppose that the input variables of the APPL program $M$ consist of shared values $\mathbf{X}_1, \dots, \mathbf{X}_n$ where each $\mathbf{X}_i = (X_i^{(0)}, \dots, X_i^{(t)})$ is a tuple of $t + 1$ shares, and suppose that the output variables are $\mathbf{O} = \{Y_1, \dots, Y_m\}$.
	We say that $M$ is \emph{$t$-noninterfering} ($t$-NI) if for any set of probes $\mathcal{O} \subseteq \mathbf{O}$ with $|\mathcal{O}| \le t$, there exist subsets $\mathcal{X}_1 \subseteq \mathbf{X}_1, \dots, \mathcal{X}_n \subseteq \mathbf{X}_n$ such that $|\mathcal{X}_i| \le |\mathcal{O}|$ for each $i$, and $M$ is $(\mathcal{X}_1 \cup \dots \cup \mathcal{X}_n, \mathcal{O})$-noninterfering.
\end{definition}

There are several variants of $t$-noninterference proposed in the literature.
One of such variants is \emph{strong noninterference} (SNI) \cite{BartheCCS2016}, which exhibits a better composability property.
This variant is also defined based on $(\mathcal{I}, \mathcal{O})$-NI.

\begin{definition}[$t$-SNI]\label{def:t-strong-noninterference}
	Suppose the same setting as in Definition~\ref{def:t-noninterference}.
	Let $\mathbf{O}_{\mathrm{int}} \subseteq \mathbf{O}$ be the set of exposed internal variables, which are not used as input to subsequent programs.
	We say that $M$ satisfies \emph{$t$-strong noninterference} ($t$-SNI) if for any set of probes $\mathcal{O} \subseteq \mathbf{O}$ with $|\mathcal{O}| \le t$, there exist subsets $\mathcal{X}_1 \subseteq \mathbf{X}_1, \dots, \mathcal{X}_n \subseteq \mathbf{X}_n$ such that for each $i$, $|\mathcal{X}_i| \le |\mathcal{O} \cap \mathbf{O}_{\mathrm{int}}|$ and $M$ is $(\mathcal{X}_1 \cup \dots \cup \mathcal{X}_n, \mathcal{O})$-noninterfering.
\end{definition}

Another variant is \emph{noninterference with unshared input} (NIU) \cite{EsginEUROCRYPT2024}, which generalizes ordinary noninterference by allowing unshared inputs in addition to shared inputs.

\begin{definition}[$t$-NIU]\label{def:t-noninterference-with-unshared-input}
	Suppose that the input variables of the APPL program $M$ consist of shared inputs $\mathbf{X}_1, \dots, \mathbf{X}_n$ and unshared inputs $\vec{Y} = (Y_1, \dots, Y_m)$ where each $\mathbf{X}_i = (X_i^{(0)}, \dots, X_i^{(t)})$ is a tuple of $t + 1$ shares.
	Let $\mathbf{O}$ be the set of output variables of $M$.
	We say that $M$ is \emph{$t$-noninterfering with unshared input} ($t$-NIU) if for any set of probes $\mathcal{O} \subseteq \mathbf{O}$ with $|\mathcal{O}| \le t$, there exist subsets $\mathcal{X}_1 \subseteq \mathbf{X}_1, \dots, \mathcal{X}_n \subseteq \mathbf{X}_n$ and $\mathcal{Y} \subseteq \mathbf{Y}$ such that $|\mathcal{X}_i| \le |\mathcal{O}|$ for each $i$, $|\mathcal{Y}| \le |\mathcal{O}|$, and $M$ is $(\mathcal{X}_1 \cup \dots \cup \mathcal{X}_n \cup \mathcal{Y}, \mathcal{O})$-noninterfering.
	A strong version of NIU (SNIU) can be defined similarly.
\end{definition}

There are other variants of $t$-NI in the literature, such as \emph{noninterference with public outputs} (NIo)~\cite{BartheEUROCRYPT2018} and \emph{probe isolating noninterference} (PINI)~\cite{CassiersTIFS2020}.
Although we do not give their definitions here, they can also be defined via $(\mathcal{I}, \mathcal{O})$-NI.
Since $(\mathcal{I}, \mathcal{O})$-NI is the basis of various notions of noninterference, we mainly focus on verifying $(\mathcal{I}, \mathcal{O})$-NI for APPL programs in the rest of this paper.

\section{Lilac: A Separation Logic for Conditional Independence}
\label{sec:lilac}
We briefly review the syntax and semantics of Lilac \cite{LiPLDI2023}, a probabilistic separation logic that can reason about conditional independence.

\paragraph*{Syntax}
\emph{Formulas} of Lilac are defined as follows.
\begin{align}
	\phi, \psi &\coloneqq\quad \top \mid \bot \mid \phi \land \psi \mid \phi \lor \psi \mid \forall x. \phi \mid \exists x. \phi \\
	&\quad \mid \forall_{\mathrm{rv}} X. \phi \mid \exists_{\mathrm{rv}} X. \phi \mid \phi * \psi \mid \phi \sepimp \psi \mid \Own E \mid \Persistent \phi \\
	&\quad \mid \Conditioning_{x : A \gets E} \phi \mid E \sim \mu \mid E_1 \AlmostSurelyEqual E_2 \mid \mathrm{wp}(M, X : A. \phi)
\end{align}
Here, $\mu \coloneqq \Uniform{\mathbb{K}} \mid \cdots $ is a distribution, including the uniform distribution $\Uniform{\mathbb{K}}$ over a finite group $\mathbb{K}$.
Formulas defined above assert propositions about random variables and the ownership of sources of randomness.
Below we briefly explain the intended meaning of each construct.
The formula $\Own E$ asserts the probabilistic ownership of the random variable $E$.
The formula $E \sim \mu$ asserts that the random variable $E$ follows the distribution $\mu$ and that its probabilistic ownership is held.
The formula $E_1 \AlmostSurelyEqual E_2$ asserts that the random variables $E_1$ and $E_2$ are almost surely equal.
The separating conjunction $\phi * \psi$ asserts that $\phi$ and $\psi$ own disjoint, probabilistically independent sources of randomness and that $\phi$ holds for one part and $\psi$ holds for the other part.
The \emph{conditioning modality} $\Conditioning_{x \gets E} \phi$, which is one of the main features of Lilac, asserts that $\phi$ holds when conditioning on the value of $E$ being $x$.
The weakest precondition $\mathrm{wp}(M, X : A. \phi)$ for program $M$ and postcondition $\phi$ asserts that after executing $M$ and binding the result to the random variable $X$, the formula $\phi$ holds.
The persistence modality~\cite{JungJFunctProg2018} is denoted by $\Persistent \phi$, which asserts that $\phi$ holds without owning any source of randomness.
There are two types of quantifiers: the standard quantifiers $\forall x. \phi$ and $\exists x. \phi$ for deterministic variables, and the random variable quantifiers $\forall_{\mathrm{rv}} X. \phi$ and $\exists_{\mathrm{rv}} X. \phi$ for random variables.

For example, $\Own X * \Own Y$ means that the random variables $X$ and $Y$ are independent.
Conditional independence of $X$ and $Y$ given $Z$ is expressed as $\Conditioning_{z \gets Z} (\Own X * \Own Y)$.

Typing judgements for Lilac formulas are denoted by $\Gamma; \Delta \vdash \phi$ where $\Gamma$ is a typing context for deterministic variables and $\Delta$ is a typing context for random variables.
Typing rules are standard \cite{LiPLDI2023} and omitted here.

We write $\phi \LilacEntails \psi$ if $\phi$ entails $\psi$ in Lilac.
We also write $\LilacEntails \psi$ when $\phi$ is $\top$.
If both $\phi \LilacEntails \psi$ and $\psi \LilacEntails \phi$ hold, then we write $\phi \LilacIff \psi$.
The semantics of the entailment relation is given later.
Inference rules for the entailment relation are shown in Fig.~\ref{fig:lilac-rules}.
Most of the rules are taken from \cite{LiPLDI2023}, but we also include some rules that are not explicitly in \cite{LiPLDI2023} but are easily proved to be valid.
Later in Section~\ref{sec:additional-rules}, we also provide several non-trivial additional rules that are not in \cite{LiPLDI2023}.

To reason about APPL programs, Hoare triples are defined in Lilac as follows.
\[ \HoareTriple{\phi}{M}{X. \psi} \quad\coloneqq\quad \Persistent (\phi \sepimp \mathrm{wp}(M, X : A. \psi)) \]
This is semantically equivalent to
$\phi \LilacEntails \mathrm{wp}(M, X : A. \psi)$.
Inference rules for Hoare triples are shown in Fig.~\ref{fig:lilac-hoare-rules}.

\begin{figure*}
	\begin{mathpar}
		\inferrule[\hypertarget{rule:C-Entail}{C-Entail}]{
			\phi \LilacEntails \psi
		}{
			\Conditioning_{x \gets E} \phi \LilacEntails \Conditioning_{x \gets E} \psi
		}
		\and
		\inferrule[\hypertarget{rule:C-Indep}{C-Indep}]{}{
			\Own E * \phi \LilacEntails \Conditioning_{x \gets E} \phi
		}
		\and
		\inferrule[\hypertarget{rule:C-Subst}{C-Subst}]{}{
			\Own X \LilacEntails \Conditioning_{x \gets X} (X \AlmostSurelyEqual x)
		}
		\and
		\inferrule[\hypertarget{rule:C-Own}{C-Own}]{}{
			\Own E \LilacEntails \Conditioning_{x \gets E} \Own E
		}
		\and
		\inferrule[\hypertarget{rule:Persistent-Elim}{$\Persistent$-Elim}]{}{
			\Persistent \phi \LilacEntails \phi
		}
		\and
		\inferrule[\hypertarget{rule:C-And-Distributive}{C-And-Distributive}]{}{
			\Conditioning_{x \gets X} (\phi \land \psi) \LilacIff \Conditioning_{x \gets X} \phi \land \Conditioning_{x \gets X} \psi
		}
		\and
		\inferrule[\hypertarget{rule:C-Unif}{C-Unif}]{}{
			\Own X \land \Conditioning_{x \gets X} (Y \sim \Uniform{\mathbb{K}}) \LilacEntails \Own X * (Y \sim \Uniform{\mathbb{K}})
		}
		\and
		\inferrule[\hypertarget{rule:Persistent-Conj}{$\Persistent$-Conj}]{}{
			\phi \land \Persistent \psi \LilacEntails \phi * \Persistent \psi
		}
		\and
		\inferrule[\hypertarget{rule:Unif-Bijection}{Unif-Bijection}]{
			\Gamma ; X : \mathbb{K} \vdash E : \mathbb{K} \\
			\text{$E$ is bijective in $X$}
		}{
			X \sim \Uniform{\mathbb{K}} \LilacEntails E \sim \Uniform{\mathbb{K}}
		}
		\and
		\inferrule[\hypertarget{rule:Transfer-Own}{Transfer-Own}]{}{
			\Own E_1 \land (E_1 \AlmostSurelyEqual E_2) \LilacEntails \Own E_2
		}
	\end{mathpar}
	\caption{Selected inference rules for Lilac formulas~\cite{LiPLDI2023}.
	See
	\iflong
	Fig.~\ref{fig:lilac-rules-full} in the appendix
	\fi
	for a more comprehensive list of rules.
	We omit obvious structural rules like commutativity and associativity of the separating conjunction $*$.
	Note that Lilac is an \emph{affine} separation logic, and thus $\phi * \psi \LilacEntails \phi$ holds.}
	\label{fig:lilac-rules}
\end{figure*}

\begin{figure*}
	\begin{mathpar}
		\inferrule[H-Unif]{ }{
			\HoareTriple{\top}{\mathtt{unif}\ \mathbb{K}}{X. X \sim \Uniform{\mathbb{K}}}
		}
		\and
		\inferrule[H-Ret]{ }{
			\HoareTriple{Q[E/X]}{\ret{E}}{X. Q}
		}
		\and
		\inferrule[H-Let]{
			\HoareTriple{P}{M}{X. Q} \\
			\forall_{\mathrm{rv}} X. \HoareTriple{Q}{N}{Y. R}
		}{
			\HoareTriple{P}{X \gets M; N}{Y. R}
		}
		\and
		\inferrule[H-For]{
			\forall i : \mathtt{index}. \forall_{\mathrm{rv}} X : A. \HoareTriple{I(i, X)}{M}{X'. I(i + 1, X')}
		}{
			\HoareTriple{I(n_s, E)}{\forloop{i}{n_s}{n_e}{X}{E}{M}}{X. I(n_e + 1, X)}
		}
		\and
		\inferrule[\hypertarget{rule:H-Conseq}{H-Conseq}]{
			P \LilacEntails P' \\
			\HoareTriple{P'}{M}{X. Q'} \\
			Q' \LilacEntails Q
		}{
			\HoareTriple{P}{M}{X. Q}
		}
		\and
		\inferrule[\hypertarget{rule:H-Frame}{H-Frame}]{
			\HoareTriple{P}{M}{X. Q} \\
			X \notin R
		}{
			\HoareTriple{P * R}{M}{X. Q * R}
		}
	\end{mathpar}
	\caption{Inference rules for Hoare triples in Lilac \cite{LiPLDI2023}}
	\label{fig:lilac-hoare-rules}
\end{figure*}

\paragraph*{Semantics}
We briefly explain the semantics of Lilac formulas \cite{LiPLDI2023}.
The semantics of separating conjunction in Lilac is based on the notion of \emph{Kripke resource monoids} (KRMs) \cite{GalmicheMSCS2005}.
A KRM is a partially ordered partial commutative monoid $(M, 1, {\cdot}, {\sqsubseteq})$ such that the order relation $\sqsubseteq$ is compatible with the monoid product.
To capture probabilistic independence of random variables, Lilac uses a specific KRM defined on the set of probability spaces on a fixed sample space $\Omega$.
\begin{definition}
	Let $(\Omega, \Sigma_{\Omega})$ be the Hilbert cube $[0, 1]^{\mathbb{N}}$ with the standard Borel $\sigma$-algebra.
	We define a KRM $\mathcal{M}_{\mathrm{ind}} = (M, 1, {\cdot}, {\sqsubseteq})$ as follows.
	\begin{itemize}
		\item The underlying set $\Omega$ is the set of probability spaces $\mathcal{P} = (\Omega, \mathcal{F}, \mu)$ over the sample space $\Omega$ where $\mathcal{F} \subseteq \Sigma_{\Omega}$ is a $\sigma$-algebra and $\mu : \mathcal{F} \to [0, 1]$ is a probability measure.
		We often omit $\Omega$ and simply write $(\mathcal{F}, \mu)$ for a probability space in $\mathcal{M}$.
		\item The monoid product $\cdot$ is defined by the \emph{independent combination} of probability spaces: given a pair of probability spaces $(\mathcal{F}, \mu)$ and $(\mathcal{G}, \nu)$ in $\KRM$, their \emph{independent combination} $(\mathcal{F}, \mu) \cdot (\mathcal{G}, \nu)$ is a (necessarily unique) probability space $(\mathcal{H}, \lambda)$ such that $\mathcal{H}$ is the smallest $\sigma$-algebra containing both $\mathcal{F}$ and $\mathcal{G}$, and $\lambda$ witnesses the independence of $\mathcal{F}$ and $\mathcal{G}$, i.e., for any $F \in \mathcal{F}$ and $G \in \mathcal{G}$, $\lambda(F \cap G) = \mu(F) \cdot \nu(G)$.
		The unit element is the trivial probability space $(\{\emptyset, \Omega\}, \mu)$.
		\item The order relation $(\mathcal{F}, \mu) \sqsubseteq (\mathcal{G}, \nu)$ is defined by $\mathcal{F} \subseteq \mathcal{G}$ and $\mu = \nu|_{\mathcal{F}}$ where $\nu|_{\mathcal{F}}$ is the restriction of $\nu$ to $\mathcal{F}$.
	\end{itemize}
\end{definition}
Conditioning modalities in Lilac are interpreted by disintegrations \cite{ChangStatisticaNeerlandica1997}.
To ensure the existence of disintegrations, Lilac formulas are interpreted using a submonoid $\KRM \subseteq \mathcal{M}_{\mathrm{ind}}$ that imposes technical conditions on probability spaces (see \cite[Section~3]{LiPLDI2023} for details).

The semantics of Lilac formulas $\Gamma; \Delta \vdash \phi$ is defined by the relation
$\gamma, D, \mathcal{P} \models \phi$
where $\gamma \in \interpret{\Gamma}$ is a valuation for deterministic variables in $\Gamma$, $D : \Omega \to \interpret{\Delta}$ is a measurable function representing a random variable environment for $\Delta$, and $\mathcal{P} \in \mathcal{M}_{\mathrm{disint}}$ is an element of the KRM of which the formula has ownership.
Below, we write $\RV{A} \coloneqq \{ D : \Omega \to A \}$ for the set of random variables where $A$ is a measurable space.
We show the semantics of important constructs; the semantics of other constructs can be found in~\cite{LiPLDI2023}.\footnote{The semantics of almost-sure equality presented in~\cite{LiPLDI2023} contains a flaw and was later corrected in the arXiv version. In this paper, we adopt the corrected semantics. In fact, our results do not depend heavily on the semantics of almost-sure equality itself, since we use it primarily through syntactic inference rules, whose soundness has been proved under the corrected semantics.}
\begin{itemize}
	\item $\gamma, D, (\mathcal{F}, \mu) \models \Own E$ if $\interpret{E}(\gamma) \circ D$ is $\mathcal{F}$-measurable where $\circ$ is function composition.
	\item $\gamma, D, (\mathcal{F}, \mu) \models E \sim \nu$ if $\interpret{E}(\gamma) \circ D$ is $\mathcal{F}$-measurable and the pushforward measure $(\interpret{E}(\gamma) \circ D)_{*} \mu$ is $\nu$.
	\item $\gamma, D, \mathcal{P} \models \phi * \psi$ if $\gamma, D, \mathcal{P}_1 \models \phi$ and $\gamma, D, \mathcal{P}_2 \models \psi$ for some $\mathcal{P}_1, \mathcal{P}_2 \in \mathcal{M}_{\mathrm{disint}}$ such that $\mathcal{P}_1 \cdot \mathcal{P}_2 \sqsubseteq \mathcal{P}$.
	\item $\gamma, D, (\mathcal{F}, \mu) \models \Conditioning_{x : A \gets E} \phi$ if for any $(\Sigma_{\Omega}, \mu') \sqsupseteq (\mathcal{F}, \mu)$ and any $\mu'$-disintegration $\{\nu_x\}_{x \in \interpret{A}}$ with respect to $\interpret{E}(\gamma) \circ D$, it holds that $\gamma, D, (\mathcal{F}, \nu_x|_{\mathcal{F}}) \models \phi$ for almost every $x \in \interpret{A}$.
	Here, the disintegration $\{\nu_x\}_{x \in \interpret{A}}$ intuitively represents the conditional distribution given the value of $\interpret{E}(\gamma) \circ D \in \RV{\interpret{A}}$.
	\item $\gamma, D, \mathcal{P} \models \mathrm{wp}(M, X : A. \phi)$ if for all $\mathcal{P}_f$ and $\mu$ such that $\mathcal{P}_f \cdot \mathcal{P} \sqsubseteq (\Sigma_{\Omega}, \mu)$ and for all $D \in \RV{\interpret{\Delta_{\mathrm{e}}}}$, there exist $X \in \RV{\interpret{A}}$, $\mathcal{P}'$, and $\mu'$ such that $\mathcal{P}_f \cdot \mathcal{P}' \sqsubseteq (\Sigma_{\Omega}, \mu')$, $\gamma, (D, X), \mathcal{P}' \models \phi$, and the following two distributions are equal:
\end{itemize}
\[ \left( \begin{aligned}
		&\omega \gets \mu; \\
		&v \gets \interpret{M}(D(\omega)); \\
		&\ret{(D_{\mathrm{e}}(\omega), D(\omega), v)}
	\end{aligned} \right) = \left( \begin{aligned}
		&\omega \gets \mu'; \\
		&\ret{(D_{\mathrm{e}}(\omega), D(\omega), X(\omega))}
	\end{aligned} \right) \]

Finally, we define the entailment relation $\phi \LilacEntails \psi$ as follows.
\[ \forall \gamma. \forall D. \forall \mathcal{P}.\quad \gamma, D, \mathcal{P} \models \phi \quad\implies\quad \gamma, D, \mathcal{P} \models \psi \]

\begin{toappendix}
\section{Basic Facts about Probability Theory}

Let $f : A \to B$ be a measurable function between measurable spaces $(A, \Sigma_A)$ and $(B, \Sigma_B)$, and let $\mu$ be a probability measure on $(A, \Sigma_A)$.
The \emph{pushforward measure} $f_{*} \mu$ is the probability measure on $(B, \Sigma_B)$ defined by
\[ (f_{*} \mu)(E) = \mu(f^{-1}(E)) \]
for any measurable set $E \in \Sigma_B$.

Let $P$ be a predicate on a set $X$.
The \emph{Iverson bracket} $\IversonBracket{P}$ is defined by
\[ \IversonBracket{P} = \begin{cases}
	1 & \text{$P$ is true} \\
	0 & \text{otherwise.}
\end{cases} \]

\begin{definition}[$\pi$-system]
	Let $P \subseteq 2^X$ be a nonempty collection of subsets of $X$.
	Then, $P$ is a \emph{$\pi$-system} if for all $A, B \in P$, we have $A \cap B \in P$.
\end{definition}

\begin{definition}[$\lambda$-system]
	Let $L \subseteq 2^X$ be a nonempty collection of subsets of $X$.
	Then, $L$ is a \emph{$\lambda$-system} if:
	\begin{itemize}
		\item $X \in L$.
		\item If $A, B \in L$, then $B \setminus A \in L$.
		\item If $A_1 \subseteq A_2 \subseteq \dots$ is an increasing sequence in $L$, then $\bigcup_{i = 1}^{\infty} A_i \in L$.
	\end{itemize}
\end{definition}

\begin{theorem}[Dynkin's $\pi$-$\lambda$ theorem]
	Let $\mathcal{P}$ be a $\pi$-system and $\mathcal{L}$ be a $\lambda$-system such that $\mathcal{P} \subseteq \mathcal{L}$.
	Then, the $\sigma$-algebra generated by $\mathcal{P}$ is contained in $\mathcal{L}$.
	\[ \mathcal{P} \subseteq \mathcal{L} \implies \sigma(\mathcal{P}) \subseteq \mathcal{L} \]
\end{theorem}

We list a few criteria for two functions/Markov kernels to be almost everywhere equal.
\begin{lemma}\label{lem:aee-function-le}
	Let $\mu$ be a probability measure on $(X, \Sigma_X)$ and $f, g : (X, \Sigma_X) \to [0, \infty]$ be measurable functions such that $f(x) \le g(x)$ for any $x \in X$.
	Then, $f$ and $g$ are $\mu$-almost everywhere equal if and only if 
	$\Expectation{x}{\mu}{f(x)} = \Expectation{x}{\mu}{g(x)}$.
	\qed
\end{lemma}

\begin{lemma}\label{lem:aee-function}
	Let $f, g : (X, \Sigma_X) \to [0, \infty]$ be measurable functions and $\mu$ be a probability measure on $(X, \Sigma_X)$.
	Then, $f$ and $g$ are $\mu$-almost everywhere equal if and only if for any measurable set $E \in \Sigma_X$, we have the following.
	\[ \Expectation{x}{\mu}{\IversonBracket{x \in E} \cdot f(x)} = \Expectation{x}{\mu}{\IversonBracket{x \in E} \cdot g(x)} \]
\end{lemma}
\begin{proof}
	The only-if part is obvious.
	For the if part, let $E_{<} = \{ x \in X : f(x) < g(x) \}$.
	If $E_{<}$ is non-negligible, then we have
	\[ 
	\Expectation{x}{\mu}{\IversonBracket{x \in E_{<}} \cdot f(x)} < \Expectation{x}{\mu}{\IversonBracket{x \in E_{<}} \cdot g(x)}, \]
	which contradicts the assumption.
	Similarly, we can define $E_{>} = \{ x \in X : f(x) > g(x) \}$ and obtain a contradiction if $E_{>}$ is non-negligible.
	Therefore, $\{ x \in X : f(x) \neq g(x) \} = E_{<} \cup E_{>}$ is negligible.
\end{proof}

\begin{lemma}\label{lem:aee-markov-kernel}
	Assume that $(Y, \Sigma_Y)$ is a countably generated $\sigma$-algebra.
	Let $f, g : (X, \Sigma_X) \to \Giry (Y, \Sigma_Y)$ be Markov kernels and $\mu$ be a probability measure on $(X, \Sigma_X)$.
	Then, $f$ and $g$ are $\mu$-almost everywhere equal if and only if for any measurable set $E \in \Sigma_X$ and $F \in \Sigma_Y$, we have
	\[ \Expectation{x}{\mu}{\IversonBracket{x \in E} \cdot f(x)(F)} = \Expectation{x}{\mu}{\IversonBracket{x \in E} \cdot g(x)(F)}. \]
	If $\mathcal{B}$ is a basis of $\Sigma_Y$ (i.e., $\sigma(\mathcal{B}) = \Sigma_Y$), then we can restrict $F$ to be in $\mathcal{B}$.
\end{lemma}
\begin{proof}
	The only-if part is obvious.
	For the if part, let $\{ B_m \}_{m \in \mathbb{N}}$ be a countable basis of $\Sigma_Y$.
	Without loss of generality, we assume that $\{ B_m \}_{m \in \mathbb{N}}$ is a $\pi$-system.
	By Lemma~\ref{lem:aee-function}, for any $m$, $f({-})(B_m)$ and $g({-})(B_m)$ are $\mu$-almost everywhere equal, that is, $S_m \coloneqq \{ x \in X \mid f(x)(B_m) = g(x)(B_m) \}$ satisfies $\mu(S_m) = 1$.
	Now, we show that for any $x \in \bigcap_{m \in \mathbb{N}} S_m$ and any $F \in \Sigma_Y$, we have $f(x)(F) = g(x)(F)$.
	For any $x \in \bigcap_{m \in \mathbb{N}} S_m$, it is easy to see that $\{ F \in \Sigma_Y \mid f(x)(F) = g(x)(F) \}$ is a $\lambda$-system containing $\{ B_m \}_{m \in \mathbb{N}}$.
	By the Dynkin's $\pi$-$\lambda$ theorem, for any $F \in \Sigma_Y$, we have $f(x)(F) = g(x)(F)$.
	The last claim follows because the following set is a $\lambda$-system.
	\[ \{ F \in \Sigma_Y \mid \Expectation{x}{\mu}{\IversonBracket{x \in E} \cdot f(x)(F)} = \Expectation{x}{\mu}{\IversonBracket{x \in E} \cdot g(x)(F)} \} \qedhere \]
\end{proof}

\begin{definition}[disintegration]\label{def:disintegration}
	Let $(A, \Sigma_A, \mu)$ and $(B, \Sigma_B, \lambda)$ be probability spaces, and $f : A \to B$ be a measurable function.
	We say that $\{ \nu_y \}_{y \in B}$ is a \emph{disintegration} of $\mu$ if $y \mapsto \nu_y$ is a Markov kernel from $B$ to $A$ that satisfies the following conditions~\cite{ChangStatisticaNeerlandica1997}:
	\begin{itemize}
		\item For $\lambda$-almost all $y \in B$, we have $\nu_y(\{ x \mid f(x) = y \}) = 1$.
		Equivalently, $\Expectation{y}{\lambda}{\nu_y(\{ x \mid f(x) = y \})} = 1$.
		\item The probability measure defined by $E \mapsto \Expectation{y}{\lambda}{\nu_y(E)}$ coincides with $\mu$.
		Equivalently, $\Expectation{y}{\lambda}{\Expectation{x}{\nu_y}{g(x)}} = \Expectation{x}{\mu}{g(x)}$ for any measurable function $g : A \to [0, \infty]$.
	\end{itemize}
\end{definition}
We often consider the case where $\lambda = f_{*} \mu$, in which case we say that $\{ \nu_y \}_{y \in B}$ is a \emph{$\mu$-disintegration} with respect to $f$.

By the first property of disintegration, we have the following lemma.
\begin{lemma}\label{lem:disintegration-concentration}
	Consider the setting in Definition~\ref{def:disintegration}.
	For any measurable function $g : A \times B \to [0, \infty]$, for $\lambda$-almost all $y \in B$, we have
	\[ \Expectation{x}{\nu_y}{g(x, f(x))} \quad=\quad \Expectation{x}{\nu_y}{g(x, y)}. \]
\end{lemma}
\begin{proof}
	For $\nu_y$-almost all $x \in A$, we have $g(x, f(x)) = g(x, y)$.
\end{proof}

Disintegrations are unique up to almost everywhere equality.
\begin{proposition}
	Consider the setting in Definition~\ref{def:disintegration}.
	If $\{ \nu_y \}_{y \in B}$ and $\{ \nu'_y \}_{y \in B}$ are disintegrations, then $\nu_y$ and $\nu'_y$ are equal for $\lambda$-almost all $y \in B$.
\end{proposition}
\begin{proof}
	By the properties of disintegration, for any measurable set $E \in \Sigma_B$ and $F \in \Sigma_A$, we have the following equality.
	\begin{align}
		&\Expectation{y}{\lambda}{\IversonBracket{y \in E} \cdot \nu_y(F)} \\
		&= \Expectation{x}{\mu}{\IversonBracket{f(x) \in E} \cdot \IversonBracket{x \in F}} \\
		&= \Expectation{y}{\lambda}{\IversonBracket{y \in E} \cdot \nu'_y(F)}
	\end{align}
	By Lemma~\ref{lem:aee-markov-kernel}, $\nu_y$ and $\nu'_y$ are equal for $\lambda$-almost all $y \in B$.
\end{proof}

\end{toappendix}

\section{Noninterference and Conditional Independence}
\label{sec:main-result}

In this section, we present our main result that characterizes noninterference as conditional independence expressed by Lilac Hoare triples.
Let $\Delta \vdash M : \mathtt{G} B$ be an APPL program where $\Delta = X_1 : A_1, \dots, X_m : A_m$ and $B = B_1 \times \dots \times B_n$.
Let $\mathcal{I}$ and $\mathcal{O}$ be subsets of input and output variables, respectively; i.e., $\mathcal{I} \subseteq \{ X_1, \dots, X_m \}$ and $\mathcal{O} \subseteq \{ Y_1, \dots, Y_n \}$ where $Y_1, \dots, Y_n$ are output variables of $M$.
We show that $(\mathcal{I}, \mathcal{O})$-NI is equivalent to the following Hoare triple.
\[ \HoareTriple{\Own \Delta}{M}{\mathbf{Y}. \Conditioning_{x \gets \mathcal{I}} (\Own (\Delta \setminus \mathcal{I}) * \Own \mathcal{O})} \]
Here, we identify $\Delta$ with the tuple of all input variables $(X_1, \dots, X_m)$.
Without loss of generality, we assume that input and output variables are $\{ X_1, X_2 \}$ and $\{ Y_1, Y_2 \}$, respectively, and consider $(\{ X_1 \}, \{ Y_1 \})$-NI.

\begin{theorem}\label{thm:noninterference-conditional-independence}
	Let $X_1 : A_1, X_2 : A_2 \vdash M : \mathtt{G} (B_1 \times B_2)$ be an APPL program.
	We have
	\begin{equation}
		\HoareTriple{\Own (X_1, X_2)}{M}{(Y_1, Y_2). \Conditioning_{x \gets X_1} (\Own X_2 * \Own Y_1)}
		\label{eq:hoare-noninterference}
	\end{equation}
	if and only if there exists a simulator $X_1 : A_1 \vdash \mathbf{Sim}(M) : \mathtt{G} B_1$ such that the following equation holds.
	\begin{equation}
		X_1 : A_1, X_2 : A_2 \vdash \mathbf{Sim}(M) = (Y_1, Y_2) \gets M; \ret{Y_1} : \mathtt{G} B_1
		\label{eq:simulator}
	\end{equation}
\end{theorem}

We give a proof sketch of Theorem~\ref{thm:noninterference-conditional-independence}.
We write $\Prob{Y | X}$ for the conditional probability distribution of $Y$ given $X$.
\begin{proofsketch}
	($\implies$):
	By the Hoare triple for $M$, the input $(X_1, X_2)$ and the output $(Y_1, Y_2)$ satisfy $X_2 \perp Y_1 \mid X_1$, which implies $\Prob{Y_1 | X_1, X_2} = \Prob{Y_1 | X_1}$.
	Since $\Prob{Y_1 | X_1, X_2}$ is the output probability distribution of $(Y_1, Y_2) \gets M; \ret{Y_1}$ for given inputs $X_1$ and $X_2$, the equation $\Prob{Y_1 | X_1, X_2} = \Prob{Y_1 | X_1}$ implies that $(Y_1, Y_2) \gets M; \ret{Y_1}$ does not actually depend on $X_2$.
	Therefore, we define $\mathbf{Sim}(M)$ as follows where $v : A_2$ is an arbitrary value.
	\begin{equation}
		\mathbf{Sim}(M) \quad\coloneqq\quad X_2 \gets \ret{v}; (Y_1, Y_2) \gets M; \ret{Y_1} \label{eq:simulator-define}
	\end{equation}

	($\impliedby$):
	If a simulator $\mathbf{Sim}(M)$ satisfying~\eqref{eq:simulator} exists, then we have $X_2 \perp Y_1 \mid X_1$ because $\mathbf{Sim}(M)$ does not depend on $X_2$.
	By~\eqref{eq:simulator}, the output distribution of $(Y_1, Y_2) \gets M; \ret{Y_1}$ also satisfies the same conditional independence.
	Therefore, we have the desired Hoare triple for $M$.
\end{proofsketch}
The idea of the proof itself is simple, but its formal proof requires several technical lemmas on disintegrations and the semantics of Lilac.
Details are given in
\iflong
Appendix~\ref{sec:proof-main-result}.
\fi

\begin{toappendix}
\section{Proof of Theorem~\ref{thm:noninterference-conditional-independence}}
\label{sec:proof-main-result}

To prove Theorem~\ref{thm:noninterference-conditional-independence}, we prepare several lemmas.

The following lemma relates the conditional independence formula $\Conditioning_{x \gets X} (\Own Y * \Own Z)$ and the conditional independence in terms of product probability measures.
\begin{lemma}\label{lem:conditional-independence-formula}
	Let $\emptyctx; X : A, Y : B, Z : C \vdash \Conditioning_{x \gets X} (\Own Y * \Own Z)$ be a Lilac formula.
	We have $\gamma, D, (\mathcal{F}, \mu) \vDash \Conditioning_{x \gets X} (\Own Y * \Own Z)$ if and only if for any $(\Sigma_{\Omega}, \mu') \sqsupseteq (\mathcal{F}, \mu)$ and $\mu'$-disintegration $\{ \nu_x \}_x$ with respect to $X$, for almost all $x$, we have $Y^{-1}(\Sigma_B) \cup Z^{-1}(\Sigma_C) \subseteq \mathcal{F}$ and $Y_{*} \nu_x \otimes Z_{*} \nu_x = (Y, Z)_{*} \nu_x$.
	Here, $\Sigma_B$ and $\Sigma_C$ is the $\sigma$-algebra on $\interpret{B}$ and $\interpret{C}$, respectively; and $\otimes$ is the product probability measure.
\end{lemma}
\begin{proof}
	By definition of $\gamma, D, (\mathcal{F}, \mu) \vDash \Conditioning_{x \gets X} (\Own Y * \Own Z)$, it suffices to show that for any $(\Sigma_{\Omega}, \mu') \sqsupseteq (\mathcal{F}, \mu)$, any $\mu'$-disintegration $\{ \nu_x \}_x$ with respect to $X$, and any $x$, we have
	\[ \gamma, D, (\mathcal{F}, \nu_x|_{\mathcal{F}}) \vDash \Own Y * \Own Z \quad\iff\quad Y^{-1}(\Sigma_B) \cup Z^{-1}(\Sigma_C) \subseteq \mathcal{F} \text{ and } Y_{*} \nu_x \otimes Z_{*} \nu_x = (Y, Z)_{*} \nu_x \]
	This is proved as follows.
	\begin{align}
		&\gamma, D, (\mathcal{F}, \nu_x|_{\mathcal{F}}) \vDash \Own Y * \Own Z \\
		&\iff \exists \mathcal{P}_1, \mathcal{P}_2 \text{ such that } \mathcal{P}_1 \cdot \mathcal{P}_2 \sqsubseteq (\mathcal{F}, \nu_x|_{\mathcal{F}}); \gamma, D, \mathcal{P}_1 \vDash \Own Y; \text{ and } \gamma, D, \mathcal{P}_2 \vDash \Own Z \\
		&\iff (Y^{-1}(\Sigma_B), \nu_x|_{Y^{-1}(\Sigma_B)}) \cdot (Z^{-1}(\Sigma_C), \nu_x|_{Z^{-1}(\Sigma_C)}) \sqsubseteq (\mathcal{F}, \nu_x|_{\mathcal{F}}) \label{eq:lem:conditional-independence-formula-1} \\
		&\iff Y^{-1}(\Sigma_B) \cup Z^{-1}(\Sigma_C) \subseteq \mathcal{F} \text{ and } Y_{*} \nu_x \otimes Z_{*} \nu_x = (Y, Z)_{*} \nu_x \label{eq:lem:conditional-independence-formula-2}
	\end{align}
	Here,~\eqref{eq:lem:conditional-independence-formula-1} and \eqref{eq:lem:conditional-independence-formula-2} follow from the following facts:
	\begin{itemize}
		\item \eqref{eq:lem:conditional-independence-formula-1}: Since $\Sigma_B$ and $\Sigma_C$ are countably generated, so are $Y^{-1}(\Sigma_B)$ and $Z^{-1}(\Sigma_C)$.
		\item \eqref{eq:lem:conditional-independence-formula-2}: By definition of product measurable spaces, we have $Y_{*} \nu_x \otimes Z_{*} \nu_x = (Y, Z)_{*} \nu_x$ if and only if for any measurable sets $E \in \Sigma_B$ and $F \in \Sigma_C$, we have $(Y_{*} \nu_x \otimes Z_{*} \nu_x)(E \times F) = (Y, Z)_{*} \nu_x (E \times F)$.
		Note here that we have the following equations.
		\begin{align}
			(Y_{*} \nu_x \otimes Z_{*} \nu_x)(E \times F) &= (Y_{*} \nu_x)(E) \cdot (Z_{*} \nu_x)(F) \\
			&= \nu_x(Y^{-1}(E)) \cdot \nu_x(Z^{-1}(F))\\
			(Y, Z)_{*} \nu_x (E \times F) &= \nu_x(Y^{-1}(E) \cap Z^{-1}(F))
			\qedhere
		\end{align}
	\end{itemize}
\end{proof}

In the semantics of conditioning modalities in Lilac, one needs to take a probability distribution over the common sample space $\Omega$.
This requires assigning an arbitrary probability distribution to the region outside the domain of interest, but the following lemma ensures that this choice does not essentially affect the semantics.
\begin{lemma}\label{lem:disintegration-equivalent}
	Let $\mu, \mu'$ be probability measures on $\Sigma_{\Omega}$.
	Let $D = (D_X, D_Y)$ and $D' = (D'_X, D'_Y)$ be random variables $D, D' \in \RV{\interpret{\Delta}}$ where $\Delta = X : A, Y : B$.
	If $D_{*} \mu = D'_{*} \mu'$, then for any $\mu$-disintegration $\{ \nu_x \}_{x \in A}$ with respect to $D_X$ and any $\mu'$-disintegration $\{ \nu'_x \}_{x \in A}$ with respect to $D'_X$, for $D_{*} \mu$-almost all $x$ and for any $E \in \Sigma_B$, we have $\nu_x(D_Y^{-1}(E)) = \nu'_x(D'^{-1}_Y(E))$.
\end{lemma}
\begin{proof}
	By Lemma~\ref{lem:aee-markov-kernel}, it suffices to show that for any measurable set $E \in \Sigma_B$ and $F \in \Sigma_A$, we have
	\[ \Expectation{x}{(D_X)_{*} \mu}{\IversonBracket{x \in F} \cdot \nu_x(D_Y^{-1}(E))} = \Expectation{x}{(D'_X)_{*} \mu'}{\IversonBracket{x \in F} \cdot \nu'_x({D'}_Y^{-1}(E))}. \]
	By $D_{*} \mu = D'_{*} \mu'$, it suffices to prove
	\[ \Expectation{x}{(D_X)_{*} \mu}{\IversonBracket{x \in F} \cdot \nu_x(D_Y^{-1}(E))} = \mathbb{P}_{(x, y) \sim D_{*} \mu}\IversonBracket{x \in F \land y \in E}. \]
	This is proved as follows.
	\begin{align*}
		&\mathbb{P}_{(x, y) \sim D_{*} \mu}[x \in F \land y \in E] \\
		&= \mathbb{P}_{\omega \sim \mu}[D_X(\omega) \in F \land D_Y(\omega) \in E] \\
		&= \Expectation{x}{(D_X)_{*} \mu}{\Expectation{\omega}{\nu_x}{\IversonBracket{D_X(\omega) \in F} \cdot \IversonBracket{D_Y(\omega) \in E}}} \\
		&= \Expectation{x}{(D_X)_{*} \mu}{\IversonBracket{x \in F} \cdot \nu_x(D_Y^{-1}(E))}
		\qedhere
	\end{align*}
\end{proof}

By Lemma~\ref{lem:disintegration-equivalent}, we have the following lemma.
\begin{lemma}\label{lem:conditional-independence-equivalent}
	Let $\Delta = X : A, Y : B, Z : C$ and $D = (D_X, D_Y, D_Z), D' = (D'_X, D'_Y, D'_Z)$ be random variables such that $(D_X, D_Y, D_Z), (D'_X, D'_Y, D'_Z) \in \RV{\interpret{\Delta}}$.
	Let $\mu, \mu'$ be probability measures on $\Sigma_{\Omega}$ such that $D_{*} \mu = D'_{*} \mu'$.
	Let $\mathcal{F}, \mathcal{F}' \subseteq \Sigma_{\Omega}$ be $\sigma$-algebras such that $(D_Y, D_Z)$ is $\mathcal{F}$-measurable and $(D'_Y, D'_Z)$ is $\mathcal{F}'$-measurable.
	Then, we have
	\begin{align}
		&\gamma, D, (\mathcal{F}, \mu|_{\mathcal{F}}) \vDash \Conditioning_{x \gets X} (\Own Y * \Own Z) \\
		&\iff\quad \gamma, D', (\mathcal{F}', \mu'|_{\mathcal{F}'}) \vDash \Conditioning_{x \gets X} (\Own Y * \Own Z)
	\end{align}
\end{lemma}
\begin{proof}
	By symmetry, it suffices to prove that if $\gamma, D, \mathcal{P} \vDash \Conditioning_{x \gets X} (\Own Y * \Own Z)$, then $\gamma, D', \mathcal{P}' \vDash \Conditioning_{x \gets X} (\Own Y * \Own Z)$.
	By unfolding the definition, we have $\gamma, D', \mathcal{P}' \vDash \Conditioning_{x \gets X} (\Own Y * \Own Z)$ if and only if for any $\mu'$-disintegration $\{ \nu'_x \}_x$ with respect to $D'_X$, we have for $D'_{*} \mu'$-almost all $x$, and for any $E \in \Sigma_B$ and $F \in \Sigma_C$, we have ${D'}_Y^{-1}(E) \cap {D'}_Z^{-1}(F) \in \mathcal{F}'$ and $\nu'_x({D'}_Y^{-1}(E) \cap {D'}_Z^{-1}(F)) = \nu'_x({D'}_Y^{-1}(E)) \cdot \nu'_x({D'}_Z^{-1}(F))$.
	Fix a $\mu$-disintegration $\{ \nu_x \}_x$ with respect to $D_X$.
	By $\gamma, D, \mathcal{P} \vDash \Conditioning_{x \gets X} (\Own Y * \Own Z)$ and Lemma~\ref{lem:disintegration-equivalent}, we have the following equation for almost all $x$.
	\begin{align*}
		\nu'_x({D'}_Y^{-1}(E) \cap {D'}_Z^{-1}(F)) &= \nu'_x((D'_Y, D'_Z)^{-1}(E \times F)) \\
		&= \nu_x((D_Y, D_Z)^{-1}(E \times F)) \\
		&= \nu_x(D_Y^{-1}(E) \cap D_Z^{-1}(F)) \\
		&= \nu_x(D_Y^{-1}(E)) \cdot \nu_x(D_Z^{-1}(F)) \\
		&= \nu'_x(D'^{-1}_Y(E)) \cdot \nu'_x(D'^{-1}_Z(F))
	\end{align*}
	Therefore, we have $\gamma, D', \mathcal{P}' \vDash \Conditioning_{x \gets X} (\Own Y * \Own Z)$.
\end{proof}

\begin{lemma}\label{lem:own-variables}
	For any APPL program $\Delta \vdash M : \mathtt{G} A$, we have
	\[ \HoareTriple{\Own \Delta}{M}{X. \Own \Delta \land \Own X} \]
	where $\Own \Delta$ is shorthand for $\Own (X_1, \dots, X_n)$ when $\Delta = X_1 : A_1, \dots, X_n : A_n$.
\end{lemma}
\begin{proof}
	By induction on the structure of $M$.
	\begin{itemize}
		\item $\ret{E}$: Since $\emptyset; \Delta \vdash E : A$, we have $\Own \Delta \vdash \Own E$ by \hyperlink{rule:FV-Own}{\textsc{FV-Own}}.
		\begin{mathpar}
			\inferrule*[Left=H-Conseq]{
				\inferrule*[Left=H-Ret]{ }{
					\HoareTriple{\Own \Delta \land \Own E}{\ret{E}}{X. \Own \Delta \land \Own X}
				}
			}{
				\HoareTriple{\Own \Delta}{\ret{E}}{X. \Own \Delta \land \Own X}
			}
		\end{mathpar}
		\item $X \gets M; N$:
		\begin{mathpar}
			\inferrule*[Left=H-Let]{
				\HoareTriple{\Own \Delta}{M}{X. \Own \Delta \land \Own X} \\
				\inferrule*[Left=H-Conseq]{
					\HoareTriple{\Own \Delta \land \Own X}{N}{Y. \Own \Delta \land \Own X \land \Own Y}
				}{
					\HoareTriple{\Own \Delta \land \Own X}{N}{Y. \Own \Delta \land \Own Y}
				}
			}{
				\HoareTriple{\Own \Delta}{X \gets M; N}{Y. \Own \Delta \land \Own Y}
			}
		\end{mathpar}
		\item $\mathtt{unif}\ \mathbb{K}$:
		\begin{mathpar}
			\inferrule*[Left=H-Conseq]{
				\inferrule*[Left=H-Frame]{
					\inferrule*[Left=H-Unif]{ }{
						\HoareTriple{\top}{\mathtt{unif}\ \mathbb{K}}{X. X \sim \Uniform{\mathbb{K}}}
					}
				}{
					\HoareTriple{\Own \Delta}{\mathtt{unif}\ \mathbb{K}}{X. \Own \Delta * X \sim \Uniform{\mathbb{K}}}
				}
			}{
				\HoareTriple{\Own \Delta}{\mathtt{unif}\ \mathbb{K}}{X. \Own \Delta \land \Own X}
			}
		\end{mathpar}
		\item $\forloop{i}{n_s}{n_e}{X}{E}{M}$: Let $I(i, X) \coloneqq \Own \Delta \land \Own X$.
		\begin{mathpar}
			\inferrule*[Left=H-Conseq]{
			\inferrule*[Left=H-For]{
				\inferrule*[Left=H-Conseq]{
					\HoareTriple{\Own \Delta \land \Own X}{M}{X'. \Own \Delta \land \Own X \land \Own X'}
				}{
					\HoareTriple{\Own \Delta \land \Own X}{M}{X'. \Own \Delta \land \Own X'}
				}
			}{
				\HoareTriple{\Own \Delta \land \Own E}{\forloop{i}{n_s}{n_e}{X}{E}{M}}{X. \Own \Delta \land \Own X}
			}
			}{
				\HoareTriple{\Own \Delta}{\forloop{i}{n_s}{n_e}{X}{E}{M}}{X. \Own \Delta \land \Own X}
			}
		\end{mathpar}
		Here, we use \hyperlink{rule:FV-Own}{\textsc{FV-Own}} to show $\Own \Delta \vdash \Own E$.
		\qedhere
	\end{itemize}
\end{proof}

To prove the ($\implies$) direction of Theorem~\ref{thm:noninterference-conditional-independence}, we prepare several lemmas.

In classical probability theory, $X$ and $Y$ are independent if and only if $\Prob{Y \mid X} = \Prob{Y}$.
The following lemma states a similar characterization for the Lilac formula $\Own X * \Own Y$ in terms of disintegrations.
\begin{lemma}\label{lem:independence-disintegration}
	Let $X : A, Y : B \vdash \Own X * \Own Y$ be a Lilac formula.
	We have $\gamma, D, (\mathcal{F}, \mu) \vDash \Own X * \Own Y$ if and only if $X$ and $Y$ are $\mathcal{F}$-measurable, and for any $(\Sigma_{\Omega}, \mu') \sqsupseteq (\mathcal{F}, \mu)$ and $\mu'$-disintegration $\{ \nu_x \}_{x \in \interpret{A}}$ with respect to $X$, we have $Y_{*} \nu_x = Y_{*} \mu$ for almost all $x \in \interpret{A}$.
	\qed
\end{lemma}
\begin{proof}[Proof of Lemma~\ref{lem:independence-disintegration}]
	The only if direction follows from the proof of \hyperlink{rule:C-Indep}{\textsc{C-Indep}} in \cite[Lemma~B.27]{LiPLDI2023}.
	To prove the if direction, it suffices to show
	\[ \mu(X^{-1}(E) \cap Y^{-1}(F)) \quad=\quad \mu(X^{-1}(E)) \cdot \mu(Y^{-1}(F)) \]
	for any measurable sets $E \in \Sigma_{\interpret{A}}$ and $F \in \Sigma_{\interpret{B}}$.
	This is proved as follows.
	\begin{align}
		&\mu(X^{-1}(E) \cap Y^{-1}(F)) \\
		&= \Expectation{x}{X_{*} \mu}{\Expectation{\omega}{\nu_x}{\IversonBracket{X(\omega) \in E} \cdot \IversonBracket{Y(\omega) \in F}}} \\
		&= \Expectation{x}{X_{*} \mu}{\IversonBracket{x \in E} \cdot \Expectation{\omega}{\nu_x}{\IversonBracket{Y(\omega) \in F}}} \\
		&= \Expectation{x}{X_{*} \mu}{\IversonBracket{x \in E} \cdot Y_{*} \nu_x (F)} \\
		&= \Expectation{x}{X_{*} \mu}{\IversonBracket{x \in E} \cdot Y_{*} \mu (F)} \\
		&= Y_{*} \mu (F) \cdot \Expectation{x}{X_{*} \mu}{\IversonBracket{x \in E}} \\
		&= \mu(Y^{-1}(F)) \cdot \mu(X^{-1}(E))
		\qedhere
	\end{align}
\end{proof}

Let $M$ be a program with input variables $X$ and output variables $Y$.
The following lemma states that any disintegration that represents the conditional distribution $\Prob{Y \mid X}$ is almost everywhere equal to the distribution of the output of $M$ given the input $X$.
\begin{lemma}\label{lem:program-disintegration}
	Let $\Delta \vdash M : \mathtt{G} A$ be an APPL program, $D \in \RV{\interpret{\Delta}}$ be a random variable, and $(\Sigma_{\Omega}, \mu) \in \KRM$.
	If a random variable $Y \in \RV{\interpret{A}}$ and $(\Sigma_{\Omega}, \mu_M) \in \KRM$ satisfy the following equation:
	\[ \left( \begin{aligned}
		&\omega \gets \mu; \\
		&y \gets \interpret{M}(D(\omega)); \\
		&\ret{(D(\omega), y)}
	\end{aligned} \right) \quad=\quad \left( \begin{aligned}
		&\omega \gets \mu_{M}; \\
		&\ret{(D(\omega), Y(\omega))}
	\end{aligned} \right) \]
	then, for any $\mu_M$-disintegration $\{ \nu_x \}_{x \in \interpret{\Delta}}$ with respect to $D$ and for almost all $x \in \interpret{\Delta}$, we have
	\[ Y_{*} \nu_x \quad=\quad \interpret{M}(x). \tag*{\qed} \]
\end{lemma}
\begin{proof}[Proof of Lemma~\ref{lem:program-disintegration}]
	We apply Lemma~\ref{lem:aee-markov-kernel} to $Y_{*} \nu_x$ and $\interpret{M}(x)$.
	Note first that we have $D_{*} \mu_M = D_{*} \mu$.
	We show $\Expectation{x}{D_{*} \mu_M}{\IversonBracket{x \in E} \cdot \nu_x(Y^{-1}(F))} = \Expectation{x}{D_{*} \mu}{\IversonBracket{x \in E} \cdot \interpret{M}(x)(F)}$ for any measurable sets $E \in \Sigma_A$ and $F \in \Sigma_B$.
	This is proved as follows:
	\begin{align*}
		&\Expectation{x}{D_{*} \mu_M}{\IversonBracket{x \in E} \cdot \nu_x(Y^{-1}(F))} \\
		&= \Expectation{x}{D_{*} \mu_M}{\IversonBracket{x \in E} \cdot \Expectation{\omega}{\nu_x}{\IversonBracket{Y(\omega) \in F}}} \\
		&= \Expectation{x}{D_{*} \mu_M}{\Expectation{\omega}{\nu_x}{\IversonBracket{x \in E} \cdot \IversonBracket{Y(\omega) \in F}}} \\
		&= \Expectation{x}{D_{*} \mu_M}{\Expectation{\omega}{\nu_x}{\IversonBracket{D(\omega) \in E} \cdot \IversonBracket{Y(\omega) \in F}}} \\
		&= \Expectation{\omega}{\mu_M}{\IversonBracket{D(\omega) \in E} \cdot \IversonBracket{Y(\omega) \in F}} \\
		&= \Expectation{(x, y)}{(D, Y)_{*} \mu_M}{\IversonBracket{x \in E} \cdot \IversonBracket{y \in F}} \\
		&= \Expectation{x}{D_{*} \mu}{\Expectation{y}{\interpret{M}(x)}{\IversonBracket{x \in E} \cdot \IversonBracket{y \in F}}} \\
		&= \Expectation{x}{D_{*} \mu}{\IversonBracket{x \in E} \cdot \Expectation{y}{\interpret{M}(x)}{\IversonBracket{y \in F}}} \\
		&= \Expectation{x}{D_{*} \mu}{\IversonBracket{x \in E} \cdot \interpret{M}(x)(F)}
		\qedhere
	\end{align*}
\end{proof}

Given a pair of random variables $(X, Y)$, we can consider two types of disintegrations: one is \emph{joint} disintegrations with respect to the pair $(X, Y)$ and the other is \emph{iterated} disintegrations with respect to $X$ and then $Y$.
We consider the relationship between these two types of disintegrations.
\begin{lemma}\label{lem:disintegration-pushforward-of-pair}
	Let $(\Omega, \Sigma_{\Omega}, \mu)$ be a probability space, and $X : \Omega \to A$ and $Y : \Omega \to B$ be random variables.
	Suppose that $\{ \nu_x \}_{x \in A}$ is a $\mu$-disintegration with respect to $X$.
	For any $g : A \times B \to \mathbb{R}$, we have the following equality.
	\[ \Expectation{x}{X_{*} \mu}{\Expectation{y}{Y_{*} \nu_x}{g(x, y)}} \quad=\quad \Expectation{(x, y)}{(X, Y)_{*} \mu}{g(x, y)} \]
\end{lemma}
\begin{proof}
	\begin{align*}
		&\Expectation{x}{X_{*} \mu}{\Expectation{y}{Y_{*} \nu_x}{g(x, y)}} \\
		&= \Expectation{x}{X_{*} \mu}{\Expectation{\omega}{\nu_x}{g(x, Y(\omega))}} \\
		&= \Expectation{x}{X_{*} \mu}{\Expectation{\omega}{\nu_x}{g(X(\omega), Y(\omega))}} &\text{by Lemma~\ref{lem:disintegration-concentration}} \\
		&= \Expectation{\omega}{\mu}{g(X(\omega), Y(\omega))} &\text{by Definition~\ref{def:disintegration}} \\
		&= \Expectation{(x, y)}{(X, Y)_{*} \mu}{g(x, y)}
	\end{align*}
\end{proof}

\begin{lemma}\label{lem:disintegration-pair-1}
	Let $(\Omega, \Sigma_{\Omega}, \mu)$ be a probability space, and $X : \Omega \to A$ and $Y : \Omega \to B$ be random variables.
	Suppose that $\{ \nu_x \}_{x \in A}$ is a $\mu$-disintegration with respect to $X$, and $\{ (\nu_x)_y \}_{y \in B}$ is a $\nu_x$-disintegration with respect to $Y$.
	If $(x, y) \mapsto (\nu_x)_y$ is a Markov kernel, i.e., $(x, y) \mapsto (\nu_x)_y(E)$ is measurable for any measurable set $E \in \Sigma_{\Omega}$, then $\{ (\nu_x)_y \}_{(x, y) \in A \times B}$ is a $\mu$-disintegration with respect to $(X, Y)$.
\end{lemma}
\begin{proof}
	\begin{itemize}
		\item For $(X, Y)_{*} \mu$-almost all $(x, y)$, it follows that
		\[ \Expectation{\omega}{(\nu_x)_y}{\IversonBracket{(X(\omega), Y(\omega)) = (x, y)}} = 1 \]
		holds because
		\begin{align*}
			&\Expectation{(x, y)}{(X, Y)_{*} \mu}{\Expectation{\omega}{(\nu_x)_y}{\IversonBracket{(X(\omega), Y(\omega)) = (x, y)}}} \\
			&= \Expectation{x}{X_{*} \mu}{\Expectation{y}{Y_{*} \nu_x}{\Expectation{\omega}{(\nu_x)_y}{\IversonBracket{(X(\omega), Y(\omega)) = (x, y)}}}} & \text{by Lemma~\ref{lem:disintegration-pushforward-of-pair}} \\
			&= \Expectation{x}{X_{*} \mu}{\Expectation{y}{Y_{*} \nu_x}{\Expectation{\omega}{(\nu_x)_y}{\IversonBracket{(X(\omega), y) = (x, y)}}}} & \text{by Lemma~\ref{lem:disintegration-concentration}} \\
			&= \Expectation{x}{X_{*} \mu}{\Expectation{y}{Y_{*} \nu_x}{\Expectation{\omega}{(\nu_x)_y}{\IversonBracket{X(\omega) = x}}}} \\
			&= \Expectation{x}{X_{*} \mu}{\Expectation{\omega}{\nu_x}{\IversonBracket{X(\omega) = x}}} \\
			&= 1
		\end{align*}
		\item We show
		\[ \Expectation{\omega}{\mu}{f(\omega)} = \Expectation{(x, y)}{(X, Y)_{*} \mu}{\Expectation{\omega}{(\nu_x)_y}{f(\omega)}} \]
		for any measurable function $f : \Omega \to [0, \infty]$.
		\begin{align}
			&\Expectation{\omega}{\mu}{f(\omega)} \\
			&= \Expectation{x}{X_{*} \mu}{\Expectation{\omega}{\nu_x}{f(\omega)}} \\
			&= \Expectation{x}{X_{*} \mu}{\Expectation{y}{Y_{*} \nu_x}{\Expectation{\omega}{(\nu_x)_y}{f(\omega)}}} \\
			&= \Expectation{(x, y)}{(X, Y)_{*} \mu}{\Expectation{\omega}{(\nu_x)_y}{f(\omega)}} &\qedhere
		\end{align}
	\end{itemize}
\end{proof}

The following lemma states that a joint disintegration can be seen as an iterated disintegration.
\begin{lemma}\label{lem:disintegration-pair-2}
	Let $(\Omega, \Sigma_{\Omega}, \mu)$ be a probability space with countably generated $\sigma$-algebra, and $X : \Omega \to A$ and $Y : \Omega \to B$ be random variables.
	Suppose that $\{ \nu_{x, y} \}_{(x, y) \in A \times B}$ is a $\mu$-disintegration with respect to $(X, Y)$, and $\{ \nu_x \}_{x \in A}$ is a $\mu$-disintegration with respect to $X$.
	Then $\{ \nu_{x, y} \}_{y \in B}$ is a $\nu_x$-disintegration with respect to $Y$ for $X_{*} \mu$-almost all $x$.
	\qed
\end{lemma}
\begin{proof}[Proof of Lemma~\ref{lem:disintegration-pair-2}]
	\begin{itemize}
		\item We show $\Expectation{y}{Y_{*} \nu_x}{\Expectation{\omega}{\nu_{x, y}}{\IversonBracket{Y(\omega) = y}}} = 1$ for almost all $x$.
		We use Lemma~\ref{lem:disintegration-pushforward-of-pair}.
		\begin{align}
			&\Expectation{x}{X_{*} \mu}{\Expectation{y}{Y_{*} \nu_x}{\Expectation{\omega}{\nu_{x, y}}{\IversonBracket{Y(\omega) = y}}}} \\
			&= \Expectation{(x, y)}{(X, Y)_{*} \mu}{\Expectation{\omega}{\nu_{x, y}}{\IversonBracket{Y(\omega) = y}}} \\
			&= 1
		\end{align}
		\item Using Lemma~\ref{lem:aee-markov-kernel}, we show that $F \mapsto \Expectation{y}{Y_{*} \nu_x}{\nu_{x, y}(F)}$ and $\nu_x$ are equal for almost all $x$.
		\begin{align}
			&\Expectation{x}{X_{*} \mu}{\IversonBracket{x \in E} \cdot \Expectation{y}{Y_{*} \nu_x}{\nu_{x, y}(F)}} \\
			&\Expectation{x}{X_{*} \mu}{\Expectation{y}{Y_{*} \nu_x}{\IversonBracket{x \in E} \cdot \nu_{x, y}(F)}} \\
			&= \Expectation{(x, y)}{(X, Y)_{*} \mu}{\IversonBracket{x \in E} \cdot \nu_{x, y}(F)} \\
			&= \Expectation{(x, y)}{(X, Y)_{*} \mu}{\Expectation{\omega}{\nu_{x, y}}{\IversonBracket{X(\omega) \in E} \cdot \IversonBracket{\omega \in F}}} \\
			&= \Expectation{\omega}{\mu}{\IversonBracket{X(\omega) \in E} \cdot \IversonBracket{\omega \in F}} \\
			&= \Expectation{x}{X_{*} \mu}{\Expectation{\omega}{\nu_x}{\IversonBracket{X(\omega) \in E} \cdot \IversonBracket{\omega \in F}}} \\
			&= \Expectation{x}{X_{*} \mu}{\IversonBracket{x \in E} \cdot \nu_x(F)}
		\end{align}
	\end{itemize}
\end{proof}

\begin{proof}[Proof of Theorem~\ref{thm:noninterference-conditional-independence} ($\implies$)]
	To show that the simulator~\eqref{eq:simulator-define} satisfies~\eqref{eq:simulator}, it suffices to show that
	\[ X_1 : A_1, X_2 : A_2 \vdash (Y_1, Y_2) \gets M; \ret{Y_1} : \mathtt{G} B_1 \]
	does not actually depend on $X_2$.
	We prove by contradiction.
	Let $M' \coloneqq (Y_1, Y_2) \gets M; \ret{Y_1}$.
	Suppose that there exist $a_1 \in \interpret{A_1}$ and $a_2, a'_2 \in \interpret{A_2}$ such that $\interpret{M'}(a_1, a_2) \neq \interpret{M'}(a_1, a'_2)$.
	By assumption and the rules in Fig.~\ref{fig:lilac-hoare-rules}, we have the following Hoare triple:
	\begin{align}
		\HoareTriple{\Own X_1 \land \Own X_2}{M'}{Y_1. \Conditioning_{x \gets X_1} (\Own X_2 * \Own Y_1)}
	\end{align}
	Suppose that the input distribution is chosen so that $(X_1, X_2)$ is uniformly distributed over $\{ (a_1, a_2), (a_1, a'_2) \}$.
	By the above Hoare triple, the output $Y_1$ satisfies $\Conditioning_{x \gets X_1} (\Own X_2 * \Own Y_1)$.
	Formally, let $\mu$ be a probability measure on $\Sigma_\Omega$ and $D \in \RV{\interpret{X_1 : A_1, X_2 : A_2}}$ be such that $D_{*} \mu$ is the uniform distribution over $\{ (a_1, a_2), (a_1, a'_2) \}$.
	Then, there exists $Y_1 \in \RV{\interpret{B_1}}$, $\mathcal{P}'$, and $\mu'$ with $\mathcal{P}' \sqsubseteq (\Sigma_{\Omega}, \mu')$ such that
	\[ \left( \begin{aligned}
		&\omega \gets \mu; \\
		&(y_1, y_2) \gets \interpret{M}(D(\omega)); \\
		&\ret{(D(\omega), y_1)}
	\end{aligned} \right) \quad=\quad \left( \begin{aligned}
		&\omega \gets \mu'; \\
		&\ret{(D(\omega), Y_1(\omega))}
	\end{aligned} \right) \]
	and $\gamma, (D, Y_1), \mathcal{P}' \vDash \Conditioning_{x \gets X_1} (\Own X_2 * \Own Y_1)$.
	Let $\{ \nu_{x_1} \}_{x_1 \in \interpret{A_1}}$ be a $\mu'$-disintegration with respect to $X_1$.
	Since $(X_1)_{*} \mu' = (X_1)_{*} \mu = \delta_{a_1}$, we have $\mu' = \nu_{a_1}$ and thus $\gamma, (D, Y_1), \mathcal{P}' \vDash \Own X_2 * \Own Y_1$ by the semantics of conditioning modalities.
	\[ \mu'(E) = \Expectation{x}{(X_1)_{*} \mu'}{\Expectation{\omega}{\nu_x}{\IversonBracket{\omega \in E}}} = \nu_{a_1}(E) \]
	Let $\{ \nu_{x_1, x_2} \}_{x_1 \in \interpret{A_1}, x_2 \in \interpret{A_2}}$ be a $\mu'$-disintegration with respect to $(X_1, X_2)$.
	By Lemma~\ref{lem:disintegration-pair-2}, $\{ \nu_{x_1, x_2} \}_{x_2 \in \interpret{A_2}}$ is a $\nu_{x_1}$-disintegration with respect to $X_2$ for almost all $x_1$, i.e., $\{ \nu_{a_1, x_2} \}_{x_2 \in \interpret{A_2}}$ is a $\mu'$-disintegration with respect to $X_2$.
	By applying Lemma~\ref{lem:independence-disintegration} to $\gamma, (D, Y_1), \mathcal{P}' \vDash \Own X_2 * \Own Y_1$, we have $(Y_1)_{*} \nu_{a_1, a_2} = (Y_1)_{*} \mu' = (Y_1)_{*} \nu_{a_1, a'_2}$.
	By Lemma~\ref{lem:program-disintegration}, we have $(Y_1)_{*} \nu_{x_1, x_2} = \interpret{M}(x_1, x_2)$ for almost all $(x_1, x_2)$, i.e., for any $(x_1, x_2) \in \{ (a_1, a_2), (a_1, a'_2) \}$.
	Therefore, we have $\interpret{M'}(a_1, a_2) = \interpret{M'}(a_1, a'_2)$, which is a contradiction.
\end{proof}

For the ($\impliedby$) direction of Theorem~\ref{thm:noninterference-conditional-independence}, the key lemma is the following.
\begin{lemma}\label{lem:extra-variable-conditional-independence}
	Let $\Delta = X_1 : A_1, \dots, X_n : A_n$ be a context, $\Delta \vdash M : \mathtt{G} B$ be an APPL program, and $Z \in \RV{\interpret{C}}$ be a fresh random variable.
	Then, we have
	\[ \HoareTriple{\Own \Delta \land \Own Z}{M}{ Y. \Conditioning_{\tilde{x} \gets \Delta} (\Own Y * \Own Z)} \]
	where $\tilde{x} = (x_1, \dots, x_n)$ is the tuple of deterministic variables, and we write $\Delta$ to denote both the context and the tuple of random variables $(X_1, \dots, X_n)$.
	\qed
\end{lemma}
\begin{proof}[Proof of Lemma~\ref{lem:extra-variable-conditional-independence}]
	Let $\mu$ be a probability measure on $(\Omega, \Sigma_{\Omega})$, and $D \in \RV{\interpret{\Delta}}$ be a random variable.
	Suppose $\gamma, (D, Z), (\mathcal{F}, \mu|_\mathcal{F}) \vDash \Own \Delta \land \Own Z$.
	By Lemma~\ref{lem:own-variables}, for any $\mathcal{P}_{\mathrm{frame}}$ such that $\mathcal{P}_{\mathrm{frame}} \cdot (\mathcal{F}, \mu|_\mathcal{F}) \sqsubseteq (\Sigma_{\Omega}, \mu)$ and any $D_{\mathrm{ext}}$, there exists $\mu_M$, $Y$, and $\mathcal{F}_M$ such that $\mathcal{P}_{\mathrm{frame}} \cdot (\mathcal{F}_M, \mu_M|_{\mathcal{F}_M}) \sqsubseteq (\Sigma_{\Omega}, \mu_M)$,
	\[ \left( \begin{aligned}
		&\omega \gets \mu; \\
		&y \gets \interpret{M}(D(\omega)); \\
		&\ret{(D_{\mathrm{ext}}(\omega), D(\omega), Z(\omega), y)}
	\end{aligned} \right) \quad=\quad \left( \begin{aligned}
		&\omega \gets \mu_M; \\
		&\ret{(D_{\mathrm{ext}}(\omega), D(\omega), Z(\omega), Y(\omega))}
	\end{aligned} \right), \]
	and $\gamma, (D, Z, Y), (\mathcal{F}_M, \mu|_{\mathcal{F}_M}) \vDash \Own \Delta \land \Own Z \land \Own Y$.
	Now, we show $\gamma, (D, Z, Y), (\mathcal{F}_M, \mu|_{\mathcal{F}_M}) \vDash \Conditioning_{\tilde{x} \gets \Delta} (\Own Y * \Own Z)$.
	By Lemma~\ref{lem:conditional-independence-formula}, it suffices to show that for any $\mu_M$-disintegration $\{ \nu_x \}_{x \in \interpret{\Delta}}$ with respect to $D$, for almost all $x$, we have $Y_{*} \nu_x \otimes Z_{*} \nu_x = (Y, Z)_{*} \nu_x$.
	Strictly speaking, we need to consider $(\Sigma_{\Omega}, \mu'_M) \sqsupseteq (\mathcal{F}_M, \mu|_{\mathcal{F}_M})$ and $\mu'_M$-disintegration $\{ \nu_x \}_x$ with respect to $D$, but by Lemma~\ref{lem:conditional-independence-equivalent}, it suffices to consider $\mu_M$-disintegration.
	By Lemma~\ref{lem:aee-markov-kernel}, it suffices to show the following equation for any $E, F, G$.
	\begin{align}
		&\Expectation{x}{D_{*} \mu_M}{\IversonBracket{x \in E} \cdot \nu_x(Y^{-1}(F)) \cdot \nu_x(Z^{-1}(G))} \\
		&= \Expectation{x}{D_{*} \mu_M}{\IversonBracket{x \in E} \cdot \nu_x(Y^{-1}(F) \cap Z^{-1}(G))}
	\end{align}
	This is proved as follows:
	\begin{align*}
		&\Expectation{x}{D_{*} \mu_M}{\IversonBracket{x \in E} \cdot \nu_x(Y^{-1}(F)) \cdot \nu_x(Z^{-1}(G))} \\
		&= \Expectation{x}{D_{*} \mu_M}{\IversonBracket{x \in E} \cdot \nu_x(Y^{-1}(F)) \cdot \Expectation{\omega}{\nu_x}{\IversonBracket{Z(\omega) \in G}}} \\
		&= \Expectation{x}{D_{*} \mu_M}{\Expectation{\omega}{\nu_x}{\IversonBracket{x \in E} \cdot \nu_x(Y^{-1}(F)) \cdot \IversonBracket{Z(\omega) \in G}}} \\
		&= \Expectation{x}{D_{*} \mu_M}{\Expectation{\omega}{\nu_x}{\IversonBracket{D(\omega) \in E} \cdot \nu_{D(\omega)}(Y^{-1}(F)) \cdot \IversonBracket{Z(\omega) \in G}}} \\
		&= \Expectation{\omega}{\mu_M}{\IversonBracket{D(\omega) \in E} \cdot \nu_{D(\omega)}(Y^{-1}(F)) \cdot \IversonBracket{Z(\omega) \in G}} \\
		&= \Expectation{(x, z)}{(D, Z)_{*} \mu_M}{\IversonBracket{x \in E} \cdot \nu_x(Y^{-1}(F)) \cdot \IversonBracket{z \in G}} \\
		&= \Expectation{(x, z)}{(D, Z)_{*} \mu}{\IversonBracket{x \in E} \cdot \nu_x(Y^{-1}(F)) \cdot \IversonBracket{z \in G}} \\
		&= \Expectation{(x, z)}{(D, Z)_{*} \mu}{\IversonBracket{x \in E} \cdot \interpret{M}(x)(F) \cdot \IversonBracket{z \in G}} \\
		&= \Expectation{(x, z)}{(D, Z)_{*} \mu}{\IversonBracket{x \in E} \cdot \Expectation{y}{\interpret{M}(x)}{\IversonBracket{y \in F}} \cdot \IversonBracket{z \in G}} \\
		&= \Expectation{(x, z)}{(D, Z)_{*} \mu}{\Expectation{y}{\interpret{M}(x)}{\IversonBracket{x \in E} \cdot \IversonBracket{y \in F} \cdot \IversonBracket{z \in G}}} \\
		&= \Expectation{(x, z, y)}{(D, Z, Y)_{*} \mu_M}{\IversonBracket{x \in E} \cdot \IversonBracket{y \in F} \cdot \IversonBracket{z \in G}} \\
		&= \Expectation{\omega}{\mu_M}{\IversonBracket{D(\omega) \in E} \cdot \IversonBracket{Y(\omega) \in F} \cdot \IversonBracket{Z(\omega) \in G}} \\
		&= \Expectation{x}{D_{*} \mu_M}{\Expectation{\omega}{\nu_x}{\IversonBracket{D(\omega) \in E} \cdot \IversonBracket{\omega \in Y^{-1}(F) \cap Z^{-1}(G)}}} \\
		&= \Expectation{x}{D_{*} \mu_M}{\Expectation{\omega}{\nu_x}{\IversonBracket{x \in E} \cdot \IversonBracket{\omega \in Y^{-1}(F) \cap Z^{-1}(G)}}} \\
		&= \Expectation{x}{D_{*} \mu_M}{\IversonBracket{x \in E} \cdot \Expectation{\omega}{\nu_x}{\IversonBracket{\omega \in Y^{-1}(F) \cap Z^{-1}(G)}}} \\
		&= \Expectation{x}{D_{*} \mu_M}{\IversonBracket{x \in E} \cdot \nu_x(Y^{-1}(F) \cap Z^{-1}(G))}
		\qedhere
	\end{align*}
\end{proof}

In Lemma~\ref{lem:extra-variable-conditional-independence}, $Z$ need not be independent of the input variables $X_1, \dots, X_n$, which makes the lemma stronger than the frame rule \hyperlink{rule:H-Frame}{\textsc{H-Frame}}.

By Lemma~\ref{lem:extra-variable-conditional-independence}, we have the following Hoare triple for the simulator~\eqref{eq:simulator-define}:
\[ \HoareTriple{\Own X_1 \land \Own X_2}{\mathbf{Sim}(M)}{Y_1. \Conditioning_{x \gets X_1} (\Own X_2 * \Own Y_1)} \]
By~\eqref{eq:simulator}, this Hoare triple is also satisfied by $(Y_1, Y_2) \gets M; \ret{Y_1}$.
Thus, we have the desired Hoare triple~\eqref{eq:hoare-noninterference}.

Precisely speaking, we cannot directly conclude~\eqref{eq:hoare-noninterference} because we have $\mathrm{wp}(M, X. \mathrm{wp}(N, Y. \phi)) \LilacEntails \mathrm{wp}(X \gets M; N, Y.\phi)$ but in general no guarantee for the converse.

\begin{proof}[Proof of Theorem~\ref{thm:noninterference-conditional-independence} ($\impliedby$)]
	Suppose that we have $\gamma, D, (\mathcal{F}_{\mathrm{pre}}, \mu_{\mathrm{pre}}) \vDash \Own X_1 \land \Own X_2$ where $\gamma \in \interpret{\cdot}$ and $D \in \RV{\interpret{X_1 : A_1, X_2 : A_2}}$ is a random variable.
	It is straightforward to show that 
	\[ \HoareTriple{\Own (X_1, X_2)}{M}{(Y_1, Y_2). \Own (X_1, X_2, Y_1, Y_2)}\]
	holds (Lemma~\ref{lem:own-variables}).
	By unfolding the definition of the Hoare triple, for any $\mathcal{P}_{\mathrm{frame}}$, $D_{\mathrm{ext}} \in \RV{\interpret{\Delta_{\mathrm{ext}}}}$, and $\mu$ such that $\mathcal{P}_{\mathrm{frame}} \cdot (\mathcal{F}_{\mathrm{pre}}, \mu_{\mathrm{pre}}) \sqsubseteq (\Sigma_{\Omega}, \mu)$, there exist $(D_{Y_1}, D_{Y_2}) \in \RV{\interpret{Y_1 : B_1, Y_2 : B_2}}$, $\mu_{M}$, and $\mathcal{P}_{M}$ such that $\mathcal{P}_{\mathrm{frame}} \cdot \mathcal{P}_{M} \sqsubseteq (\Sigma_{\Omega}, \mu_{M})$,
	\begin{align*}
		&\left( \begin{aligned}
		&\omega \gets \mu; \\
		&(y_1, y_2) \gets \interpret{M}(D(\omega)); \\
		&\ret{(D_{\mathrm{ext}}(\omega), D(\omega), y_1, y_2)}
	\end{aligned} \right) \\
	&=\quad \left( \begin{aligned}
		&\omega \gets \mu_{M}; \\
		&\ret{(D_{\mathrm{ext}}(\omega), D(\omega), D_{Y_1}(\omega), D_{Y_2}(\omega))}
	\end{aligned} \right),
	\end{align*}
	and $\gamma, (D, D_{Y_1}, D_{Y_2}), \mathcal{P}_{M} \vDash \Own (X_1, X_2, Y_1, Y_2)$ hold.
	Below, we show that the same $(D_{Y_1}, D_{Y_2})$, $\mu_{M}$, and $\mathcal{P}_{M}$ witness the desired Hoare triple~\eqref{eq:hoare-noninterference}.
	For this purpose, it suffices to show $\gamma, (D, D_{Y_1}, D_{Y_2}), \mathcal{P}_{M} \vDash \Conditioning_{x \gets X_1} (\Own X_2 * \Own Y_1)$.
	By the substitution property \cite[Lemma~B.10]{LiPLDI2023}, this is equivalent to $\gamma, (D, D_{Y_1}), \mathcal{P}_{M} \vDash \Conditioning_{x \gets X_1} (\Own X_2 * \Own Y_1)$.
	By Lemma~\ref{lem:extra-variable-conditional-independence}, we have
	\[ \{ \Own X_1 \land \Own X_2 \}\ \mathbf{Sim}(M)\ \{ Y_1. \Conditioning_{x \gets X_1} (\Own X_2 * \Own Y_1) \} \]
	and thus, there exist $D'_{Y_1} \in \RV{\interpret{Y_1 : B_1}}$, $\mu_{\mathbf{Sim}(M)}$, and $\mathcal{P}_{\mathbf{Sim}(M)}$ such that $\mathcal{P}_{\mathrm{frame}} \cdot \mathcal{P}_{\mathbf{Sim}(M)} \sqsubseteq (\Sigma_{\Omega}, \mu_{\mathbf{Sim}(M)})$,
	\begin{align}
		&\left( \begin{aligned}
			&\omega \gets \mu; \\
			&y_1 \gets \interpret{\mathbf{Sim}(M)}(D_{X_1}(\omega)); \\
			&\ret{(D_{\mathrm{ext}}(\omega), D(\omega), y_1)}
		\end{aligned} \right) \\
		&= \left( \begin{aligned}
			&\omega \gets \mu_{\mathbf{Sim}(M)}; \\
			&\ret{(D_{\mathrm{ext}}(\omega), D(\omega), D'_{Y_1}(\omega))}
		\end{aligned} \right),
	\end{align}
	and $\gamma, (D, D_{Y_1}), \mathcal{P}_{\mathbf{Sim}(M)} \vDash \Conditioning_{x \gets X_1} (\Own X_2 * \Own Y_1)$.
	Here, $\mu_M$ and $\mu_{\mathbf{Sim}(M)}$ are not necessarily the same, but by~\eqref{eq:simulator}, we have the following equation:
	\begin{align}
		&\left( \begin{aligned}
			&\omega \gets \mu_{\mathbf{Sim}(M)}; \\
			&\ret{(D(\omega), D'_{Y_1}(\omega))}
		\end{aligned} \right)\ =\ \left( \begin{aligned}
			&\omega \gets \mu_{M}; \\
			&\ret{(D(\omega), D_{Y_1}(\omega))}
		\end{aligned} \right)
		\label{eq:sim-distribution}
	\end{align}

	The remaining task is to show that if $\gamma, (D, D'_{Y_1}), \mathcal{P}_{\mathbf{Sim}(M)} \vDash \Conditioning_{x \gets X_1} (\Own X_2 * \Own Y_1)$, then $\gamma, (D, D_{Y_1}), \mathcal{P}_{M} \vDash \Conditioning_{x \gets X_1} (\Own X_2 * \Own Y_1)$.
	This follows because $(D, D'_{Y_1})_{*} \mu_{\mathbf{Sim}(M)}$ and $(D, D_{Y_1})_{*} \mu_{M}$ give the same distribution by~\eqref{eq:sim-distribution} (see Lemma~\ref{lem:conditional-independence-equivalent} for the formal argument).
\end{proof}

\end{toappendix}

\section{Additional Proof Rules for Lilac}
\label{sec:additional-rules}

The original Lilac paper \cite{LiPLDI2023} provides several proof rules, which are summarized in Fig.~\ref{fig:lilac-rules}.
However, these rules are not sufficient for our purpose.
In this section, we present additional proof rules that we use in Section~\ref{sec:examples}.
In particular, we show that semi-graphoid axioms for conditional independence hold in Lilac.
We also provide composability rules for $(\mathcal{I}, \mathcal{O})$-NI properties as proof rules for Lilac.

\begin{toappendix}
\begin{figure*}
	\begin{mathpar}
		\inferrule[\hypertarget{rule:Dist-Own}{Dist-Own}]{}{
			E \sim \mu \LilacEntails \Own E
		}
		\and
		\inferrule[Transfer-Own]{}{
			\Own E_1 \land (E_1 \AlmostSurelyEqual E_2) \LilacEntails \Own E_2
		}
		\and
		\inferrule[\hypertarget{rule:Transfer-Dist}{Transfer-Dist}]{}{
			E_1 \sim \mu \land (E_1 \AlmostSurelyEqual E_2) \LilacEntails E_2 \sim \mu
		}
		\and
		\inferrule[\hypertarget{rule:Eq-Congruence}{Eq-Congruence}]{}{
			\Own (F[E_1], F[E_2]) \land (E_1 \AlmostSurelyEqual E_2) \LilacEntails F[E_1] \AlmostSurelyEqual F[E_2]
		}
		\and
		\inferrule[\hypertarget{rule:Eq-Conj}{Eq-Conj}]{}{
			\phi \land (E_1 \AlmostSurelyEqual E_2) \LilacIff \phi * (E_1 \AlmostSurelyEqual E_2)
		}
		\and
		\inferrule[C-Entail]{
			\phi \LilacEntails \psi
		}{
			\Conditioning_{x \gets E} \phi \LilacEntails \Conditioning_{x \gets E} \psi
		}
		\and
		\inferrule[C-Indep]{}{
			\Own E * \phi \LilacEntails \Conditioning_{x \gets E} \phi
		}
		\and
		\inferrule[C-Subst]{}{
			\Own X \LilacEntails \Conditioning_{x \gets X} (X \AlmostSurelyEqual x)
		}
		\and
		\inferrule[C-Own]{}{
			\Own E \LilacEntails \Conditioning_{x \gets E} \Own E
		}
		\and
		\inferrule[Unif-Bijection]{
			\Gamma ; X : \mathbb{K} \vdash E : \mathbb{K} \\
			\text{$E$ is bijective in $X$}
		}{
			X \sim \Uniform{\mathbb{K}} \LilacEntails E \sim \Uniform{\mathbb{K}}
		}
		\and
		\inferrule[Necessitation]{
			\LilacEntails \phi
		}{
			\LilacEntails \Conditioning_{x \gets X} \phi
		}
		\and
		\inferrule[C-And-Distributive]{}{
			\Conditioning_{x \gets X} (\phi \land \psi) \LilacIff \Conditioning_{x \gets X} \phi \land \Conditioning_{x \gets X} \psi
		}
		\and
		\inferrule[\hypertarget{rule:C-Or-SemiDistributive}{C-Or-SemiDistributive}]{}{
			\Conditioning_{x \gets X} \phi \lor \Conditioning_{x \gets X} \psi \LilacEntails \Conditioning_{x \gets X} (\phi \lor \psi)
		}
		\and
		\inferrule[\hypertarget{rule:FV-Own}{FV-Own}]{
			\Gamma; X_1 : A_1, \dots, X_n : A_n \vdash E : B
		}{
			\Own (X_1, \dots, X_n) \LilacEntails \Own E
		}
		\and
		\inferrule[C-Unif]{}{
			\Own X \land \Conditioning_{x \gets X} (Y \sim \Uniform{\mathbb{K}}) \LilacEntails \Own X * (Y \sim \Uniform{\mathbb{K}})
		}
		\and
		\inferrule{}{
			X \sim \mu \land (\Own X * \phi) \LilacEntails (X \sim \mu) * \phi
		}
		\and
		\inferrule[\hypertarget{rule:Persistent-Mono}{$\Persistent$-Mono}]{
			\phi \LilacEntails \psi
		}{
			\Persistent \phi \LilacEntails \Persistent \psi
		}
		\and
		\inferrule[$\Persistent$-Idem]{}{
			\Persistent \phi \LilacEntails \Persistent \Persistent \phi
		}
		\and
		\inferrule[$\Persistent$-Elim]{}{
			\Persistent \phi \LilacEntails \phi
		}
		\and
		\inferrule[\hypertarget{rule:Persistent-True}{$\Persistent$-True}]{}{
			\LilacEntails \Persistent \top
		}
		\and
		\inferrule[$\Persistent$-And]{}{
			\Persistent \phi \land \Persistent \psi \LilacEntails \Persistent (\phi \land \psi)
		}
		\and
		\inferrule[$\Persistent$-Conj]{}{
			\phi \land \Persistent \psi \LilacEntails \phi * \Persistent \psi
		}
		\and
		\inferrule{}{
			\Own (X_1, \dots, X_n) \LilacIff \bigwedge_{i = 1}^{n} \Own X_i
		}
	\end{mathpar}
	\caption{List of inference rules for Lilac formulas~\cite{LiPLDI2023}.
	We omit obvious structural rules like commutativity and associativity of the separating conjunction $*$.
	Note that Lilac is an \emph{affine} separation logic, and thus $\phi * \psi \LilacEntails \phi$ holds.}
	\label{fig:lilac-rules-full}
\end{figure*}
\end{toappendix}

\subsection{Semi-Graphoid Axioms for Conditional Independence}

In the classical probability theory, conditional independence is known to satisfy the \emph{graphoid axioms} \cite{PearlECAI1986}.
Here, we show that a subset of the graphoid axioms, called \emph{semi-graphoid axioms}, hold in Lilac.

Using the conventional notation $X \perp Y \mid Z$ meaning conditional independence of random variables $X$ and $Y$ given $Z$, the semi-graphoid axioms are stated as follows.
\begin{itemize}
	\item Symmetry: $X \perp Y \mid Z$ implies $Y \perp X \mid Z$.
	\item Decomposition: $W \perp (X, Y) \mid Z$ implies $W \perp X \mid Z$.
	\item Weak union: $W \perp (X, Y) \mid Z$ implies $W \perp X \mid (Y, Z)$.
	\item Contraction: $W \perp X \mid Z$ and $W \perp Y \mid (X, Z)$ together imply $W \perp (X, Y) \mid Z$.
\end{itemize}
In Lilac, these axioms can be written as follows.
\begin{align}
	&\Conditioning_{z \gets Z}(\Own X * \Own Y) \LilacEntails \Conditioning_{z \gets Z}(\Own Y * \Own X) \\
	&\Conditioning_{z \gets Z} (\Own W * \Own (X, Y)) \LilacEntails \Conditioning_{z \gets Z} (\Own W * \Own X) \\
	&\Conditioning_{z \gets Z} (\Own W * \Own (X, Y)) \LilacEntails \Conditioning_{z \gets Z} \Conditioning_{y \gets Y} (\Own W * \Own X) \\
	&\Conditioning_{z \gets Z} (\Own W * \Own X) \land \Conditioning_{z \gets Z} \Conditioning_{x \gets X} (\Own W * \Own Y) \\
	&\quad\LilacEntails \Conditioning_{z \gets Z} (\Own W * \Own (X, Y))
\end{align}
We show these entailments hold in Lilac's semantics.
By \hyperlink{rule:C-Entail}{\textsc{C-Entail}} and \hyperlink{rule:C-And-Distributive}{\textsc{C-And-Distributive}}, it suffices to show these entailments without the outermost conditioning modalities $\Conditioning_{z \gets Z}$.
The symmetry axiom is immediate from the commutativity of the separating conjunction.
The decomposition axiom is also straightforward because $\Own (X, Y) \LilacEntails \Own X$ holds.
The weak union and contraction axioms are non-trivial, but they also hold as stated in the following proposition.
The proofs are given in
\iflong
Appendix~\ref{sec:proof-semi-graphoid-axioms}.
\fi

\begin{proposition}\label{prop:semi-graphoid-axioms}
	The following entailments hold.
	\begin{mathpar}
		\inferrule[\hypertarget{rule:Gen-Weak-Union}{Gen-Weak-Union}]{}{
			\phi * \Own (X, Y) \LilacEntails \mathop{\mathbf{C}}_{x : A \gets X}(\phi * \Own (X, Y))
		}
		\and
		\inferrule[\hypertarget{rule:Contraction}{Contraction}]{}{
			{\begin{aligned}
				&(\Own W * \Own X) \land \Conditioning_{x \gets X}(\Own W * \Own Y) \\
				&\quad\LilacEntails \Own W * \Own (X, Y)
				\qedhere
			\end{aligned}}
		}
	\end{mathpar}
	\qed
\end{proposition}
\begin{toappendix}
\subsection{Proof of Proposition~\ref{prop:semi-graphoid-axioms}}
\label{sec:proof-semi-graphoid-axioms}
\end{toappendix}
\begin{appendixproof}[Proof of Proposition~\ref{prop:semi-graphoid-axioms} (\hyperlink{rule:Gen-Weak-Union}{\textsc{Gen-Weak-Union}})]
	Suppose $\gamma, D, \mathcal{P} \vDash \phi * \Own (X, Y)$.
	By definition, there exists $\mathcal{P}_1 \cdot \mathcal{P}_2 \sqsubseteq \mathcal{P}$ such that $\gamma, D, \mathcal{P}_1 \vDash \phi$ and $\gamma, D, \mathcal{P}_2 \vDash \Own (X, Y)$.
	To show $\gamma, D, \mathcal{P} \vDash \mathop{\mathbf{C}}_{x : A \gets X}(\phi * \Own Y)$, suppose $\mathcal{P} \sqsubseteq (\Sigma_{\Omega}, \mu)$ and let $\{ \nu_x \}_{x \in X}$ be a $\mu$-disintegration with respect to $X$.
	By the proof of \hyperlink{rule:C-Indep}{\textsc{C-Indep}}, we know that $\nu_x|_{\mathcal{P}_1} = \mathcal{P}_1$ for almost all $x$.
	We also have $\gamma, D, \nu_x|_{\mathcal{P}_2} \vDash \Own (X, Y)$ because $\Own (X, Y)$ does not depend on the probability distribution.
	Thus, it suffices to show that $\nu_x|_{\mathcal{P}_1} \cdot \nu_x|_{\mathcal{P}_2} \sqsubseteq \nu_x$ holds for almost all $x$.
	Since we already know that $\mathcal{P}_1 \cdot \mathcal{P}_2 \sqsubseteq \mathcal{P}$, we show that for almost all $x$, for any $F \in \mathcal{F}_{\mathcal{P}_1}$ and $G \in \mathcal{F}_{\mathcal{P}_2}$, $\nu_x(F \cap G) = \nu_x(F) \nu_x(G)$.
	By Lemma~\ref{lem:aee-markov-kernel}, this is equivalent to
	\begin{align}
		&\Expectation{x}{(X(\gamma) \circ D)_{*} \mu}{\IversonBracket{x \in E} \cdot \nu_x(F \cap G)} \\
		&= \Expectation{x}{(X(\gamma) \circ D)_{*} \mu}{\IversonBracket{x \in E} \cdot \nu_x(F) \cdot \nu_x(G)}
	\end{align}
	which is proved as follows.
	\begin{align*}
		&\Expectation{x}{(X(\gamma) \circ D)_{*} \mu}{\IversonBracket{x \in E} \cdot \nu_x(F) \cdot \nu_x(G)} \\
		&= \Expectation{x}{(X(\gamma) \circ D)_{*} \mu}{\IversonBracket{x \in E} \cdot \mu(F) \cdot \nu_x(G)} \\
		&= \mu(F) \Expectation{x}{(X(\gamma) \circ D)_{*} \mu}{\IversonBracket{x \in E} \cdot \Expectation{\omega}{\nu_x}{\IversonBracket{\omega \in G}}} \\
		&= \mu(F) \Expectation{x}{(X(\gamma) \circ D)_{*} \mu}{\Expectation{\omega}{\nu_x}{\IversonBracket{\omega \in G \cap (X(\gamma) \circ D)^{-1}(E)}}} \\
		&= \mu(F) \Expectation{\omega}{\mu}{\IversonBracket{\omega \in G \cap (X(\gamma) \circ D)^{-1}(E)}} \\
		&= \mu(F) \mu(G \cap (X(\gamma) \circ D)^{-1}(E)) \\
		&= \mu(F \cap G \cap (X(\gamma) \circ D)^{-1}(E)) \qquad \text{by $\gamma, D, \mathcal{P}_2 \vDash \Own X$ and $\mathcal{P}_1 \cdot \mathcal{P}_2 \sqsubseteq (\Sigma_{\Omega}, \mu)$} \\
		&= \Expectation{x}{(X(\gamma) \circ D)_{*} \mu}{\Expectation{\omega}{\nu_x}{\IversonBracket{\omega \in F \cap G \cap (X(\gamma) \circ D)^{-1}(E)}}} \\
		&= \Expectation{x}{(X(\gamma) \circ D)_{*} \mu}{\Expectation{\omega}{\nu_x}{\IversonBracket{\omega \in F \cap G} \cdot \IversonBracket{x \in E}}} \\
		&= \Expectation{x}{(X(\gamma) \circ D)_{*} \mu}{\IversonBracket{x \in E} \cdot \nu_x(F \cap G)}
		\qedhere
	\end{align*}
\end{appendixproof}

\begin{appendixproof}[Proof of Proposition~\ref{prop:semi-graphoid-axioms} (\hyperlink{rule:Contraction}{\textsc{Contraction}})]
	Suppose $\gamma, D, \mathcal{P} \vDash (\Own W * \Own X) \land \Conditioning_{x \gets X}(\Own W * \Own Y)$.
	Let $\mathcal{F}_X$ and $\mathcal{F}_W$ be the sub-$\sigma$-algebras generated by $X$ and $W$, respectively.
	Let $\mu$ be a probability measure on $(\Omega, \Sigma_{\Omega})$ such that $\mathcal{P} \sqsubseteq (\Sigma_{\Omega}, \mu)$.
	By $\gamma, D, \mathcal{P} \vDash \Own W * \Own X$, we have $(\mu|_{\mathcal{F}_X}, \mathcal{F}_X) \cdot (\mu|_{\mathcal{F}_W}, \mathcal{F}_W) \sqsubseteq (\mu, \Sigma_{\Omega})$.
	By $\gamma, D, \mathcal{P} \vDash \Conditioning_{x \gets X}(\Own W * \Own Y)$, for any $\mu$-disintegration $\{ \nu_x \}_{x \in X}$ with respect to $X$, there exists $\mathcal{F}_Y \subseteq \mathcal{F}$ such that $Y$ is $\mathcal{F}_Y$-measurable and $(\nu_x|_{\mathcal{F}_W}, \mathcal{F}_W) \cdot (\nu_x|_{\mathcal{F}_Y}, \mathcal{F}_Y) \sqsubseteq (\nu_x, \Sigma_{\Omega})$ for almost all $x$.
	By the proof of \hyperlink{rule:C-Indep}{\textsc{C-Indep}}, we know that $\nu_x|_{\mathcal{F}_W} = \mu|_{\mathcal{F}_W}$ for almost all $x$.
	Now, we show that for any $E_1 \in \mathcal{F}_X$, $E_2 \in \mathcal{F}_Y$, and $E_3 \in \mathcal{F}_W$, we have
	\[ \mu(E_1 \cap E_2 \cap E_3) = \mu(E_1 \cap E_2) \cdot \mu(E_3). \]
	Since $\mathcal{F}_X = \{ X^{-1}(E) \mid E \in \Sigma_X \}$, there exists $E_X \in \Sigma_X$ such that $E_1 = X^{-1}(E_X)$.
	\begin{align*}
		&\mu(E_1 \cap E_2 \cap E_3) \\
		&= \Expectation{x}{X_{*} \mu}{\Expectation{\omega}{\nu_x}{\IversonBracket{\omega \in X^{-1}(E_X)}} \cdot \IversonBracket{\omega \in E_2 \cap E_3}} \\
		&= \Expectation{x}{X_{*} \mu}{\IversonBracket{x \in E_X} \cdot \Expectation{\omega}{\nu_x}{\IversonBracket{\omega \in E_2 \cap E_3}}} \\
		&= \Expectation{x}{X_{*} \mu}{\IversonBracket{x \in E_X} \cdot \nu_x(E_2) \cdot \nu_x(E_3)} \\
		&= \Expectation{x}{X_{*} \mu}{[x \in E_X] \cdot \nu_x(E_2) \cdot \mu(E_3)} \\
		&= \Expectation{x}{X_{*} \mu}{[x \in E_X] \cdot \nu_x(E_2)} \cdot \mu(E_3) \\
		&= \Expectation{x}{X_{*} \mu}{\Expectation{\omega}{\nu_x}{[\omega \in X^{-1}(E_X)] \cdot [\omega \in E_2]}} \cdot \mu(E_3) \\
		&= \mu(E_1 \cap E_2) \cdot \mu(E_3)
	\end{align*}
	Consider the following set.
	\[ \{ E \in \sigma(\mathcal{F}_X \cup \mathcal{F}_Y) \mid \forall E_3 \in \mathcal{F}_W, \mu(E \cap E_3) = \mu(E) \cdot \mu(E_3) \} \]
	It is easy to see that this is a $\lambda$-system containing the $\pi$-system $\{ E_1 \cap E_2 \mid E_1 \in \mathcal{F}_X, E_2 \in \mathcal{F}_Y \}$.
	By Dynkin's $\pi$-$\lambda$ theorem, the $\lambda$-system contains $\sigma(\{ E_1 \cap E_2 \mid E_1 \in \mathcal{F}_X, E_2 \in \mathcal{F}_Y \}) = \sigma(\mathcal{F}_X \cup \mathcal{F}_Y)$.
	Hence, we have $(\mu|_{\mathcal{F}_X}, \mathcal{F}_X) \cdot (\mu|_{\sigma(\mathcal{F}_X \cup \mathcal{F}_Y)}, \sigma(\mathcal{F}_X \cup \mathcal{F}_Y)) \sqsubseteq (\mu, \Sigma_{\Omega})$.
	Note that $\sigma(\mathcal{F}_X \cup \mathcal{F}_Y)$ is countably generated because both $\mathcal{F}_X$ and $\mathcal{F}_Y$ are countably generated.
\end{appendixproof}

Note that we slightly generalize the statements of the weak union.
It subsumes \hyperlink{rule:C-Indep}{\textsc{C-Indep}} as a special case when $Y$ is a trivial random variable over a singleton set and also subsumes \hyperlink{rule:C-Own}{\textsc{C-Own}}.
When verifying noninterference properties of the examples in Section~\ref{sec:examples}, we often use \hyperlink{rule:Gen-Weak-Union}{\textsc{Gen-Weak-Union}} to introduce conditioning modalities.

\begin{toappendix}
\subsection{Generalized \textsc{Unif-Bijection}}

\begin{lemma}\label{lem:generalised-bijective-substitution-uniform}
	\hyperlink{rule:Unif-Bijection}{\textsc{Unif-Bijection}} can be generalized to the following.
	\begin{mathpar}
		\inferrule{
			\Gamma; \Delta, Y : \mathbb{K} \vdash E : \mathbb{K} \\
			\text{$E$ is bijective in $Y$}
		}{
			\Own \Delta * F \sim \Uniform{\mathbb{K}} \LilacEntails \Own \Delta * E[F/Y] \sim \Uniform{\mathbb{K}}
		}
	\end{mathpar}
\end{lemma}
\begin{proofsketch}
	Suppose that we have $\Own \Delta * F \sim \Uniform{\mathbb{K}}$.
	By \hyperlink{rule:C-Indep}{\textsc{C-Indep}}, we have $\Conditioning_{x \gets \Delta} (F \sim \Uniform{\mathbb{K}})$.
	By \hyperlink{rule:C-Subst}{\textsc{C-Subst}}, we have $\Conditioning_{x \gets \Delta} (E[x/\Delta] \AlmostSurelyEqual E \land F \sim \Uniform{\mathbb{K}})$.
	By applying \hyperlink{rule:Unif-Bijection}{\textsc{Unif-Bijection}} to $E[x/\Delta]$, we have $\Conditioning_{x \gets \Delta} (E[F/Y] \sim \Uniform{\mathbb{K}})$.
	Finally, by \hyperlink{rule:C-Unif}{\textsc{C-Unif}}, we have $\Own \Delta * E[F/Y] \sim \Uniform{\mathbb{K}}$.
\end{proofsketch}
\begin{appendixproof}[Proof of Lemma~\ref{lem:generalised-bijective-substitution-uniform}]
	For simplicity, we show the case where $\Delta = X : A$, but the proof below can be easily generalized to the case of multiple variables.
	\begin{align}
		&\Own X * F \sim \Uniform{\mathbb{K}} \\
		&\LilacEntails \Conditioning_{x \gets X} (\Own X * F \sim \Uniform{\mathbb{K}}) & \text{by \hyperlink{rule:Gen-Weak-Union}{\textsc{Gen-Weak-Union}}} \\
		&\LilacEntails \Conditioning_{x \gets X} (\Own X * E[x/X,F/Y] \sim \Uniform{\mathbb{K}}) \\
		&\LilacEntails \Conditioning_{x \gets X} ((\Own X * E[x/X,F/Y]  \sim \Uniform{\mathbb{K}}) \land x \AlmostSurelyEqual X) & \text{by \hyperlink{rule:C-Subst}{\textsc{C-Subst}}} \\
		&\LilacEntails \Conditioning_{x \gets X} (\Own X * (E[x/X,F/Y]  \sim \Uniform{\mathbb{K}} \land E[x/X,F/Y] \AlmostSurelyEqual E[F/Y])) & \text{by \hyperlink{rule:Eq-Congruence}{\textsc{Eq-Congruence}} and \hyperlink{rule:Eq-Conj}{\textsc{Eq-Conj}}} \\
		&\LilacEntails \Conditioning_{x \gets X} (\Own X * E[F/Y] \sim \Uniform{\mathbb{K}}) & \text{by \hyperlink{rule:Transfer-Dist}{\textsc{Transfer-Dist}}}
	\end{align}
	Therefore, we have
	\begin{align}
		&\Own X * F \sim \Uniform{\mathbb{K}} \\
		&\LilacEntails \Own X \land \Conditioning_{x \gets X} (\Own X * E[F/Y] \sim \Uniform{\mathbb{K}}) \\
		&\LilacEntails \Own X * E[F/Y] \sim \Uniform{\mathbb{K}} & \text{by \hyperlink{rule:C-Unif}{\textsc{C-Unif}}}
	\end{align}
\end{appendixproof}
\end{toappendix}

\subsection{Hoare Triples with Ownership Assertions}

When we write Hoare triples for APPL programs $\Delta \vdash M : \mathtt{G} A$, we often want to assume $\Own \Delta$ in the precondition.
Writing $\Own \Delta$ explicitly in every Hoare triple is tedious, so we introduce a syntactic sugar for this purpose.
For each APPL program $\Delta \vdash M : \mathtt{G} A$, we define $\HoareTripleOwn{P}{M}{X. Q}$ as syntactic sugar for
\[ \HoareTriple{P \land \Own \Delta}{M}{X. Q \land \Own \Delta \land \Own X}. \]

We can use this syntactic sugar in the same way as the original Hoare triples.
\begin{proposition}
	All rules in Fig.~\ref{fig:lilac-hoare-rules} is sound when we replace all occurrences of $\HoareTriple{P}{M}{X. Q}$ with $\HoareTripleOwn{P}{M}{X. Q}$.
	\qed
\end{proposition}

\begin{proposition}
	Theorem~\ref{thm:noninterference-conditional-independence} holds when we replace $\HoareTriple{P}{M}{X. Q}$ with $\HoareTripleOwn{P}{M}{X. Q}$.
	\qed
\end{proposition}

In particular, $(\mathcal{I}, \mathcal{O})$-NI property of an APPL program $M$ can be written as follows by Theorem~\ref{thm:noninterference-conditional-independence}.
\[ \HoareTripleOwn{\top}{M}{\mathbf{Y}. \Conditioning_{x \gets \mathcal{I}} (\Own (\Delta \setminus \mathcal{I}) * \Own \mathcal{O})} \]

\subsection{Monotonicity}
\label{sec:monotonicity}

By the simulator-based definition of $(\mathcal{I}, \mathcal{O})$-NI, it is straightforward to see that if $M$ is $(\mathcal{I}, \mathcal{O})$-NI, $\mathcal{I} \subseteq \mathcal{I}'$, and $\mathcal{O} \supseteq \mathcal{O}'$, then $M$ is also $(\mathcal{I}', \mathcal{O}')$-NI.
This gives rise to the following proof rule.
\begin{mathpar}
	\inferrule{
		\mathcal{I} \subseteq \mathcal{I}' \subseteq \Delta \\
		\mathbf{Y} \supseteq \mathcal{O} \supseteq \mathcal{O}' \\\\
		\HoareTripleOwn{\top}{M}{\mathbf{Y}. \Conditioning_{x \gets \mathcal{I}} (\Own (\Delta \setminus \mathcal{I}) * \Own \mathcal{O})}
	}{
		\HoareTripleOwn{\top}{M}{\mathbf{Y}. \Conditioning_{x \gets \mathcal{I}'} (\Own (\Delta \setminus \mathcal{I}') * \Own \mathcal{O}')}
	}
\end{mathpar}
This is a consequence of Theorem~\ref{thm:noninterference-conditional-independence}, but we can also prove it from semi-graphoid axioms.

\begin{lemma}
	Suppose $\mathcal{I} \subseteq \mathcal{I}' \subseteq \Delta$ and $\mathcal{X} \supseteq \mathcal{O} \supseteq \mathcal{O}'$.
	Then, we have the following.
	\[ (\Delta \setminus \mathcal{I}) \perp \mathcal{O} \mid \mathcal{I} \quad\implies\quad (\Delta \setminus \mathcal{I}') \perp \mathcal{O}' \mid \mathcal{I}' \]
\end{lemma}
\begin{proof}
	Without loss of generality, we show the case where both $\Delta = \{ X_1, X_2, X_3 \}$ and $\mathbf{Y} = (Y_1, Y_2, Y_3)$ consist of three variables, and we have $\mathcal{I} = \{ X_1 \}$, $\mathcal{I}' = \{ X_1, X_2 \}$, $\mathcal{O} = \{ Y_1, Y_2 \}$, and $\mathcal{O}' = \{ Y_1 \}$.
	Then, we have the following.
	\begin{align}
		&X_2, X_3 \perp Y_1, Y_2 \mid X_1 \\
		&\implies X_2, X_3 \perp Y_1 \mid X_1 &\text{by decomposition} \\
		&\implies X_3 \perp Y_1 \mid X_1, X_2 &\text{by weak union} &
		\qedhere
	\end{align}
\end{proof}

\subsection{Composability for Sequential Composition}
Suppose we have two APPL programs $\Delta \vdash M : \mathtt{G} A$ and $\Delta, \mathbf{X} : A \vdash N : \mathtt{G} B$ whose output variables are tuples $\mathbf{X}$ and $\mathbf{Y}$, respectively.
Consider the sequential composition $\Delta \vdash \mathbf{X} \gets M; N : \mathtt{G} B$.
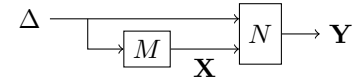
\begin{center}
	\begin{tikzpicture}
		\node[draw] (M) at (0, 0) {$M$};
		\node[draw, minimum height=0.8cm] (N) at (1.5, 0.2) {$N$};
		\draw[->] ($ (N.west) + (-2.5, 0.2) $) node[left] {$\Delta$} -- ($ (N.west) + (0, 0.2) $);
		\draw[->] (N.east) -- ($ (N.east) + (0.5, 0) $) node[right] {$\mathbf{Y}$};
		\draw[->] (M) -- node[midway, below] {$\mathbf{X}$} (N.west |- M.east);
		\draw[->] ($ (N.west) + (-2, 0.2) $) |- (M);
	\end{tikzpicture}
\end{center}
If $M$ is $(\mathcal{I}_M, \mathcal{O}_M)$-NI and $N$ is $(\mathcal{I}_N, \mathcal{O}_N)$-NI with $\mathcal{O}_M \supseteq \mathcal{I}_N \cap \mathbf{X}$, then the sequential composition $\mathbf{X} \gets M; N$ is $(\mathcal{I}_M \cup (\mathcal{I}_N \cap \Delta), \mathcal{O}_N)$-NI.
We formalize this composability property as a proof rule in Lilac as follows.
By monotonicity of $(\mathcal{I}, \mathcal{O})$-NI (Section~\ref{sec:monotonicity}), it suffices to consider the case where $\mathcal{O}_M = \mathcal{I}_N \cap \mathbf{X}$ and $\mathcal{I}_M = \mathcal{I}_N \cap \Delta$:
\begin{mathpar}
	\inferrule{
		\text{$M$ is $(\mathcal{I}, \mathcal{O}_M)$-NI}\\
		\text{$N$ is $(\mathcal{I} \cup \mathcal{O}_M, \mathcal{O}_N)$-NI}
	}{
		\text{$\mathbf{X} \gets M; N$ is $(\mathcal{I}, \mathcal{O}_N)$-NI}
	}
\end{mathpar}
This gives rise to the following proof rule.
\begin{mathpar}
	\inferrule{
		\HoareTripleOwn{\top}{M}{\mathbf{X}. \Conditioning_{w \gets \mathcal{I}} (\Own (\Delta \backslash \mathcal{I}) * \Own \mathcal{O}_M)} \\
		\HoareTripleOwn{\top}{N}{\mathbf{Y}. \mkern-20mu \Conditioning_{(w, x) \gets (\mathcal{I}, \mathcal{O}_M)} \mkern-20mu (\Own (\Delta \backslash \mathcal{I}, \mathbf{X} \backslash \mathcal{O}_M) * \Own \mathcal{O}_N)}
	}{
		\HoareTripleOwn{\top}{\mathbf{X} \gets M; N}{\mathbf{Y}. \Conditioning_{x \gets \mathcal{I}} (\Own (\Delta \backslash \mathcal{I}) * \Own \mathcal{O}_N)}
	}
\end{mathpar}

Similarly to monotonicity, this is a consequence of Theorem~\ref{thm:noninterference-conditional-independence}, but we can also prove it from semi-graphoid axioms.

\begin{lemma}\label{lem:sequential-composability-from-semi-graphoid}
	We have the following.
	\begin{align}
		&(\Delta \backslash \mathcal{I} \perp \mathcal{O}_M \mid \mathcal{I}) \land (\Delta \backslash \mathcal{I}, \mathbf{X} \backslash \mathcal{O}_M \perp \mathcal{O}_N \mid \mathcal{I}, \mathcal{O}_M) \\
		&\implies \Delta \backslash \mathcal{I} \perp \mathcal{O}_N \mid \mathcal{I}
		\tag*{\qed}
	\end{align}
\end{lemma}
\begin{toappendix}
\subsection{Proof of Lemma~\ref{lem:sequential-composability-from-semi-graphoid}}
\end{toappendix}
\begin{appendixproof}[Proof of Lemma~\ref{lem:sequential-composability-from-semi-graphoid}]
	Without loss of generality, we show the case where $\Delta = \{ W_1, W_2 \}$, $\mathbf{X} = \{ X_1, X_2 \}$, and $\mathbf{Y} = \{ Y_1, Y_2 \}$.
	Also, let $\mathcal{I} = \{ W_1 \}$, $\mathcal{O}_M = \{ X_1 \}$, and $\mathcal{O}_N = \{ Y_1 \}$.
	Then, we have the following.
	\begin{align}
		&W_2 \perp X_1 \mid W_1\ \land\ W_2, X_2 \perp Y_1 \mid W_1, X_1 \\
		&\implies W_2 \perp X_1 \mid W_1\ \land\ W_2 \perp Y_1 \mid W_1, X_1 \\
		&\implies W_2 \perp X_1, Y_1 \mid W_1 \\
		&\implies W_2 \perp Y_1 \mid W_1
		\qedhere
	\end{align}
\end{appendixproof}

\subsection{Composability for For-Loops}

Consider a for-loop $\forloop{i}{n_s}{n_e}{\mathbf{Y}}{\mathbf{E}}{M}$ where $\mathbf{Y} = (Y_1, \dots, Y_n)$ is a tuple of random variables, and $\Delta, i : \mathtt{index}, \mathbf{Y} : B \vdash M : B$.
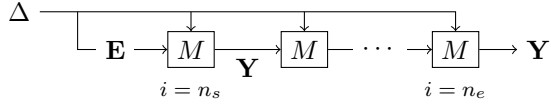
\begin{center}
	\begin{tikzpicture}
		\coordinate (delta) at (-2, 0.5);
		\node[left] at (delta) {$\Delta$};
		\coordinate (E) at (-1, 0);
		\node[draw] (M1) at (0, 0) {$M$};
		\node[below=0.3cm,font=\footnotesize] at (M1) {$i = n_s$};
		\node[draw] (M2) at (1.5, 0) {$M$};
		\node[draw] (M_last) at (3.5, 0) {$M$};
		\node[below=0.3cm,font=\footnotesize] at (M_last) {$i = n_e$};
		\draw[->] (delta) -| (M_last);
		\draw[->] ($(E |- delta) + (-0.5, 0)$) |- (M1);
		\node[fill=white] at (E) {$\mathbf{E}$};
		\draw[->] (M1) -- node[midway, below] {$\mathbf{Y}$} (M2);
		\draw[->] (M1 |- delta) -- (M1);
		\draw[->] (M2 |- delta) -- (M2);
		\draw[->] (M2) -- node[midway, fill=white] {$\cdots$} (M_last);
		\draw[->] (M_last) -- ($ (M_last.east) + (0.5, 0) $) node[right] {$\mathbf{Y}$};
	\end{tikzpicture}
\end{center}
Let $I_i \subseteq \{ 1, \dots, n \}$ be a subset of indices for each $i = n_s, \dots, n_e + 1$, and let $\mathcal{I}$ be a subset of $\Delta$.
If (i) for each $i = n_s, \dots, n_e$, $M$ is $(\mathcal{I} \cup \{ Y_j \}_{j \in I_i}, \{ Y'_j \}_{j \in I_{i + 1}})$-NI where $\mathbf{Y}' = (Y'_1, \dots, Y'_n)$ is the output of $M$, and (ii) for each $j \in I_{n_s}$, free variables of $E_j$ are in $\mathcal{I}$, then the for-loop is $(\mathcal{I}, \{ Y_j \}_{j \in I_{n_e + 1}})$-NI where $\mathbf{Y} = (Y_1, \dots, Y_n)$ is used as the output of the for-loop.
This gives the following proof rule.
\begin{mathpar}
	\inferrule{
		\Own \mathbf{Y}|_{n_s} \LilacEntails \Own \mathbf{E}|_{n_s} \\
		\HoareTripleOwn{\top}{M}{\mathbf{Y}'. \mkern-24mu \Conditioning_{(x, y) \gets ( \mathcal{I}, \mathbf{Y}|_i )} \mkern-24mu (\Own ( \Delta \backslash \mathcal{I}, \mathbf{Y} \backslash \mathbf{Y}|_i ) * \Own \mathbf{Y}'|_{i + 1})} \\
	}{
		{\begin{aligned}
			\llbrace \top \rrbrace\ 
			&\forloop{i}{n_s}{n_e}{\mathbf{Y}}{E}{M} \\
			&\llbrace \mathbf{Y}. \Conditioning_{x \gets \mathcal{I}} (\Own (\Delta \backslash \mathcal{I}) * \Own \mathbf{Y}|_{n_e + 1}) \rrbrace
		\end{aligned}}
	}
\end{mathpar}
Here, we write $\mathbf{Y}|_i$ for the tuple $(Y_j)_{j \in I_i}$ and similarly for $\mathbf{Y}'|_i$ and $\mathbf{E}|_i$.
We can prove this either from Theorem~\ref{thm:noninterference-conditional-independence} or by semi-graphoid axioms.

\subsection{Weakening}

When composing programs, a component program often uses only a subset of input variables.
The situation can be illustrated as follows.
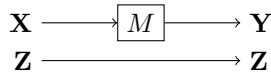
\begin{center}
	\begin{tikzpicture}
		\node[draw] (M) at (0, 0) {$M$};
		\draw[->] ($ (M.west) + (-1, 0) $) node[left] {$\mathbf{X}$} -- ($ (M.west) + (0, 0) $);
		\draw[->] ($ (M.east) + (0, 0) $) -- ($ (M.east) + (1, 0) $) node[right] {$\mathbf{Y}$};
		\draw[->] ($ (M.west) + (-1, 0) + (0, -0.5) $) node[left] {$\mathbf{Z}$} -- ($ (M.east) + (1, 0) + (0, -0.5) $) node[right] {$\mathbf{Z}$};
	\end{tikzpicture}
\end{center}
In this case, we want to reason about the component program $M$ with the subset $\mathbf{X}$ of input variables and then extend the result to the whole set of input variables $\mathbf{X}, \mathbf{Z}$.
Specifically, suppose that $M$ is $(\mathcal{I}, \mathcal{O})$-NI where $\mathcal{I} \subseteq \mathbf{X}$ and $\mathcal{O} \subseteq \mathbf{Y}$.
Then, we want to conclude that $M$ is also $(\mathcal{I} \cup \mathcal{Z}, \mathcal{O} \cup \mathcal{Z})$-NI for any $\mathcal{Z} \subseteq \mathbf{Z}$.
This gives rise to the following proof rule.
\begin{mathpar}
	\inferrule{
		\HoareTripleOwn{\top}{M}{\mathbf{Y}. \Conditioning_{x \gets \mathcal{I}} (\Own (\mathbf{X} \setminus \mathcal{I}) * \Own \mathcal{O})}
	}{
		\forall_{\mathrm{rv}} \mathbf{Z}. \HoareTripleOwn{\top}{M}{\mathbf{Y}. \mkern-20mu \Conditioning_{(x, z) \gets (\mathcal{I}, \mathcal{Z})} \mkern-20mu (\Own (\mathbf{X} \backslash \mathcal{I}, \mathbf{Z} \backslash \mathcal{Z}) * \Own (\mathcal{O}, \mathcal{Z}))}
	}
\end{mathpar}
Note that $\mathbf{Z}$ need not be independent of $\mathbf{X}$.
Hence, this is a different local reasoning principle from the frame rule in probabilistic separation logic.

\section{Applications to Probing Security}
\label{sec:examples}

Raccoon~\cite{DelPino2023} is a masking-friendly lattice-based signature scheme.
Raccoon employs a sum-of-uniforms distribution as its noise source to address the difficulty of generating discrete Gaussian noise in a side-channel-resistant manner.
An algorithm called \textsc{AddRepNoise} is used in Raccoon to add uniform noise to masked values.

In this section, we demonstrate how Lilac can be used to prove the probing security of \textsc{AddRepNoiseER}, a variant of \textsc{AddRepNoise} that is used in the security proof of Raccoon~\cite{DelPinoCRYPTO2024}.
The algorithm \textsc{AddRepNoiseER} (Fig.~\ref{fig:addrepnoiseer}) adds uniform noise $\rho_{i, i_{\mathrm{rep}}, j}$ to each share $V_i^j$ using \textsc{MiniAddRepNoise} (Fig.~\ref{fig:miniaddrepnoise}), and then refreshes the shares using the \textsc{Refresh} algorithm (Fig.~\ref{fig:refreshm}).
This procedure is repeated $r$ times for each of the $l$ shared variables.
The uniform noise $\rho_{i, i_{\mathrm{rep}}, j}$ is sampled externally and provided to \textsc{AddRepNoiseER} as an unshared input (ER stands for Explicit Randomness).
Below, we prove that \textsc{AddRepNoiseER} satisfies the $t$-SNIU property (Definition~\ref{def:t-noninterference-with-unshared-input}), which plays a central role in the security proof of Raccoon~\cite{DelPinoCRYPTO2024}.
Following~\cite[Lemma~6.1]{DelPinoCRYPTO2024}, the proof is divided into three parts.

After proving the security of \textsc{AddRepNoiseER}, we present another example in which we prove the $t$-SNI property of a masked multiplication algorithm~\cite{IshaiCRYPTO2003}.

\begin{figure}[t]
	\begin{algorithmic}
		\Function{AddRepNoiseER}{$\{ \mathbf{V}_i \}_{i \in \RangeOneTo{l}}, \rho$}
			\State $\{ \mathbf{X}_{i, 0} \}_{i \in \RangeOneTo{l}} \gets \{ \mathbf{V}_i \}_{i \in \RangeOneTo{l}}$
			\For{$i = 1, \dots, l$}
				\For{$i_{\mathrm{rep}} = 1, \dots, r$}
					\State $\mathbf{Y}_{i, i_{\mathrm{rep}}} \gets \Call{MiniAddRepNoise}{\mathbf{X}_{i, i_{\mathrm{rep}} - 1}, \rho_{i, i_{\mathrm{rep}}}}$
					\State $\mathbf{X}_{i, i_{\mathrm{rep}}} \gets \Call{Refresh}{\mathbf{Y}_{i, i_{\mathrm{rep}}}}$
				\EndFor
			\EndFor
			\State \Return{$\{ \mathbf{X}_{i, r} \}_{i \in \RangeOneTo{l}}$}
		\EndFunction
	\end{algorithmic}
	\caption{AddRepNoiseER~\cite{DelPinoCRYPTO2024} in conventional pseudocode. Here, $\RangeOneTo{l} = \{ 1, 2, \dots, l \}$, and $\rho = \{\rho_{i, i_{\mathrm{rep}}}\}_{i \in \RangeOneTo{l}, i_{\mathrm{rep}} \in \RangeOneTo{r}}$ where $\rho_{i, i_{\mathrm{rep}}} = \{\rho_{i, i_{\mathrm{rep}}, j}\}_{j \in \{ 0, \dots, t \}}$ and $\rho_{i, i_{\mathrm{rep}}, j} : \mathbb{K}$.}
	\label{fig:addrepnoiseer}
\end{figure}

\subsection{$t$-NIU of \textsc{MiniAddRepNoise}}

\begin{figure}[t]
	\begin{algorithmic}
		\Function{MiniAddRepNoise}{$\mathbf{V}, \rho$}
			\For{$j = 0, \dots, t$}
				\State $X^j \gets V^j + \rho_j$
			\EndFor
			\State \Return{$\mathbf{X}$}
		\EndFunction
	\end{algorithmic}
	\caption{MiniAddRepNoise in conventional pseudocode.}
	\label{fig:miniaddrepnoise}
\end{figure}

\begin{figure}
	\begin{align*}
		&\textcolor{gray}{\{\ \Own (\mathbf{V}, \rho)\ \}} \\
		&\textcolor{gray}{\{ I(0; \mathbf{0})\ \}} \\
		&\mathtt{for}\ ( \mathbf{X} \gets \mathbf{0};\ j \gets 0, \dots, t )\ \{\\
		&\qquad \textcolor{gray}{\{\ I(j; \mathbf{X})\ \}} \\
		&\qquad \tikzmarknode{Frame3Begin}{\textcolor{gray}{\{\ \top\ \}}} \\
		&\qquad \ret{\mathbf{X}[V^j + \rho_j / X^j]}\\
		&\qquad \tikzmarknode{Frame3End}{\textcolor{gray}{\{\ \mathbf{X}'.\ \mathbf{X}' \AlmostSurelyEqual \mathbf{X}[V^j + \rho_j / X^j]\ \}}} \\
		&\qquad \textcolor{gray}{\{\ \mathbf{X}'.\ \mathbf{X}' \AlmostSurelyEqual \mathbf{X}[V^j + \rho_j / X^j] * I(j; \mathbf{X})\ \}} \\
		&\qquad \textcolor{gray}{\{\ \mathbf{X}'.\ I(j + 1; \mathbf{X}')\ \}} \\
		&\} \\
		& \textcolor{gray}{\{\ \mathbf{X}.\ I(t + 1; \mathbf{X})\ \}} \\
		& \textcolor{gray}{\{\ \mathbf{X}.\ \Conditioning_{(v, r) \gets (\mathbf{V} \cap \mathcal{I}, \rho \cap \mathcal{I})} (\Own (\mathbf{V} \setminus \mathcal{I}, \rho \setminus \mathcal{I}) * \Own \mathcal{O})\ \}}
	\end{align*}
	\[ I(j; \mathbf{X}) \quad\coloneqq\quad \Own (\mathbf{V}, \rho, \mathbf{X}) * \bigast_{k < j} X^k \AlmostSurelyEqual V^k + \rho_k \]
	\begin{tikzpicture}[overlay,remember picture]
		\draw[gray, rounded corners] ($(Frame3Begin.north west)+(-0.1,0.1)$) rectangle ($(Frame3End.south east)+(1.1,-0.1)$);
		\node[gray, anchor=north east] at ($(Frame3Begin.north west -| Frame3End.south east) + (1.1, 0.1)$) {\textsc{H-Frame}};
	\end{tikzpicture}
	\caption{Algorithm \textsc{MiniAddRepNoise} in APPL. Gray lines show the proof outline for Theorem~\ref{thm:miniaddrepnoise-ni}. We apply the \textsc{H-For} rule to the loop with the loop invariant $I(j; \mathbf{X})$.}
	\label{fig:miniaddrepnoise-appl}
\end{figure}

We first prove the noninterference property of \textsc{MiniAddRepNoise}~\cite{DelPinoCRYPTO2024}, which is a subroutine used in \textsc{AddRepNoiseER}.
The program in conventional pseudocode is shown in Fig.~\ref{fig:miniaddrepnoise} and its APPL version with proof outline is shown in Fig.~\ref{fig:miniaddrepnoise-appl}.
This algorithm takes a shared input $\mathbf{V} = (V^0, \dots, V^t)$ and an unshared input $\rho = (\rho_0, \dots, \rho_t)$.
The output is $\mathbf{X} = (X^0, \dots, X^t)$, and there are no internal variables.

\begin{theorem}\label{thm:miniaddrepnoise-ni}
	Let $\mathcal{O} \subseteq \mathbf{X}$ be a set of probed variables such that $|\mathcal{O}| \le t$.
	We define $\mathcal{I} = \mathcal{I}^{\mathrm{shared}}_{\mathrm{MARN}}(\mathcal{O}) \cup \mathcal{I}^{\mathrm{unshared}}_{\mathrm{MARN}}(\mathcal{O})$ by
	$\mathcal{I}^{\mathrm{shared}}_{\mathrm{MARN}}(\mathcal{O}) \coloneqq \{ V^j \mid X^j \in \mathcal{O} \}$ and $\mathcal{I}^{\mathrm{unshared}}_{\mathrm{MARN}}(\mathcal{O}) \coloneqq \{ \rho_j \mid X^j \in \mathcal{O} \}$ (here, ``MARN'' stands for MiniAddRepNoise).
	Then, the following Hoare triple is derivable in Lilac, meaning that \textsc{MiniAddRepNoise} is $(\mathcal{I}, \mathcal{O})$-NI.
	\begin{align*}
		&\{ \Own (\mathbf{V}, \rho) \} \\
		&\mathtt{for}\ ( \mathbf{X} \gets \mathbf{0};\ j \gets 0, \dots, t )\ \{\ \ret{\mathbf{X}[V^j + \rho_j / X^j]}\ \} \\
		& \{ \Conditioning_{(v, r) \gets (\mathbf{V} \cap \mathcal{I}, \rho \cap \mathcal{I})} (\Own (\mathbf{V} \setminus \mathcal{I}, \rho \setminus \mathcal{I}) * \Own \mathcal{O}) \}
	\end{align*}
\end{theorem}
\begin{proofsketch}
	The outline of the proof is shown in Fig.~\ref{fig:miniaddrepnoise-appl}.
	It is straightforward to see that the following postcondition holds after executing the program by considering the loop invariant $I(j; X)$ defined in the figure.
	\[ (\Own \mathbf{V} \land \Own \rho) * \mathbf{X} \AlmostSurelyEqual \mathbf{V} + \rho \]
	We can show that this implies the desired postcondition.
	By \hyperlink{rule:Gen-Weak-Union}{\textsc{Gen-Weak-Union}}, we have the following entailment.
	\begin{align}
		&\Own (\mathbf{V}, \rho) * \mathbf{X} \AlmostSurelyEqual \mathbf{V} + \rho \\
		&\LilacEntails \Conditioning_{(v, r) \gets (\mathbf{V} \cap \mathcal{I}, \rho \cap \mathcal{I})} (\Own (\mathbf{V}, \rho) * \mathbf{X} \AlmostSurelyEqual \mathbf{V} + \rho)
	\end{align}
	By \hyperlink{rule:C-Subst}{\textsc{C-Subst}}, we also have the following entailment.
	\[ \Own (\mathbf{V}, \rho) \LilacEntails \Conditioning_{(v, r) \gets (\mathbf{V} \cap \mathcal{I}, \rho \cap \mathcal{I})} \bigwedge_{X^j \in \mathcal{O}} (v^j \AlmostSurelyEqual V^j \land r_j \AlmostSurelyEqual \rho_j) \]
	Then, it follows that the following entailment holds.
	\begin{align}
		&\Own (\mathbf{V}, \rho) * \mathbf{X} \AlmostSurelyEqual \mathbf{V} + \rho \\
		&\LilacEntails \!\!\!\!\Conditioning_{(v, r) \gets (\mathbf{V} \cap \mathcal{I}, \rho \cap \mathcal{I})} (\Own (\mathbf{V} \setminus \mathcal{I}, \rho \setminus \mathcal{I}) \land \!\!\!\! \bigwedge_{X^j \in \mathcal{O}} X^j \AlmostSurelyEqual v^j + r_j)
	\end{align}
	Since $X^j \AlmostSurelyEqual v^j + r_j$ is deterministic, we have $\Persistent \Own X^j$ for each $X^j \in \mathcal{O}$.
	Thus, we obtain the desired postcondition.
	Details of the proof are given in
	\iflong
	Appendix~\ref{sec:proof-miniaddrepnoise-ni}.
	\fi
\end{proofsketch}
\begin{toappendix}
\subsection{Detailed Proof of Theorem~\ref{thm:miniaddrepnoise-ni}}
\label{sec:proof-miniaddrepnoise-ni}
\end{toappendix}
\begin{appendixproof}[Proof of Theorem~\ref{thm:miniaddrepnoise-ni}]
	The proof outline is shown in Fig.~\ref{fig:miniaddrepnoise-appl}.
	Here, we show non-trivial entailments used in the proof.
	We often use \hyperlink{rule:Eq-Conj}{\textsc{Eq-Conj}} implicitly in the following.
	\begin{itemize}
		\item $\mathbf{X}' \AlmostSurelyEqual \mathbf{X}[V^j + \rho_j / X^j] * I(j; \mathbf{X}) \LilacEntails I(j + 1; \mathbf{X}')$:
		By definition of $I(j; \mathbf{X})$, the premise is equivalent to
		\[ \mathbf{X}' \AlmostSurelyEqual \mathbf{X}[V^j + \rho_j / X^j] * \Own (\mathbf{V}, \rho, \mathbf{X}) * \bigast_{k < j} X^k \AlmostSurelyEqual V^k + \rho_k. \]
		\begin{itemize}
			\item For each $k \neq j$, we have $\Own {X'}^k$ by \hyperlink{rule:Transfer-Own}{\textsc{Transfer-Own}}.
			\[ {X'}^k \AlmostSurelyEqual X^k \land \Own X^j \LilacEntails \Own {X'}^k \]
			\item We have $\Own {X'}^j$ by \hyperlink{rule:Transfer-Own}{\textsc{Transfer-Own}} and \hyperlink{rule:FV-Own}{\textsc{FV-Own}}.
			\begin{align*}
				&\Own V^k \land \Own \rho_j \land {X'}^k \AlmostSurelyEqual V^k + \rho_j \\
				&\LilacEntails \Own (V^k + \rho_j) \land {X'}^k \AlmostSurelyEqual V^k + \rho_j \\
				&\LilacEntails \Own {X'}^k
			\end{align*}
			\item For each $k < j$, we have ${X'}^k \AlmostSurelyEqual V^k + \rho_k$ by transitivity of $\AlmostSurelyEqual$.
			\[ X^k \AlmostSurelyEqual V^k + \rho_k \land {X'}^k \AlmostSurelyEqual X^k \LilacEntails {X'}^k \AlmostSurelyEqual V^k + \rho_k \]
			\item We immediately have ${X'}^j \AlmostSurelyEqual V^j + \rho_j$ from the premise.
		\end{itemize}
		Therefore, we obtain the following entailments.
		\begin{align*}
			&\mathbf{X}' \AlmostSurelyEqual \mathbf{X}[V^j + \rho_j / X^j] * \Own (\mathbf{V}, \rho, \mathbf{X}) * \bigast_{k < j} X^k \AlmostSurelyEqual V^k + \rho_k \\
			&\LilacEntails \Own (\mathbf{V}, \rho) \land \Own \mathbf{X}' \land \bigwedge_{k < j} ({X'}^k \AlmostSurelyEqual V^k + \rho_k) \land {X'}^j \AlmostSurelyEqual V^j + \rho_j \\
			&\LilacEntails I(j + 1; \mathbf{X}')
		\end{align*}
		\item $I(t + 1; \mathbf{X}) \LilacEntails \Conditioning_{(v, r) \gets (\mathbf{V} \cap \mathcal{I}, \rho \cap \mathcal{I})} (\Own (\mathbf{V} \setminus \mathcal{I}, \rho \setminus \mathcal{I}) * \Own \mathcal{O})$:
		\begin{align}
			&\Own (\mathbf{V}, \rho) * \mathbf{X} \AlmostSurelyEqual \mathbf{V} + \rho \\
			&\LilacEntails \Conditioning_{(v, r) \gets (\mathbf{V} \cap \mathcal{I}, \rho \cap \mathcal{I})} (\Own (\mathbf{V}, \rho) * \mathbf{X} \AlmostSurelyEqual \mathbf{V} + \rho) &\text{by \hyperlink{rule:Gen-Weak-Union}{\textsc{Gen-Weak-Union}}} \\
			&\qquad\qquad \land \Conditioning_{(v, r) \gets (\mathbf{V} \cap \mathcal{I}, \rho \cap \mathcal{I})} \bigwedge_{X^j \in \mathcal{O}} (v^j \AlmostSurelyEqual V^j \land r_j \AlmostSurelyEqual \rho_j) &\text{by \hyperlink{rule:C-Subst}{\textsc{C-Subst}}} \\
			&\LilacEntails \Conditioning_{(v, r) \gets (\mathbf{V} \cap \mathcal{I}, \rho \cap \mathcal{I})} ((\Own (\mathbf{V}, \rho) * \mathbf{X} \AlmostSurelyEqual \mathbf{V} + \rho) \land \bigwedge_{X^j \in \mathcal{O}} (v^j \AlmostSurelyEqual V^j \land r_j \AlmostSurelyEqual \rho_j)) & \text{by \hyperlink{rule:C-And-Distributive}{\textsc{C-And-Distributive}}}
		\end{align}
		\begin{itemize}
			\item By \hyperlink{rule:FV-Own}{\textsc{FV-Own}} and \hyperlink{rule:Persistent-Mono}{\textsc{$\Persistent$-Mono}} and \hyperlink{rule:Persistent-True}{\textsc{$\Persistent$-True}}, we have $\LilacEntails \Persistent \Own (v^j + r_j)$ for each $X^j \in \mathcal{O}$.
			\item For each $X^j \in \mathcal{O}$, we also have the following entailment.
			\begin{align*}
				&\Own V^j \land \Own \rho_j \land v^j \AlmostSurelyEqual V^j \land r_j \AlmostSurelyEqual \rho_j \\
				&\LilacEntails \Own (v^j + r_j, V^j + \rho_j) \land v^j \AlmostSurelyEqual V^j \land r_j \AlmostSurelyEqual \rho_j &\text{by \hyperlink{rule:FV-Own}{\textsc{FV-Own}}} \\
				&\LilacEntails v^j + r_j \AlmostSurelyEqual V^j + \rho_j &\text{\hyperlink{rule:Eq-Congruence}{\textsc{Eq-Congruence}}}
			\end{align*}
		\end{itemize}
		Thus, we continue as follows.
		\begin{align*}
			&\Conditioning_{(v, r) \gets (\mathbf{V} \cap \mathcal{I}, \rho \cap \mathcal{I})} ((\Own (\mathbf{V}, \rho) * \mathbf{X} \AlmostSurelyEqual \mathbf{V} + \rho) \land \bigwedge_{X^j \in \mathcal{O}} (v^j \AlmostSurelyEqual V^j \land r_j \AlmostSurelyEqual \rho_j)) \\
			&\LilacEntails \Conditioning_{(v, r) \gets (\mathbf{V} \cap \mathcal{I}, \rho \cap \mathcal{I})} (\Own (\mathbf{V} \setminus \mathcal{I}, \rho \setminus \mathcal{I}) \land \mathbf{X} \AlmostSurelyEqual \mathbf{V} + \rho \land \bigwedge_{X^j \in \mathcal{O}} (\Persistent \Own (v^j + r_j) \land v^j + r_j \AlmostSurelyEqual V^j + \rho_j)) \\
			&\LilacEntails \Conditioning_{(v, r) \gets (\mathbf{V} \cap \mathcal{I}, \rho \cap \mathcal{I})} (\Own (\mathbf{V} \setminus \mathcal{I}, \rho \setminus \mathcal{I}) \land \bigwedge_{X^j \in \mathcal{O}} (\Persistent \Own (v^j + r_j) \land v^j + r_j \AlmostSurelyEqual X^j)) &\text{by transitivity} \\
			&\LilacEntails \Conditioning_{(v, r) \gets (\mathbf{V} \cap \mathcal{I}, \rho \cap \mathcal{I})} (\Own (\mathbf{V} \setminus \mathcal{I}, \rho \setminus \mathcal{I}) * \bigwedge_{X^j \in \mathcal{O}} (\Persistent \Own (v^j + r_j) \land v^j + r_j \AlmostSurelyEqual X^j)) &\text{by \hyperlink{rule:Persistent-Conj}{\textsc{$\Persistent$-Conj}}} \\
			&\LilacEntails \Conditioning_{(v, r) \gets (\mathbf{V} \cap \mathcal{I}, \rho \cap \mathcal{I})} (\Own (\mathbf{V} \setminus \mathcal{I}, \rho \setminus \mathcal{I}) * \bigwedge_{X^j \in \mathcal{O}} (\Own (v^j + r_j) \land v^j + r_j \AlmostSurelyEqual X^j)) &\text{by \hyperlink{rule:Persistent-Elim}{\textsc{$\Persistent$-Elim}}} \\
			&\LilacEntails \Conditioning_{(v, r) \gets (\mathbf{V} \cap \mathcal{I}, \rho \cap \mathcal{I})} (\Own (\mathbf{V} \setminus \mathcal{I}, \rho \setminus \mathcal{I}) * \Own \mathcal{O}) &\text{by \hyperlink{rule:Transfer-Own}{\textsc{Transfer-Own}}}
		\end{align*}
	\end{itemize}
\end{appendixproof}

\begin{remark}
	As explained in Section~\ref{sec:probing-model}, reasoning about noninterference properties in our approach requires exposing internal variables.
	This typically does not undermine the local reasoning principle of separation logic, because the internal variables exposed in earlier steps are not used in subsequent steps.
	Moreover, the frame rule allows us to frame out postconditions concerning those internal variables.
	This is illustrated by the use of \hyperlink{rule:H-Frame}{\textsc{H-Frame}} in Fig.~\ref{fig:miniaddrepnoise-appl}.
\end{remark}

\begin{corollary}
	\textsc{MiniAddRepNoise} is $t$-NIU.
\end{corollary}
\begin{proof}
	Immediate from Theorem~\ref{thm:miniaddrepnoise-ni}.
\end{proof}

\subsection{$t$-SNI of \textsc{Refresh}}

\textsc{Refresh} is an algorithm shown in Fig.~\ref{fig:refreshm} and used in masked algorithms to ensure noninterference~\cite{BartheCCS2016,DelPinoCRYPTO2024}.
The input is a shared variable $\mathbf{A} = (A^0, \dots, A^t)$, and the output is $\mathbf{C}_t = (C^0_t, \dots, C^t_t)$ such that $\mathbf{C}_t$ and $\mathbf{A}$ represent the same secret value.
The internal variables are $\{ R_{i, j} \}_{0 \le i < j \le t}$ and $\{ C^i_j \}_{0 \le i \le t, 0 \le j < t}$.

\begin{figure}[t]
	\begin{algorithmic}
		\Function{Refresh}{$\mathbf{A}$}
			\For{$i = 0, \dots, t$}
				\State $C^i_0 \gets A^i$
			\EndFor
			\For{$i = 0, \dots, t$}
				\For{$j = i + 1, \dots, t$}
					\State $R_{i, j} \gets \mathtt{unif}\ \mathbb{K}$
					\State $C^i_j \gets C^i_{j - 1} + R_{i, j}$
					\State $C^j_{i + 1} \gets C^j_i - R_{i, j}$
				\EndFor
			\EndFor
			\State \Return{$\mathbf{C}_t = (C^0_t, \dots, C^t_t)$}
		\EndFunction
	\end{algorithmic}
	\caption{Algorithm \textsc{Refresh} \cite{BartheCCS2016} in conventional pseudocode. Subscripts for $C$ are used to conform to single-assignment style but may be omitted if multiple assignments to $C$ are allowed.}
	\label{fig:refreshm}
\end{figure}

\begin{figure*}[t]
	\begin{align*}
		&\textcolor{gray}{\{\ \Own \mathbf{A}\ \}} \\
		&\tikzmarknode{Frame1Begin}{\textcolor{gray}{\{\ I(0, 1; (\mathbf{A}, \mathbf{0}, \dots, \mathbf{0}), \mathbf{0})\ \}}} \\
		&\mathtt{for}\ (\ (\mathbf{C}', \mathbf{R}') \gets ((\mathbf{A}, \mathbf{0}, \dots, \mathbf{0}), \mathbf{0});\ i \gets 0, \dots, t\ )\ \{ \\
		&\qquad\mathtt{for}\ (\ (\mathbf{C}, \mathbf{R}) \gets (\mathbf{C}', \mathbf{R}');\ j \gets i + 1, \dots, t\ )\ \{ \\
		&\qquad\qquad \textcolor{gray}{\{\ I(i, j; \mathbf{C}, \mathbf{R})\ \}} \\
		&\qquad\qquad \tikzmarknode{Frame2Begin}{\textcolor{gray}{\{\ \top\ \}}} \\
		&\qquad\qquad S \gets \mathtt{unif}\ \mathbb{K}; \\
		&\qquad\qquad \textcolor{gray}{\{\ S.\ S \sim \Uniform{\mathbb{K}}\ \}} \\
		&\qquad\qquad \ret{(\mathbf{C}, \mathbf{R})[ C^i_{j - 1} + S / C^i_j, C^j_i - S / C^j_{i + 1}, S / R_{i, j}]} \\
		&\qquad\qquad \tikzmarknode{Frame2End}{\textcolor{gray}{\{\ (\mathbf{C}'', \mathbf{R}'').\ S \sim \Uniform{\mathbb{K}} * (\mathbf{C}'', \mathbf{R}'') \AlmostSurelyEqual (\mathbf{C}, \mathbf{R})[ C^i_{j - 1} + S / C^i_j, C^j_i - S / C^j_{i + 1}, S / R_{i, j}]\ \}}} \\
		&\qquad\qquad \textcolor{gray}{\{\ (\mathbf{C}'', \mathbf{R}'').\ S \sim \Uniform{\mathbb{K}} * (\mathbf{C}'', \mathbf{R}'') \AlmostSurelyEqual (\mathbf{C}, \mathbf{R})[ C^i_{j - 1} + S / C^i_j, C^j_i - S / C^j_{i + 1}, S / R_{i, j}] * I(i, j; \mathbf{C}, \mathbf{R})\ \}} \\
		&\qquad\qquad \textcolor{gray}{\{\ (\mathbf{C}'', \mathbf{R}'').\ I(i, j + 1; \mathbf{C}'', \mathbf{R}'')\ \}} \\
		&\qquad\} \\
		&\} \\
		&\tikzmarknode{Frame1End}{\textcolor{gray}{\{\ (\mathbf{C}, \mathbf{R}).\ I(t, t + 1; \mathbf{C}, \mathbf{R})\ \}}} \\
		&\textcolor{gray}{\{\ (\mathbf{C}, \mathbf{R}).\ \Own \vec{A} * \bigast_{i < j} (R_{i, j} \sim \Uniform{\mathbb{K}}) * \bigast_{i, j} (C^i_j \AlmostSurelyEqual A^i + f^i_j(R))\ \}} \\
		&\textcolor{gray}{\{\ (\mathbf{C}, \mathbf{R}).\ \Conditioning_{a \gets \mathcal{I}} (\Own (\mathbf{A} \setminus \mathcal{I}) * \Own \mathcal{O})\ \}}
	\end{align*}
	\[ I(i, j; \mathbf{C}, \mathbf{R}) \quad\coloneqq\quad \Own \mathbf{A} * \bigast_{i} C^i_0 \AlmostSurelyEqual A^i * \bigast_{\substack{(i', j') <_{\mathrm{lex}} (i, j) \\ 0 \le i' < j' \le t}} \Big(R_{i', j'} \sim \Uniform{\mathbb{K}} * C^{i'}_{j'} \AlmostSurelyEqual A^{i'} + f^{i'}_{j'}(\mathbf{R}) * C^{j'}_{i' + 1} \AlmostSurelyEqual A^{j'} + f^{j'}_{i' + 1}(\mathbf{R})\Big) \]
	\begin{tikzpicture}[overlay,remember picture]
		\draw[gray, rounded corners] ($(Frame2Begin.north west)+(-0.1,0.1)$) rectangle ($(Frame2End.south east)+(0.1,-0.1)$);
		\node[gray, anchor=north east] at ($(Frame2Begin.north west -| Frame2End.south east) + (0.1, 0.1)$) {\textsc{H-Frame}};
	\end{tikzpicture}
	\caption{Algorithm \textsc{Refresh} in APPL. Here, $\mathbf{R}$ is a shorthand for the collection of random variables $R_{i, j}$, and $\mathbf{C} = (\mathbf{C}_0, \dots, \mathbf{C}_t)$ where $\mathbf{C}_i = (C_i^0, \dots, C_i^t)$ for each $i$. Gray comments show the proof outline for Theorem~\ref{thm:refreshm-ni}. In the invariant, $<_{\mathrm{lex}}$ denotes the lexicographic order on pairs of integers $(i, j)$ such that $i < j$, and $f^i_j(\mathbf{R})$ is a function defined in~\eqref{eq:refreshm-inv-func}.}
	\label{fig:refreshm-appl}
\end{figure*}

\begin{theorem}\label{thm:refreshm-ni}
	Consider \textsc{Refresh} in Fig.~\ref{fig:refreshm}.
	Let $\mathcal{O} \subseteq \{ R_{i, j} \mid i, j \in \{ 0, \dots, t \}, i < j \} \cup \{ C^i_j \mid i, j \in \{ 0, \dots, t \} \}$ be a set of probed variables such that $|\mathcal{O}| \le t$.
	We define $\mathcal{I} = \mathcal{I}_{\mathrm{Refresh}}(\mathcal{O}) \subseteq \{ A^0, \dots, A^t \}$ as follows.
	\[ \mathcal{I}_{\mathrm{Refresh}}(\mathcal{O}) \quad\coloneqq\quad \{ A^i \mid \exists j. R_{i, j} \in \mathcal{O} \lor \exists j < t. C^i_j \in \mathcal{O} \} \]
	Then, the following Hoare triple is derivable in Lilac, meaning that \textsc{Refresh} is $(\mathcal{I}, \mathcal{O})$-NI.
	\[ \HoareTriple{\Own \mathbf{A}}{\text{\textsc{Refresh}}}{(\mathbf{C}, \mathbf{R}). \Conditioning_{a \gets \mathcal{I}} (\Own (\mathbf{A} \setminus \mathcal{I}) * \Own \mathcal{O})} \]
\end{theorem}
\begin{proofsketch}
	The outline of the proof is shown in Fig.~\ref{fig:refreshm-appl}.
	It is straightforward to show that the following postcondition holds.
	\[ \Own \vec{A} * \bigast_{i < j} (R_{i, j} \sim \Uniform{\mathbb{K}}) * \bigast_{i, j} (C^i_j \AlmostSurelyEqual A^i + f^i_j(R)) \]
	where $f^i_j(\mathbf{R})$ is a function defined as follows.
	\begin{equation}
		f^i_j(\mathbf{R}) \coloneqq \begin{cases}
			- \sum_{k < j} R_{k, i} & i \ge j \\
			- \sum_{k < i} R_{k, i} + \sum_{i < k \le j} R_{i, k} & i < j
		\end{cases}
		\label{eq:refreshm-inv-func}
	\end{equation}
	We can derive the desired postcondition by translating the proof in~\cite{BartheCCS2016} into Lilac.
	\iflong
	See Appendix~\ref{sec:proof-refreshm-ni} for details.
	\fi
\end{proofsketch}
\begin{toappendix}
\subsection{Detailed Proof of Theorem~\ref{thm:refreshm-ni}}
\label{sec:proof-refreshm-ni}
\end{toappendix}
\begin{appendixproof}[Proof of Theorem~\ref{thm:refreshm-ni}]
	The outline of the proof is shown in Fig.~\ref{fig:refreshm-appl}.
	We show the non-trivial entailments used in the proof.
	We often use \hyperlink{rule:Eq-Conj}{\textsc{Eq-Conj}} implicitly in the following.
	\begin{itemize}
		\item $S \sim \Uniform{\mathbb{K}} * (\mathbf{C}'', \mathbf{R}'') \AlmostSurelyEqual (\mathbf{C}, \mathbf{R})[ C^i_{j - 1} + S / C^i_j, C^j_i - S / C^j_{i + 1}, S / R_{i, j}] * I(i, j; \mathbf{C}, \mathbf{R}) \LilacEntails I(i, j + 1; \mathbf{C}'', \mathbf{R}'')$: by \hyperlink{rule:Transfer-Dist}{\textsc{Transfer-Dist}} and \hyperlink{rule:Eq-Congruence}{\textsc{Eq-Congruence}}.
		Note that $f^i_j$ is defined so that the following holds for each $i, j$ such that $i < j$.
		\[ f^i_j(\mathbf{R}) = f^i_{j - 1}(\mathbf{R}) + R_{i, j}, \qquad f^j_{i + 1}(\mathbf{R}) = f^j_{i}(\mathbf{R}) - R_{i, j} \]
		\item $\Own \vec{A} * \bigast_{i < j} (R_{i, j} \sim \Uniform{\mathbb{K}}) * \bigast_{i, j} (C^i_j \AlmostSurelyEqual A^i + f^i_j(R)) \LilacEntails \Conditioning_{a \gets \mathcal{I}} (\Own (\mathbf{A} \setminus \mathcal{I}) * \Own \mathcal{O})$:
		It suffices to prove the following.
		\begin{equation}
			\Own \mathbf{A} * \bigast_{i < j} (R_{i, j} \sim \Uniform{\mathbb{K}}) \LilacEntails \Conditioning_{a \gets \mathcal{I}} (\Own (\vec{A} \setminus \mathcal{I}) * \phi_{\Own}^{\mathcal{O}})
			\label{eq:refreshm-ni-sufficient-condition}
		\end{equation}
		where
		\begin{align}
			\phi_{\Own}^{\mathcal{O}} &\coloneqq \bigwedge_{R_{i, j} \in \mathcal{O}} \Own R_{i, j} \land \bigwedge_{C^i_j \in \mathcal{O}, j < t} \Own (a^i + f^i_j(R)) \land \bigwedge_{C^i_t \in \mathcal{O}} \Own (A^i + f^i_t(R))
		\end{align}
		because we have the following entailments then.
		\begin{align*}
			&\Own \vec{A} * \bigast_{i < j} (R_{i, j} \sim \Uniform{\mathbb{K}}) * \bigast_{i, j} (C^i_j \AlmostSurelyEqual A^i + f^i_j(R)) \\
			&\LilacEntails \Conditioning_{a \gets \mathcal{I}} (\Own (\vec{A} \setminus \mathcal{I}) * \phi_{\Own}^{\mathcal{O}}) \land \Big(\Own \mathcal{I} * \bigast_{i, j} (C^i_j \AlmostSurelyEqual A^i + f^i_j(R))\Big) \land \Conditioning_{a \gets \mathcal{I}} \bigwedge_{A^i \in \mathcal{I}} (a^i \AlmostSurelyEqual A^i) &\text{by \hyperlink{rule:C-Subst}{\textsc{C-Subst}}} \\
			&\LilacEntails \Conditioning_{a \gets \mathcal{I}} (\Own (\vec{A} \setminus \mathcal{I}) * \phi_{\Own}^{\mathcal{O}}) \land \Conditioning_{a \gets \mathcal{I}} \bigast_{i, j} (C^i_j \AlmostSurelyEqual A^i + f^i_j(R)) \land \Conditioning_{a \gets \mathcal{I}} \bigwedge_{A^i \in \mathcal{I}} (a^i \AlmostSurelyEqual A^i) &\text{by \hyperlink{rule:C-Indep}{\textsc{C-Indep}}} \\
			&\LilacEntails \Conditioning_{a \gets \mathcal{I}} ((\Own (\vec{A} \setminus \mathcal{I}) * \phi_{\Own}^{\mathcal{O}}) \land \bigast_{i, j} (C^i_j \AlmostSurelyEqual A^i + f^i_j(R)) \land \bigwedge_{A^i \in \mathcal{I}} (a^i \AlmostSurelyEqual A^i)) &\text{by \hyperlink{rule:C-And-Distributive}{\textsc{C-And-Distributive}}} \\
			&\LilacEntails \Conditioning_{a \gets \mathcal{I}} (\Own (\vec{A} \setminus \mathcal{I}) * \Big(\bigwedge_{R_{i, j} \in \mathcal{O}} \Own R_{i, j} \land \bigwedge_{C^i_j \in \mathcal{O}, j < t} (\Own (a^i + f^i_j(R)) \land C^i_j \AlmostSurelyEqual a^i + f^i_j(R)) \\
			&\qquad\qquad\land \bigwedge_{C^i_t \in \mathcal{O}} (\Own (A^i + f^i_t(R)) \land C^i_j \AlmostSurelyEqual A^i + f^i_j(R)) \Big)) \\
			&\LilacEntails \Conditioning_{a \gets \mathcal{I}} (\Own (\vec{A} \setminus \mathcal{I}) * \Own \mathcal{O}) &\text{by \hyperlink{rule:Transfer-Own}{\textsc{Transfer-Own}}}
		\end{align*}
		Now, we prove~\eqref{eq:refreshm-ni-sufficient-condition}.
		Observe that free variables in $f^i_j(R)$ are the first $j$ elements of
		\begin{equation}
			R_{0, i},\ R_{1, i},\ \dots,\ R_{i - 1, i},\ R_{i, i + 1},\ R_{i, i + 2},\ \dots,\ R_{i, t}
			\label{eq:refreshm-free-variables}
		\end{equation}
		Let $\vec{S} \subseteq \vec{R}$ be the set of random variables defined as follows.
		\[ \vec{S} \coloneqq \{ R_{i, j} \mid R_{i, j} \in \mathcal{O} \lor \text{$R_{i, j}$ is a free variable in $f^i_j(\vec{R})$ for some $C^i_j \in \mathcal{O}$ such that $j < t$} \} \]
		Then, the intersection of $\vec{S}$ and \eqref{eq:refreshm-free-variables} is a subset of $\{ R_{i, j} \mid i < j \land A^j \in \mathcal{I} \} \cup \{ R_{j, i} \mid j < i \land A^j \in \mathcal{I} \}$, which has at most $|\mathcal{I}|$ elements.
		Since $|\mathcal{I}| + |\{ C^i_j \in \mathcal{O} \mid j = t \}| \le |\mathcal{O}| \le t$, there exists an injective mapping from $C^i_t \in \mathcal{O}$ to $R_{j, k} \in \vec{R} \setminus \vec{S}$ such that $j = i$ or $k = i$ (i.e., random variables in~\eqref{eq:refreshm-free-variables} excluding those in $\vec{S}$).
		We write such $R_{j, k}$ as $R[C^i_t]$.
		Then, for each $C^i_t \in \mathcal{O}$, we can see that $A^i + f^i_t(R)$ is a bijective function of $R[C^i_t]$ when other variables are fixed.
		Thus, we have the following entailments.
		\begin{align}
			&\Own \mathbf{A} * \bigast_{i < j} (R_{i, j} \sim \Uniform{\mathbb{K}}) \\
			&\LilacEntails \Conditioning_{a \gets \mathcal{I}} \Bigg( \Own \vec{A} * \bigast_{i < j} (R_{i, j} \sim \Uniform{\mathbb{K}}) \Bigg) &\text{by \hyperlink{rule:Gen-Weak-Union}{\textsc{Gen-Weak-Union}}} \\
			&\LilacEntails \Conditioning_{a \gets \mathcal{I}} \Bigg( \Own (\vec{A} \setminus \mathcal{I}) * \bigast_{C^i_t \in \mathcal{O}} (R[C^i_t] \sim \Uniform{\mathbb{K}}) * \bigast_{R_{i, j} \in \vec{S}} (R_{i, j} \sim \Uniform{\mathbb{K}}) \Bigg) \\
			&\LilacEntails \Conditioning_{a \gets \mathcal{I}} \Bigg( \Own (\vec{A} \setminus \mathcal{I}) * \bigast_{C^i_t \in \mathcal{O}} (A^i + f^i_t(R) \sim \Uniform{\mathbb{K}}) * \bigast_{R_{i, j} \in \vec{S}} (R_{i, j} \sim \Uniform{\mathbb{K}}) \Bigg) &\text{by Lemma~\ref{lem:generalised-bijective-substitution-uniform}}
		\end{align}
		Moreover, since expressions in $\bigwedge_{R_{i, j} \in \mathcal{O}} \Own R_{i, j} \land \bigwedge_{C^i_j \in \mathcal{O}, j < t} \Own (a^i + f^i_j(R))$ contain only random variables in $\vec{S}$ and deterministic variables, we have the following entailments.
		\begin{align}
			&\bigast_{R_{i, j} \in \vec{S}} (R_{i, j} \sim \Uniform{\mathbb{K}}) \\
			&\LilacEntails \Own \vec{S} &\text{by \hyperlink{rule:Dist-Own}{\textsc{Dist-Own}}} \\
			&\LilacEntails \bigwedge_{R_{i, j} \in \mathcal{O}} \Own R_{i, j} \land \bigwedge_{C^i_j \in \mathcal{O}, j < t} \Own (a^i + f^i_j(R)) &\text{by \hyperlink{rule:FV-Own}{\textsc{FV-Own}}}
			&\qedhere
		\end{align}
	\end{itemize}
\end{appendixproof}

\begin{corollary}
	\textsc{Refresh} is $t$-SNI.
\end{corollary}
\begin{proof}
	By Theorem~\ref{thm:refreshm-ni}, for each set of probes $\mathcal{O}$ with $|\mathcal{O}| \le t$, there exists a set of inputs $\mathcal{I}$ with $|\mathcal{I}|$ smaller than or equal to internal probes in $\mathcal{O}$ such that \textsc{Refresh} is $(\mathcal{I}, \mathcal{O})$-NI.
\end{proof}

\subsection{$t$-SNIU of \textsc{AddRepNoise} with Explicit Randomness}
\textsc{AddRepNoiseER} is shown in Fig.~\ref{fig:addrepnoiseer}.
Here, we write $\RangeOneTo{l} = \{ 1, 2, \dots, l \}$.
The shared input is $\mathbf{V}_i = (V^0_i, \dots, V^t_i)$ for $i = 1, \dots, n$, and the unshared input is $\rho = \{\rho_{i, i_{\mathrm{rep}}}\}_{i \in \RangeOneTo{l}, i_{\mathrm{rep}} \in \RangeOneTo{r}}$ where $\rho_{i, i_{\mathrm{rep}}} = \{\rho_{i, i_{\mathrm{rep}}, j}\}_{j \in \{ 0, \dots, t \}}$ and $\rho_{i, i_{\mathrm{rep}}, j} : \mathbb{K}$.
For each $i$ and $i_{\mathrm{rep}}$, $\mathbf{X}_{i, i_{\mathrm{rep}}} = (X^0_{i, i_{\mathrm{rep}}}, \dots, X^t_{i, i_{\mathrm{rep}}})$ and $\mathbf{Y}_{i, i_{\mathrm{rep}}} = (Y^0_{i, i_{\mathrm{rep}}}, \dots, Y^t_{i, i_{\mathrm{rep}}})$ are used as internal or output variables.

We can prove that \textsc{AddRepNoiseER} in Fig.~\ref{fig:addrepnoiseer} is $t$-SNIU by using the composability of NI discussed in Section~\ref{sec:additional-rules}.

\begin{theorem}
	Let $\mathcal{O}$ be a set of probed variables such that $|\mathcal{O}| \le t$.
	We write $\mathcal{O}^{\mathrm{Refresh}}_{i, i_{\mathrm{rep}}}$ and $\mathcal{O}^{\mathrm{MARN}}_{i, i_{\mathrm{rep}}}$ for the subsets of $\mathcal{O}$ consisting of probed variables in the execution of $\textsc{Refresh}$ and $\textsc{MiniAddRepNoise}$ in the $(i, i_{\mathrm{rep}})$-th iteration, respectively.
	For each $i$ and $i_{\mathrm{rep}}$, we define $\mathcal{J}_{i, i_{\mathrm{rep}}} \subseteq \mathbf{Y}_{i, i_{\mathrm{rep}}}$, $\mathcal{I}^{\mathrm{shared}}_{i, i_{\mathrm{rep}}} \subseteq \mathbf{X}_{i, i_{\mathrm{rep}} - 1}$, and $\mathcal{I}^{\mathrm{unshared}}_{i, i_{\mathrm{rep}}} \subseteq \rho_{i, i_{\mathrm{rep}}}$ as follows.
	\begin{align*}
		\mathcal{J}_{i, i_{\mathrm{rep}}} &\quad\coloneqq\quad \mathcal{I}_{\mathrm{Refresh}}(\mathcal{I}^{\mathrm{shared}}_{i, i_{\mathrm{rep}} + 1} \cup \mathcal{O}^{\mathrm{Refresh}}_{i, i_{\mathrm{rep}}}) \\
		\mathcal{I}^{\mathrm{shared}}_{i, i_{\mathrm{rep}}} &\quad\coloneqq\quad \mathcal{I}^{\mathrm{shared}}_{\mathrm{MARN}}(\mathcal{J}_{i, i_{\mathrm{rep}}} \cup \mathcal{O}^{\mathrm{MARN}}_{i, i_{\mathrm{rep}}}) \\
		\mathcal{I}^{\mathrm{unshared}}_{i, i_{\mathrm{rep}}} &\quad\coloneqq\quad \mathcal{I}^{\mathrm{shared}}_{\mathrm{MARN}}(\mathcal{J}_{i, i_{\mathrm{rep}}} \cup \mathcal{O}^{\mathrm{MARN}}_{i, i_{\mathrm{rep}}})
	\end{align*}
	Here, we define $\mathcal{I}^{\mathrm{shared}}_{i, r + 1} \coloneqq \emptyset$.
	Let $\mathcal{I}$ be defined as follows.
	\[ \mathcal{I} \quad\coloneqq\quad \{ V^k_i \mid X^k_{i, 0} \in \mathcal{I}^{\mathrm{shared}}_{i, 0} \} \cup \bigcup_{i, i_{\mathrm{rep}}} \mathcal{I}^{\mathrm{unshared}}_{i, i_{\mathrm{rep}}} \]
	Then, \textsc{AddRepNoiseER} is $(\mathcal{I}, \mathcal{O})$-NI.
\end{theorem}
\begin{proof}
	By using composability for sequential composition and for-loop with Theorem~\ref{thm:refreshm-ni} and~\ref{thm:miniaddrepnoise-ni}.
\end{proof}

\begin{corollary}
	\textsc{AddRepNoiseER} is $t$-SNIU.
	\qed
\end{corollary}

\subsection{$t$-SNI of \textsc{SecMult}}

The multiplication algorithm~\cite{IshaiCRYPTO2003} is another basic building block for masked algorithms.
In Fig.~\ref{fig:sec-mult}, \textsc{SecMult} takes two shared inputs $\mathbf{A} = (A^0, \dots, A^t)$ and $\mathbf{B} = (B^0, \dots, B^t)$ and returns a shared output $\mathbf{C}_t = (C^0_t, \dots, C^t_t)$ such that $\mathbf{C}_t$ represents the product of $\mathbf{A}$ and $\mathbf{B}$, i.e., $\sum_{i = 0}^t C^i_t = (\sum_{i = 0}^t A^i) \times (\sum_{i = 0}^t B^i)$.
It is proved in~\cite{BartheCCS2016} that \textsc{SecMult} is $t$-SNI.
We show that by Theorem~\ref{thm:noninterference-conditional-independence}, this result can be proved in Lilac as well.

\begin{figure}[t]
	\begin{algorithmic}
		\Function{SecMult}{$\mathbf{A}, \mathbf{B}$}
			\For{$i = 0, \dots, t$}
				\For{$j = 0, \dots, t$}
					\State $P_{i, j} \gets A^i \times B^j$
				\EndFor
			\EndFor
			\State $(C^0_0, \dots, C^t_0) \gets (P_{0, 0}, \dots, P_{t, t})$
			\For{$i = 0, \dots, t$}
				\For{$j = i + 1, \dots, t$}
					\State $Q_{i, j} \gets \mathtt{unif}\ \mathbb{K}$
					\State $C^i_j \gets C^i_{j - 1} - Q_{i, j}$
					\State $R_{i, j} \gets Q_{i, j} + P_{i, j}$
					\State $S_{i, j} \gets R_{i, j} + P_{j, i}$
					\State $C^j_{i + 1} \gets C^j_i + S_{i, j}$
				\EndFor
			\EndFor
			\State \Return{$\mathbf{C}_t = (C^0_t, \dots, C^t_t)$}
		\EndFunction
	\end{algorithmic}
	\caption{Algorithm \textsc{SecMult} in conventional pseudocode. Unlike the presentation in~\cite{BartheCCS2016}, our version computes $P_{i, j} \gets A^i \times B^j$ for all $i, j$ in advance, since we need to assign a distinct variable name to each $A^i \times B^j$ for the proof.}
	\label{fig:sec-mult}
\end{figure}

\begin{figure*}
	\begin{align*}
		&\textcolor{gray}{\{\ \Own (\mathbf{A}, \mathbf{B})\ \}} \\
		&\mathbf{P} \gets \mathtt{for}\ ( \mathbf{P} \gets \mathbf{0};\ i \gets 0, \dots, t )\ \{\\
		&\qquad \mathtt{for}\ ( \mathbf{P}' \gets \mathbf{P};\ j \gets 0, \dots, t )\ \{\ \ret{\mathbf{P}'[A^i \times B^j / P'_{i, j}]} \ \} \\
		&\}; \\
		&\textcolor{gray}{\{\ \mathbf{P}. \Own (\mathbf{A}, \mathbf{B}) * \bigast_{i, j} P_{i, j} \AlmostSurelyEqual A^i \times B^j \ \}} \\
		&\textcolor{gray}{\{\ \mathbf{P}. \Own (\mathbf{A}, \mathbf{B}) * \bigast_{i, j} P_{i, j} \AlmostSurelyEqual A^i \times B^j * I(0, 1; ((P_{i, i})_{i = 0}^{t}, \mathbf{0}, \dots, \mathbf{0}), \mathbf{0}, \mathbf{0}, \mathbf{0}) \ \}} \\
		&(\mathbf{C}, \mathbf{Q}, \mathbf{R}, \mathbf{S}) \gets \mathtt{for}\ ( (\mathbf{C}, \mathbf{Q}, \mathbf{R}, \mathbf{S}) \gets (((P_{i, i})_{i = 0}^{t}, \mathbf{0}, \dots, \mathbf{0}), \mathbf{0}, \mathbf{0}, \mathbf{0});\ i \gets 0, \dots, t )\ \{\\
		&\qquad \mathtt{for}\ ( (\mathbf{C}', \mathbf{Q}', \mathbf{R}', \mathbf{S}') \gets (\mathbf{C}, \mathbf{Q}, \mathbf{R}, \mathbf{S});\ j \gets i + 1, \dots, t )\ \{\\
		&\qquad\qquad\textcolor{gray}{\{\ \Own (\mathbf{A}, \mathbf{B}) * \bigast_{i, j} P_{i, j} \AlmostSurelyEqual A^i \times B^j * I(i, j; \mathbf{C}', \mathbf{Q}', \mathbf{R}', \mathbf{S}') \ \}} \\
		&\qquad\qquad Q''_{i, j} \gets \Uniform{\mathbb{K}}; \\
		&\qquad\qquad {C''}^i_j \gets \ret{({C'}^i_{j - 1} - Q''_{i, j})};
		\ R''_{i, j} \gets \ret{(Q''_{i, j} + P_{i, j})};
		\ S''_{i, j} \gets \ret{(R''_{i, j} + P_{j, i})};
		\ {C''}^j_{i + 1} \gets \ret{({C'}^j_{i} + S''_{j, i})};\\
		&\qquad\qquad \ret{(\mathbf{C}', \mathbf{Q}', \mathbf{R}', \mathbf{S}')[{C''}^i_j / {C'}^i_j, {C''}^j_{i + 1} / {C'}^j_{i + 1}, Q''_{i, j} / Q'_{i, j}, R''_{i, j} / R'_{i, j}, S''_{j, i} / S'_{j, i}]} \\
		&\qquad\qquad\textcolor{gray}{\{\ (\mathbf{C}'', \mathbf{Q}'', \mathbf{R}'', \mathbf{S}''). \Own (\mathbf{A}, \mathbf{B}) * \bigast_{i, j} P_{i, j} \AlmostSurelyEqual A^i \times B^j * I(i, j + 1; \mathbf{C}'', \mathbf{Q}'', \mathbf{R}'', \mathbf{S}'') \ \}} \\
		&\qquad \} \\
		&\}; \\
		&\textcolor{gray}{\{\ \Own (\mathbf{A}, \mathbf{B}) * \bigast_{i, j} P_{i, j} \AlmostSurelyEqual A^i \times B^j * I(t, t + 1; \mathbf{C}, \mathbf{Q}, \mathbf{R}, \mathbf{S}) \ \}} \\
		&\ret{(\mathbf{P}, \mathbf{C}, \mathbf{Q}, \mathbf{R}, \mathbf{S})} \\
		& \textcolor{gray}{\{\ (\mathbf{P}, \mathbf{C}, \mathbf{Q}, \mathbf{R}, \mathbf{S}).\ \Own (\mathbf{A}, \mathbf{B}) * \bigast_{i, j} P_{i, j} \AlmostSurelyEqual A^i \times B^j * \bigast_{i < j} Q_{i, j} \sim \Uniform{\mathbb{K}} * \bigast_{i < j} R_{i, j} \AlmostSurelyEqual Q_{i, j} + P_{i, j}} \\
		& \qquad\qquad\qquad\qquad\qquad \textcolor{gray}{* \bigast_{i < j} S_{i, j} \AlmostSurelyEqual R_{i, j} + P_{j, i} * \bigast_{i, j} C^i_j \AlmostSurelyEqual P_{i, i} + g^i_j(\mathbf{Q}, \mathbf{S}) \ \}} \\
		& \textcolor{gray}{\{\ (\mathbf{P}, \mathbf{C}, \mathbf{Q}, \mathbf{R}, \mathbf{S}).\ \Conditioning_{(a, b) \gets (\mathcal{I}, \mathcal{J})} (\Own (\mathbf{A} \setminus \mathcal{I}, \mathbf{B} \setminus \mathcal{J}) * \Own \mathcal{O})\ \}}
	\end{align*}
	\begin{align*}
		I(i, j; \mathbf{C}, \mathbf{Q}, \mathbf{R}, \mathbf{S}) \quad\coloneqq\quad \mathbf{C}_0 \AlmostSurelyEqual (P_{i, i})_{i = 0}^{t} * \bigast_{\substack{(i', j') <_{\mathrm{lex}} (i, j) \\ 0 \le i' < j' \le t}} &\Big( Q_{i', j'} \sim \Uniform{\mathbb{K}} * R_{i', j'} \AlmostSurelyEqual Q_{i', j'} + P_{i', j'} * S_{i', j'} \AlmostSurelyEqual R_{i', j'} + P_{j', i'} \\
		&* C^{i'}_{j'} \AlmostSurelyEqual P_{i', i'} + g^{i'}_{j'}(\mathbf{Q}, \mathbf{S}) * C^{j'}_{i' + 1} \AlmostSurelyEqual P_{j', j'} + g^{j'}_{i + 1'}(\mathbf{Q}, \mathbf{S}) \Big)
	\end{align*}
	\caption{Algorithm \textsc{SecMult} in APPL.
	Here, $\mathbf{Q}$, $\mathbf{R}$, and $\mathbf{S}$ are shorthands for the collection of random variables $\{ Q_{i, j} \mid 0 \le i < j \le t \}$, $\{ R_{i, j} \mid 0 \le i < j \le t \}$, and $\{ S_{i, j} \mid 0 \le i < j \le t \}$, respectively; and $\mathbf{C} = (\mathbf{C}_0, \dots, \mathbf{C}_t)$ where $\mathbf{C}_i = (C_i^0, \dots, C_i^t)$ for each $i$.
	Gray comments show the proof outline for Theorem~\ref{thm:sec-mult}, but we omit some straightforward proof steps for brevity.
	The function $g^i_j(\mathbf{Q}, \mathbf{S})$ is defined as in~\eqref{eq:sec-mult-c}.
	The loop invariant for the second loop is denoted as $I(i, j; \mathbf{C}, \mathbf{Q}, \mathbf{R}, \mathbf{S})$.
	}
	\label{fig:sec-mult-appl}
\end{figure*}

\begin{theorem}\label{thm:sec-mult}
	Let $\mathcal{O}$ be a set of probed variables.
	\begin{align}
		I_1 &\coloneqq \{ i \mid P_{i, j} \in \mathcal{O} \} \cup \{ i \mid C^i_j \in \mathcal{O} \land j < t \} \\
		J_1 &\coloneqq \{ j \mid P_{i, j} \in \mathcal{O} \} \cup \{ i \mid C^i_j \in \mathcal{O} \land j < t \} \\
		(I_2, J_2) &\coloneqq \mathrm{Add}(I_1, J_1;\ \{ (i, j) \mid Q_{i, j} \in \mathcal{O} \lor R_{i, j} \in \mathcal{O} \}) \\
		(I_3, J_3) &\coloneqq \mathrm{Add}(I_2, J_2;\ \{ (i, j) \mid S_{i, j} \in \mathcal{O} \})
	\end{align}
	Here, for each $I, J \subseteq \{ 0, \dots, t \}$ and $K \subseteq \{ 0, \dots, t \}^2$, $\mathrm{Add}(I, J; K)$ is the pair $(I', J')$ defined as follows.
	\begin{align*}
		I' &\coloneqq I \cup \{ i \mid (i, j) \in K \land i \notin I \} \cup \{ j \mid (i, j) \in K \land i \in I \} \\
		J' &\coloneqq J \cup \{ j \mid (i, j) \in K \land j \notin J \} \cup \{ i \mid (i, j) \in K \land j \in J \}
	\end{align*}
	Let $I \coloneqq I_3$, $J \coloneqq J_3$, and $\mathcal{I} = \{ A^i \mid i \in I \}$, $\mathcal{J} = \{ B^j \mid j \in J \}$.
	Then, \textsc{SecMult} is $(\mathcal{I} \cup \mathcal{J}, \mathcal{O})$-NI.
\end{theorem}
\begin{proofsketch}
	We rephrase the simulator-based proof in~\cite{BartheCCS2016} as a conditional-independence-based proof in Lilac.
	We first apply the rules in Fig.~\ref{fig:lilac-hoare-rules} to obtain a postcondition that expresses how each variable is computed from the inputs and the random variables.
	This step is rather straightforward, and the proof sketch is shown in Fig.~\ref{fig:sec-mult-appl}.
	Here, $g^i_j(\mathbf{Q}, \mathbf{S})$ is defined as follows.
	\begin{equation}
		g^i_j(\mathbf{Q}, \mathbf{S}) \coloneqq \begin{cases}
			\sum_{k < j} S_{k, i} & i \ge j \\
			\sum_{k < i} S_{k, i} + \sum_{i < k \le j} Q_{i, k} & i < j 
		\end{cases}
		\label{eq:sec-mult-c}
	\end{equation}
	Then, we apply the \hyperlink{rule:H-Conseq}{\textsc{H-Conseq}} rule to obtain the desired postcondition, which is the non-trivial part of the proof.

	We divide (the indices of) the uniform random variables $\mathbf{Q}$ into four parts: $K_{\mathbf{Q}}$, $K_{\mathbf{S}} \setminus K_{\mathbf{Q}}$, $K_{\mathbf{C}}$, and the rest.
	Here, we define $K_{\mathbf{Q}}, K_{\mathbf{S}} \subseteq \{ (i, j) \mid 0 \le i < j \le t \}$ as follows.
	\begin{align*}
		K_{\mathbf{Q}} &\coloneqq \{ (i, j) \mid Q_{i, j} \in \mathcal{O} \lor R_{i, j} \in \mathcal{O} \} \\
		&\qquad \cup \{ (i, k) \mid C^i_j \in \mathcal{O}, j < t, i < k \} \\
		K_{\mathbf{S}} &\coloneqq \{ (i, j) \mid S_{i, j} \in \mathcal{O} \} \\
		&\qquad \cup \{ (k, i) \mid C^i_j \in \mathcal{O}, j < t, k < i \}
	\end{align*}
	The set $K_{\mathbf{C}}$ is defined so that (a) $K_{\mathbf{C}}$ bijectively corresponds to the set $\{ C^i_t \in \mathcal{O} \mid \lnot (\exists j < t, C^i_j \in \mathcal{O}) \}$ and (b) for each $(i', j') \in K_{\mathbf{C}}$ corresponds to $C^i_t$, the uniform random variable $Q_{i', j'}$ occurs in $g^i_t(\mathbf{Q}, \mathbf{S})$.
	We can define such $K_{\mathbf{C}}$ because we have $|\mathcal{O} \cap \mathbf{C}_t| \le |\mathbf{V}_i \setminus (K_{\mathbf{Q}} \cup K_{\mathbf{S}})|$ for each $C^i_t \in \mathcal{O}$ with no $j < t$ such that $C^i_j \in \mathcal{O}$, where $\mathbf{V}_i = \{ (k, i) \mid k < i \} \cup \{ (i, k) \mid i < k \}$ is the indices of the free variables in $g^i_t(\mathbf{Q}, \mathbf{S})$.

	We can then show that (a) $\{ C^i_t \in \mathcal{O} \mid \lnot (\exists j < t, C^i_j \in \mathcal{O}) \}$ can be ``simulated'' by sampling from the uniform distribution, (b) $\{ S_{i, j} \mid (i, j) \in K_{\mathbf{S}} \setminus K_{\mathbf{Q}} \}$ can be ``simulated'' by sampling from the uniform distribution, and (c) the rest of the variables in $\mathcal{O}$ can be ``simulated'' from $\{ S_{i, j} \mid (i, j) \in K_{\mathbf{S}} \setminus K_{\mathbf{Q}} \}$, $\{ Q_{i, j} \mid (i, j) \in K_{\mathbf{Q}} \}$, and $\mathcal{I} \cup \mathcal{J}$.
	In Lilac, (a) and (b) correspond to applying \hyperlink{rule:Unif-Bijection}{\textsc{Unif-Bijection}}, and (c) corresponds to applying \hyperlink{rule:Transfer-Own}{\textsc{Transfer-Own}}.
	\iflong
	See Appendix~\ref{sec:proof-sec-mult} for details.
	\fi
\end{proofsketch}

\begin{toappendix}
\subsection{Proof of Theorem~\ref{thm:sec-mult}}
\label{sec:proof-sec-mult}
\begin{lemma}
	Let $E$ be a deterministic expression such that $\Gamma; {\cdot} \vdash E : A$.
	Then, we have
	\[ X \AlmostSurelyEqual E \quad\LilacEntails\quad \Own X \]
\end{lemma}
\begin{proof}
	\begin{align*}
		&X \AlmostSurelyEqual E \\
		&\LilacEntails X \AlmostSurelyEqual E \land \Own E & \text{by \hyperlink{rule:FV-Own}{\textsc{FV-Own}}} \\
		&\LilacEntails \Own X & \text{by \hyperlink{rule:Transfer-Own}{\textsc{Transfer-Own}}} & \qedhere
	\end{align*}
\end{proof}
\begin{lemma}
	Let $\Gamma; \Delta, X : A \vdash E : B$ be an expression.
	Then, we have
	\[ (\Own \Delta \land \Own X) * \phi \quad\LilacEntails\quad \Conditioning_{x \gets X} \big((\Own \Delta \land \Own X) * \phi * E \AlmostSurelyEqual E[x/X]\big) \]
\end{lemma}
\begin{proof}
	\begin{align*}
		&(\Own \Delta \land \Own X) * \phi \\
		&\LilacEntails \Conditioning_{x \gets X} \big((\Own \Delta \land \Own X) * \phi \big) \land \Conditioning_{x \gets X} X \AlmostSurelyEqual x & \text{by \hyperlink{rule:Gen-Weak-Union}{\textsc{Gen-Weak-Union}} and \hyperlink{rule:C-Subst}{\textsc{C-Subst}}} \\
		&\LilacEntails \Conditioning_{x \gets X} \big((\Own \Delta \land \Own X) * \phi \land X \AlmostSurelyEqual x\big) & \text{by \hyperlink{rule:C-And-Distributive}{\textsc{C-And-Distributive}}} \\
		&\LilacEntails \Conditioning_{x \gets X} \big((\Own \Delta \land \Own X) * \phi \land \Own (E, E[x/X]) \land X \AlmostSurelyEqual x\big) & \text{by \hyperlink{rule:FV-Own}{\textsc{FV-Own}}} \\
		&\LilacEntails \Conditioning_{x \gets X} \big((\Own \Delta \land \Own X) * \phi \land E \AlmostSurelyEqual E[x/X]\big) & \text{by \hyperlink{rule:Eq-Congruence}{\textsc{Eq-Congruence}}} \\
		&\LilacEntails \Conditioning_{x \gets X} \big((\Own \Delta \land \Own X) * \phi * E \AlmostSurelyEqual E[x/X]\big) & \text{by \hyperlink{rule:Eq-Conj}{\textsc{Eq-Conj}}} & \qedhere
	\end{align*}
\end{proof}
\end{toappendix}

\begin{appendixproof}[Proof of Theorem~\ref{thm:sec-mult}]
	$I$, $J$, $K_{\mathbf{Q}}$, and $K_{\mathbf{S}}$ satisfy the following conditions:
	\begin{enumerate}
		\item $|I|, |J| \le |\mathcal{O} \setminus \mathbf{C}_t|$
		\item If $C^i_j \in \mathcal{O}$ and $j < t$, then $i \in I \cap J$. \label{cond:sec-mult-c-i}
		\item If $i < j$, $C^i_j \in \mathcal{O}$, and $i < k \le j < t$, then $(i, k) \in K_{\mathbf{Q}}$.
		\item If $C^i_j \in \mathcal{O}$, $k < i$, $j < t$, and $(k, i) \in K_{\mathbf{Q}}$, then $k, i \in I$ and $k, i \in J$. \label{cond:sec-mult-c-q}
		\item If $Q_{i, j} \in \mathcal{O}$, then $(i, j) \in K_{\mathbf{Q}}$.
		\item If $R_{i, j} \in \mathcal{O}$, then $(i, j) \in K_{\mathbf{Q}}$, $i \in I$, and $j \in J$.
		\item If $(i, j) \in K_{\mathbf{Q}}$ and $S_{i, j} \in \mathcal{O}$, then $i, j \in I$ and $i, j \in J$. \label{cond:sec-mult-s-from-q}
		\item If $C^i_t \in \mathcal{O}$ and there exists $j < t$ with $C^i_j \in \mathcal{O}$, then $\mathbf{V}_i \subseteq K_{\mathbf{Q}} \cup K_{\mathbf{S}}$. \label{cond:sec-mult-c-t-simulated}
		\item If $C^i_t \in \mathcal{O}$ and there is no $j < t$ with $C^i_j \in \mathcal{O}$, then $\mathbf{V}_i \setminus (K_{\mathbf{Q}} \cup K_{\mathbf{S}}) \neq \emptyset$.
		\label{cond:sec-mult-c-uniform}
	\end{enumerate}
	Here, $\mathbf{V}_i = \{ (k, i) \mid k < i \} \cup \{ (i, k) \mid i < k \}$ is the indices of the free variables in $g^i_t(\mathbf{Q}, \mathbf{S})$.
	Condition~\ref{cond:sec-mult-c-uniform} follows because $|\mathbf{V}_i \cap (K_{\mathbf{Q}} \cup K_{\mathbf{S}})| \le |\mathcal{O} \setminus \mathbf{C}_t| < |\mathbf{V}_i| = t$.
	\begin{align*}
		&\mathbf{V}_i \cap (K_{\mathbf{Q}} \cup K_{\mathbf{S}})\\
		&= \mathbf{V}_i \cap \{ (i, j) \mid Q_{i, j} \in \mathbf{Q} \lor R_{i, j} \in \mathcal{O} \lor S_{i, j} \in \mathcal{O} \} \\
		&\qquad \cup \{ (k, i) \mid k < i, C^k_j \in \mathcal{O}, j < t \} \\
		&\qquad \cup \{ (i, k) \mid i < k, C^k_j \in \mathcal{O}, j < t \}
	\end{align*}
	In other words, we have $|\mathcal{O} \cap \mathbf{C}_t| \le |\mathbf{V}_i \setminus (K_{\mathbf{Q}} \cup K_{\mathbf{S}})|$.
	Thus, for each $C^i_t \in \mathcal{O}$ with no $j < t$ such that $C^i_j \in \mathcal{O}$, we take an element of $\mathbf{V}_i \setminus (K_{\mathbf{Q}} \cup K_{\mathbf{S}})$ in such a way that the chosen elements are distinct for distinct $C^i_t$'s, and define the set $K_{\mathbf{C}}$ as the set of the chosen elements.

	By applying the rules in Fig.~\ref{fig:lilac-hoare-rules}, we obtain the following postcondition (Fig.~\ref{fig:sec-mult-appl}).
	\[ \Own (\mathbf{A}, \mathbf{B}) * \bigast_{i, j} P_{i, j} \AlmostSurelyEqual A^i \times B^j * \bigast_{i < j} Q_{i, j} \sim \Uniform{\mathbb{K}} * \bigast_{i < j} R_{i, j} \AlmostSurelyEqual Q_{i, j} + P_{i, j} * \bigast_{i < j} S_{i, j} \AlmostSurelyEqual R_{i, j} + P_{j, i} * \bigast_{i, j} C^i_j \AlmostSurelyEqual P_{i, i} + g^i_j(\mathbf{Q}, \mathbf{S}) \]

	By applying \hyperlink{rule:Gen-Weak-Union}{\textsc{Gen-Weak-Union}} and \hyperlink{rule:C-Subst}{\textsc{C-Subst}}, we have the following postcondition.
	\begin{align*}
		&\Conditioning_{(a, b) \gets (\mathcal{I}, \mathcal{J})} \big( \Own (\mathbf{A}, \mathbf{B}) * \bigast_{i, j} P_{i, j} \AlmostSurelyEqual A^i \times B^j * \bigast_{i < j} Q_{i, j} \sim \Uniform{\mathbb{K}} * \bigast_{i < j} R_{i, j} \AlmostSurelyEqual Q_{i, j} + P_{i, j} * \bigast_{i < j} S_{i, j} \AlmostSurelyEqual R_{i, j} + P_{j, i} \\
		&\qquad\qquad * \bigast_{i, j} C^i_j \AlmostSurelyEqual P_{i, i} + g^i_j(\mathbf{Q}, \mathbf{S}) * \bigast_{i \in I} A^i = a^i * \bigast_{j \in J} B^j = b^j \big)
	\end{align*}
	It follows that the above postcondition implies:
	\begin{align*}
		&\Conditioning_{(a, b) \gets (\mathcal{I}, \mathcal{J})} \Big( \big(\Own (\mathbf{A}, \mathbf{B}) * \bigast_{i < j} Q_{i, j} \sim \Uniform{\mathbb{K}}\big) \land \big( \underline{\Own (\mathbf{A}, \mathbf{B}, \mathbf{Q}) \land \bigwedge_{i, j} P_{i, j} \AlmostSurelyEqual A^i \times B^j \land \bigwedge_{i < j} R_{i, j} \AlmostSurelyEqual Q_{i, j} + P_{i, j}} \\
		&\qquad\qquad \underline{\land \bigwedge_{i < j} S_{i, j} \AlmostSurelyEqual R_{i, j} + P_{j, i} \land \bigwedge_{i, j} C^i_j \AlmostSurelyEqual P_{i, i} + g^i_j(\mathbf{Q}, \mathbf{S}) \land \bigwedge_{i \in I} A^i \AlmostSurelyEqual a^i \land \bigwedge_{j \in J} B^j \AlmostSurelyEqual b^j} \big) \Big)
	\end{align*}
	Then, we focus on the underlined part; rewrite the equations for $P_{i, j}$, $R_{i, j}$, $S_{i, j}$, and $C^i_j$; and obtain equations for $\mathcal{O}$.
	\begin{itemize}
		\item $P_{i, j}$:
		\begin{align*}
			&\bigwedge_{i, j} P_{i, j} \AlmostSurelyEqual A^i \times B^j \land \bigwedge_{i \in I} A^i \AlmostSurelyEqual a^i \land \bigwedge_{j \in J} B^j \AlmostSurelyEqual b^j \\
			&\LilacEntails \bigwedge_{(i, j) \in I \times J} P_{i, j} \AlmostSurelyEqual a^i \times b^j \\
			&\LilacEntails \bigwedge_{P_{i, j} \in \mathcal{O}} P_{i, j} \AlmostSurelyEqual a^i \times b^j
		\end{align*}
		\item $R_{i, j}$:
		\begin{align*}
			&\Own \mathbf{Q} \land \bigwedge_{i < j} R_{i, j} \AlmostSurelyEqual Q_{i, j} + P_{i, j} \land \bigwedge_{(i, j) \in I \times J} P_{i, j} \AlmostSurelyEqual a^i \times b^j \\
			&\LilacEntails \bigwedge_{R_{i, j} \in \mathcal{O}} R_{i, j} \AlmostSurelyEqual Q_{i, j} + a^i \times b^j \\
		\end{align*}
		\item $S_{i, j}$:
		\begin{align*}
			&\Own (\mathbf{A}, \mathbf{B}, \mathbf{Q}) \land \bigwedge_{i, j} P_{i, j} \AlmostSurelyEqual A^i \times B^j \land \bigwedge_{i < j} R_{i, j} \AlmostSurelyEqual Q_{i, j} + P_{i, j} \land \bigwedge_{i < j} S_{i, j} \AlmostSurelyEqual R_{i, j} + P_{j, i} \\
			&\LilacEntails \bigwedge_{i < j} S_{i, j} \AlmostSurelyEqual Q_{i, j} + A^i \times B^j + A^j \times B^i
		\end{align*}
		If $S_{i, j} \in \mathcal{O}$ and $(i, j) \in K_{\mathbf{Q}}$, then $i, j \in I$ and $i, j \in J$ by Condition~\ref{cond:sec-mult-s-from-q}, and thus:
		\begin{align*}
			&\bigwedge_{i < j} S_{i, j} \AlmostSurelyEqual Q_{i, j} + A^i \times B^j + A^j \times B^i \land \bigwedge_{i \in I} A^i \AlmostSurelyEqual a^i \land \bigwedge_{j \in J} B^j \AlmostSurelyEqual b^j \\
			&\LilacEntails \bigwedge_{(i, j) \in K_{\mathbf{Q}}, S_{i, j} \in \mathcal{O}} S_{i, j} \AlmostSurelyEqual Q_{i, j} + a^i \times b^j + a^j \times b^i
		\end{align*}
		\item $C^i_j$:
		First note that we have $C^i_j \AlmostSurelyEqual A^i \times B^i + g^i_j(\mathbf{Q}, \mathbf{S})$ from $C^i_j \AlmostSurelyEqual P_{i, i} + g^i_j(\mathbf{Q}, \mathbf{S})$ and $P_{i, i} \AlmostSurelyEqual A^i \times B^i$.
		\begin{itemize}
			\item If $C^i_j \in \mathcal{O}$ and $j < t$, then $i \in I$ and $i \in J$ by Condition~\ref{cond:sec-mult-c-i}.
			\begin{align*}
				&\Own (\mathbf{A}, \mathbf{B}, \mathbf{Q}) \land \bigwedge_{i, j} C^i_j \AlmostSurelyEqual A^i \times B^i + g^i_j(\mathbf{Q}, \mathbf{S}) \land \bigwedge_{i < j} S_{i, j} \AlmostSurelyEqual Q_{i, j} + A^i \times B^j + A^j \times B^i \land \bigwedge_{i \in I} A^i \AlmostSurelyEqual a^i \land \bigwedge_{j \in J} B^j \AlmostSurelyEqual b^j \\
				&\LilacEntails \bigwedge_{C^i_j \in \mathcal{O}, j < t} C^i_j \AlmostSurelyEqual a^i \times b^i + g^i_j(\mathbf{Q}, \mathbf{S})
			\end{align*}
			If $S_{k, i}$ is a free variable in $g^i_j(\mathbf{Q}, \mathbf{S})$ and $(k, i) \in K_{\mathbf{Q}}$, then by Condition~\ref{cond:sec-mult-c-q}, $k, i \in I$ and $k, i \in J$, and thus $S_{k, i} \AlmostSurelyEqual Q_{k, i} + a^k \times b^i + a^i \times b^k$.
			Hence, we obtain
			\[ \bigwedge_{C^i_j \in \mathcal{O}, j < t} C^i_j \AlmostSurelyEqual a^i \times b^i + h^i_j(\mathbf{Q}, \mathbf{S}) \]
			where $h^i_j(\mathbf{Q}, \mathbf{S})$ is defined as follows.
			\[ h^i_j(\mathbf{Q}, \mathbf{S}) \coloneqq \begin{cases}
				\displaystyle \sum_{k < j, (k, i) \in K_{\mathbf{Q}}} (Q_{k, i} + a^k \times b^i + a^i \times b^k) + \sum_{k < j, (k, i) \notin K_{\mathbf{Q}}} S_{k, i} & i \ge j \\
				\displaystyle \sum_{k < i, (k, i) \in K_{\mathbf{Q}}} (Q_{k, i} + a^k \times b^i + a^i \times b^k) + \sum_{k < i, (k, i) \notin K_{\mathbf{Q}}} S_{k, i} + \sum_{i < k \le j} Q_{i, k} & i < j
			\end{cases} \]
			Note that the free variables in $h^i_j(\mathbf{Q}, \mathbf{S})$ are in $\{ S_{i, j} \mid (i, j) \in K_{\mathbf{S}} \setminus K_{\mathbf{Q}} \} \cup \{ Q_{i, j} \mid (i, j) \in K_{\mathbf{Q}} \}$ if $C^i_j \in \mathcal{O}$ and $j < t$.
			\item If $C^i_t \in \mathcal{O}$ and there exists $j < t$ with $C^i_j \in \mathcal{O}$, then similarly to the previous case, we have the following equations.
			\begin{align*}
				\bigwedge_{C^i_t \in \mathcal{O}, (\exists j < t, C^i_j \in \mathcal{O})} C^i_t \AlmostSurelyEqual a^i \times b^i + h^i_t(\mathbf{Q}, \mathbf{S})
			\end{align*}
			The free variables in $h^i_t(\mathbf{Q}, \mathbf{S})$ are in $\{ S_{i, j} \mid (i, j) \in K_{\mathbf{S}} \setminus K_{\mathbf{Q}} \} \cup \{ Q_{i, j} \mid (i, j) \in K_{\mathbf{Q}} \}$.
			\item If $C^i_t \in \mathcal{O}$ and there exists no $j < t$ with $C^i_j \in \mathcal{O}$:
			\begin{align*}
				&\Own (\mathbf{A}, \mathbf{B}, \mathbf{Q}) \land \bigwedge_{i, j} C^i_j \AlmostSurelyEqual A^i \times B^i + g^i_j(\mathbf{Q}, \mathbf{S}) \land \bigwedge_{i < j} S_{i, j} \AlmostSurelyEqual Q_{i, j} + A^i \times B^j + A^j \times B^i \\
				&\LilacEntails \bigwedge_{C^i_t \in \mathcal{O}, \lnot (\exists j < t, C^i_j \in \mathcal{O})} C^i_t \AlmostSurelyEqual A^i \times B^i + g^i_t(\mathbf{Q}, \{ Q_{i, j} + A^i \times B^j + A^j \times B^i \}_{i < j})
			\end{align*}
			Note that for each such $C^i_t$, there exists $Q_{i', j'} \in K_{\mathbf{C}}$ that occurs in $g^i_t(\mathbf{Q}, \{ Q_{i, j} + A^i \times B^j + A^j \times B^i \}_{i < j})$, and moreover, the right-hand side of the above equation is a bijective expression in $Q_{i', j'}$.
		\end{itemize}
	\end{itemize}
	By \hyperlink{rule:Eq-Conj}{\textsc{Eq-Conj}}, we now have the following postcondition.
	\begin{align*}
		&\Conditioning_{(a, b) \gets (\mathcal{I}, \mathcal{J})} \Big( \Own (\mathbf{A}, \mathbf{B}) * \bigast_{i < j} Q_{i, j} \sim \Uniform{\mathbb{K}} * \bigast_{(i, j) \in K_{\mathbf{S}} \setminus K_{\mathbf{Q}}} S_{i, j} \AlmostSurelyEqual Q_{i, j} + A^i \times B^j + A^j \times B^i \\
		&\qquad\qquad * \big( \bigwedge_{P_{i, j} \in \mathcal{O}} P_{i, j} \AlmostSurelyEqual a^i \times b^j \land \bigwedge_{R_{i, j} \in \mathcal{O}} R_{i, j} \AlmostSurelyEqual Q_{i, j} + a^i \times b^j \land \bigwedge_{(i, j) \in K_{\mathbf{Q}}, S_{i, j} \in \mathcal{O}} S_{i, j} \AlmostSurelyEqual Q_{i, j} + a^i \times b^j + a^j \times b^i \\
		&\qquad\qquad \land \bigwedge_{C^i_j \in \mathcal{O}, j < t} C^i_j \AlmostSurelyEqual a^i \times b^i + h^i_j(\mathbf{Q}, \mathbf{S}) \land \bigwedge_{C^i_t \in \mathcal{O}, (\exists j < t, C^i_j \in \mathcal{O})} C^i_t \AlmostSurelyEqual a^i \times b^i + h^i_t(\mathbf{Q}, \mathbf{S}) \big)\\
		&\qquad\qquad * \bigast_{C^i_t \in \mathcal{O}, \lnot (\exists j < t, C^i_j \in \mathcal{O})} C^i_t \AlmostSurelyEqual A^i \times B^i + g^i_t(\mathbf{Q}, \{ Q_{i, j} + A^i \times B^j + A^j \times B^i \}_{i < j}) \Big)
	\end{align*}
	Now, we divide $\bigast_{i < j} Q_{i, j} \sim \Uniform{\mathbb{K}}$ into three parts (technically, four parts, but one part is discarded).
	\begin{align*}
		&\bigast_{i < j} Q_{i, j} \sim \Uniform{\mathbb{K}} \\
		&\LilacEntails \bigast_{(i, j) \in K_{\mathbf{Q}}} Q_{i, j} \sim \Uniform{\mathbb{K}} * \bigast_{(i, j) \in K_{\mathbf{S}} \setminus K_{\mathbf{Q}}} Q_{i, j} \sim \Uniform{\mathbb{K}} * \bigast_{(i, j) \in K_{\mathbf{C}}} Q_{i, j} \sim \Uniform{\mathbb{K}}
	\end{align*}
	Then, we apply Lemma~\ref{lem:generalised-bijective-substitution-uniform} and \hyperlink{rule:Transfer-Uniform}{\textsc{Transfer-Uniform}} to the second part and the third part, respectively.
	\begin{align*}
		&\Own (\mathbf{A}, \mathbf{B}) * \bigast_{(i, j) \in K_{\mathbf{S}} \setminus K_{\mathbf{Q}}} Q_{i, j} \sim \Uniform{\mathbb{K}} * \bigast_{(i, j) \in K_{\mathbf{S}} \setminus K_{\mathbf{Q}}} S_{i, j} \AlmostSurelyEqual Q_{i, j} + A^i \times B^j + A^j \times B^i \\
		&\LilacEntails \Own (\mathbf{A}, \mathbf{B}) * \bigast_{(i, j) \in K_{\mathbf{S}} \setminus K_{\mathbf{Q}}} S_{i, j} \sim \Uniform{\mathbb{K}} \\
		&\Own (\mathbf{A}, \mathbf{B}) * \bigast_{(i, j) \in K_{\mathbf{C}}} Q_{i, j} \sim \Uniform{\mathbb{K}} * \bigast_{C^i_t \in \mathcal{O}, \lnot (\exists j < t, C^i_j \in \mathcal{O})} C^i_t \AlmostSurelyEqual A^i \times B^i + g^i_t(\mathbf{Q}, \{ Q_{i, j} + A^i \times B^j + A^j \times B^i \}_{i < j}) \\
		&\LilacEntails \Own (\mathbf{A}, \mathbf{B}) * \bigast_{C^i_t \in \mathcal{O}, \lnot (\exists j < t, C^i_j \in \mathcal{O})} C^i_t \sim \Uniform{\mathbb{K}}
	\end{align*}
	The current postcondition is as follows.
	\begin{align*}
		&\Conditioning_{(a, b) \gets (\mathcal{I}, \mathcal{J})} \Big( \Own (\mathbf{A}, \mathbf{B}) * \bigast_{(i, j) \in K_{\mathbf{Q}}} Q_{i, j} \sim \Uniform{\mathbb{K}} * \bigast_{(i, j) \in K_{\mathbf{S}} \setminus K_{\mathbf{Q}}} S_{i, j} \sim \Uniform{\mathbb{K}} \\
		&\qquad\qquad * \big( \bigwedge_{P_{i, j} \in \mathcal{O}} P_{i, j} \AlmostSurelyEqual a^i \times b^j \land \bigwedge_{R_{i, j} \in \mathcal{O}} R_{i, j} \AlmostSurelyEqual Q_{i, j} + a^i \times b^j \land \bigwedge_{(i, j) \in K_{\mathbf{Q}}, S_{i, j} \in \mathcal{O}} S_{i, j} \AlmostSurelyEqual Q_{i, j} + a^i \times b^j + a^j \times b^i \\
		&\qquad\qquad \land \bigwedge_{C^i_j \in \mathcal{O}, j < t} C^i_j \AlmostSurelyEqual a^i \times b^i + h^i_j(\mathbf{Q}, \mathbf{S}) \land \bigwedge_{C^i_t \in \mathcal{O}, (\exists j < t, C^i_j \in \mathcal{O})} C^i_t \AlmostSurelyEqual a^i \times b^i + h^i_t(\mathbf{Q}, \mathbf{S}) \big)\\
		&\qquad\qquad  * \bigast_{C^i_t \in \mathcal{O}, \lnot (\exists j < t, C^i_j \in \mathcal{O})} C^i_t \sim \Uniform{\mathbb{K}} \Big)
	\end{align*}
	Here, the free variables of the right-hand sides of the remaining equations are in $\{ Q_{i, j} \mid (i, j) \in K_{\mathbf{Q}} \} \cup \{ S_{i, j} \mid (i, j) \in K_{\mathbf{S}} \setminus K_{\mathbf{Q}} \}$.
	By applying \hyperlink{rule:Transfer-Own}{\textsc{Transfer-Own}}, we obtain the following postcondition.
	\begin{align*}
		&\Conditioning_{(a, b) \gets (\mathcal{I}, \mathcal{J})} \Big( \Own (\mathbf{A}, \mathbf{B}) * \big( \bigwedge_{(i, j) \in K_{\mathbf{Q}}} \Own Q_{i, j} \land \bigwedge_{(i, j) \in K_{\mathbf{S}} \setminus K_{\mathbf{Q}}} \Own S_{i, j} \land \bigwedge_{P_{i, j} \in \mathcal{O}} \Own P_{i, j} \land \bigwedge_{R_{i, j} \in \mathcal{O}} \Own R_{i, j} \\
		&\qquad\qquad \land \bigwedge_{(i, j) \in K_{\mathbf{Q}}, S_{i, j} \in \mathcal{O}} \Own S_{i, j} \land \bigwedge_{C^i_j \in \mathcal{O}, j < t} \Own C^i_j \land \bigwedge_{C^i_t \in \mathcal{O}, (\exists j < t, C^i_j \in \mathcal{O})} \Own C^i_t \big) * \bigast_{C^i_t \in \mathcal{O}, \lnot (\exists j < t, C^i_j \in \mathcal{O})} C^i_t \sim \Uniform{\mathbb{K}} \Big)
	\end{align*}
	By definition of $K_{\mathbf{Q}}$ and $K_{\mathbf{S}}$, we finally have the following postcondition.
	\[ \Conditioning_{(a, b) \gets (\mathcal{I}, \mathcal{J})} (\Own (\mathbf{A} \setminus \mathcal{I}, \mathbf{B} \setminus \mathcal{J}) * \Own \mathcal{O}) \qedhere \]
\end{appendixproof}

\begin{corollary}
	\textsc{SecMult} is $t$-SNI.
	\qed
\end{corollary}

\section{Related Work}

\paragraph{Probabilistic separation logic}
Probabilistic separation logic~\cite{BarthePOPL2020} is extended in various ways to support more expressive reasoning about probabilistic programs.
Some works support reasoning about conditional independence, such as Lilac~\cite{LiPLDI2023}, DIBI~\cite{BaoLICS2021}, BaSL~\cite{HoPOPL2026}, and BlueBell~\cite{BaoPOPL2025}.
Other works interpret separating conjunction in different ways: cryptographic probabilistic separation logic~\cite{DalLagoCSF2024} interprets it as computational independence, while LINA~\cite{BaoPOPL2022} provides a separating conjunction for negative association.
\emph{Probabilistic concurrent outcome logic} (pcOL)~\cite{ZilbersteinPOPL2026} combines probabilistic separation logic with concurrent separation logic~\cite{BrookesTCS2007} to reason about probabilistic concurrent programs.
Some of the above logics have been applied to verify cryptographic algorithms~\cite{BarthePOPL2020,LiPLDI2023,DalLagoCSF2024,BaoPOPL2025}.
However, none of these logics has been used to verify $t$-noninterference of masked algorithms.

\paragraph{Formal verification of masked programs}
Formal verification of probing security has been actively studied in recent years~\cite{BartheCCS2016,BartheESORICS2019,BartheEUROCRYPT2015,MeunierIIEEETransSoftwareEng2023,OuahmaPROOFS2017,BenElOuahmaJCryptogrEng2019,ZhangCAV2018,GaoTACAS2019,GaoTOSEM2021,GaoTOSEM2024,EldibTOSEM2014,CoronACNS2018,BloemEUROCRYPT2018,KnichelASIACRYPT2020}.
\textsc{maskVerif} \cite{BartheESORICS2019,BartheEUROCRYPT2015} repeatedly substitute expressions in masked programs to prove noninterference, based on the fact that applying bijective functions $f$ to uniform random variables $X$ gives uniform random variables $f(X)$ as well.
Lilac has an inference rule (\hyperlink{rule:Unif-Bijection}{\textsc{Unif-Bijection}}) that captures this fact, which allows us to use the simulated their noninterference proof as a proof in Lilac.
Substitution-based approaches have also been used in~\cite{MeunierIIEEETransSoftwareEng2023}.
Another line of work uses inference-based approaches~\cite{OuahmaPROOFS2017,BenElOuahmaJCryptogrEng2019,ZhangCAV2018,GaoTACAS2019,GaoTOSEM2021}, which infer distribution types for each intermediate values in masked programs.

When proving noninterference properties of masked algorithms, one must in principle consider all possible sets of probed variables $\mathcal{O}$. This combinatorial explosion can be problematic, especially when pursuing algorithmic proof search~\cite{BartheEUROCRYPT2015}. However, in the setting of formal verification with an interactive theorem prover, this issue can be avoided by representing the set of probed variables symbolically. This is the approach adopted in this work.

\section{Conclusion and Future Work}
We have presented a method for formally verifying probing security using Lilac.
We have shown that the simulator-based definition of $(\mathcal{I}, \mathcal{O})$-NI is equivalent to a Hoare triple in Lilac that asserts a certain conditional independence.
We have also provided several proof rules that facilitate the verification of probing security and demonstrated their application to \textsc{Refresh} and \textsc{AddRepNoiseER}.

Our current results focus on the information-theoretic setting of probing security.
As future work, we want to extend our approach to the computational setting by adopting the ideas from~\cite{DalLagoCSF2024}.
Implementing our method in an interactive theorem prover is another important direction for future work.
We also want to explore applications of our method to the verification of other cryptographic algorithms, such as secure multi-party computation protocols, since Theorem~\ref{thm:noninterference-conditional-independence} would be applicable to other security properties defined in terms of simulators as well.

\bibliographystyle{IEEEtran}
\bibliography{probing-security,arxiv}

\checkpagelimit{16}

\end{document}